%% file: Main.tex
\documentclass[lettersize,journal]{IEEEtran}

\usepackage{textcomp}
\usepackage{transparent}
\usepackage{stfloats}
\usepackage{url}
\usepackage{verbatim}
\usepackage{graphicx}
\def\BibTeX{{\rm B\kern-.05em{\sc i\kern-.025em b}\kern-.08em
    T\kern-.1667em\lower.7ex\hbox{E}\kern-.125emX}}
\usepackage{balance}
\renewcommand{\transparent}[1]{}

\usepackage{epsfig} %% for loading postscript figures
\usepackage{graphicx} % for pdf, bitmapped graphics files
\usepackage{epsfig} % for postscript graphics files
\usepackage{amsmath} % assumes amsmath package installed
\usepackage{amssymb}  % assumes amsmath package installed
\usepackage{subfig}
\usepackage[ruled,longend]{algorithm2e}
\usepackage{multicol}
\usepackage{color}
\usepackage[colorlinks,linkcolor=blue]{hyperref}
\usepackage{threeparttable}
\usepackage[export]{adjustbox}
\usepackage{graphicx} 
\usepackage{caption,subcaption}
\usepackage{amsthm}
\usepackage{import}
\newtheorem{theorem}{Theorem}

\newtheorem{lemma}{Lemma}
\newtheorem{remark}{Remark}

\newtheorem{assumption}{Assumption}

\newcommand{\T}{\ensuremath{^\mathsf{T}}}
\usepackage{scalerel}

\newcommand\reallywidehat[1]{\arraycolsep=0pt\relax%
\begin{array}{c}
\stretchto{
  \scaleto{
    \scalerel*[\widthof{\ensuremath{#1}}]{\kern-.5pt\bigwedge\kern-.5pt}
    {\rule[-\textheight/2]{1ex}{\textheight}} %WIDTH-LIMITED BIG WEDGE
  }{\textheight} % 
}{0.5ex}\\           % THIS SQUEEZES THE WEDGE TO 0.5ex HEIGHT
#1\\                 % THIS STACKS THE WEDGE ATOP THE ARGUMENT
\rule{-1ex}{0ex}
\end{array}
}

\begin{document}

\title{A Reduced Order Iterative Linear Quadratic Regulator (ILQR) Technique for the Optimal Control of Nonlinear Partial Differential Equations}
% \title{A Reduced Order Iterative Linear Quadratic Regulator (iLQR) for the Control of Partial Differential equations}
% \title{A Reduced Order Model-Based Reinforcement Learning Approach to the Control of Nonlinear Partial Differential Equations}
% \thispagestyle{empty}
% \pagestyle{empty}

% \author{Anonymous}
\author{Aayushman Sharma, Suman Chakravorty
\thanks{The authors are with the Department of Aerospace Engineering, Texas A\&M University, College Station, TX 77843 USA (e-mail: aayushmansharma@tamu.edu; schakrav@tamu.edu)
%\{\tt aayushmansharma, schakrav\}@tamu.edu
}}

\IEEEaftertitletext{\vspace{-2\baselineskip}}

\maketitle

\begin{abstract}
% \textcolor{red}{My remaining comments in red still remain to be addressed. Please make sure there are no typos etc.}
In this paper, we introduce a reduced order model-based reinforcement learning (MBRL) approach, utilizing the Iterative Linear Quadratic Regulator (ILQR) algorithm %Iterative Linear Quadratic Regulator (RO-ILQR) approach 
for the optimal control of nonlinear partial differential equations (PDEs). The approach proposes a novel modification of the ILQR technique: it uses the Method of Snapshots to identify a reduced order Linear Time Varying (LTV) approximation of the nonlinear PDE dynamics around a current estimate of the optimal trajectory, utilizes the identified LTV model to solve a time-varying reduced order LQR problem to obtain an improved estimate of the optimal trajectory along with a new reduced basis, and iterates till convergence. The convergence behavior of the reduced order approach is analyzed and the algorithm is shown to converge to a limit set that is dependent on the truncation error in the reduction. The proposed approach is tested on the viscous Burger's equation and two phase-field models for microstructure evolution in materials, and the results show that there is a significant reduction in the computational burden over the standard ILQR approach, without significantly sacrificing performance.

\end{abstract} 

% \vspace{-0.2cm}
% \input{Introduction_revised.tex}
\input{Introduction.tex} 
\input{pde_pod.tex}

% \input{Algorithm.tex}
\input{Algorithm2.tex}
\input{Preliminaries.tex}
% \input{PerturbationModel.tex} 
% \input{improvements.tex}

% \input{pde_introduction.tex}
% \input{Simulations.tex}
\input{Simulations_new.tex}

\input{Conclusions.tex}
% \input{Acknowledgements}
% \vspace{-0.5cm}
\input{Appendix.tex}

\input{Bibliography.tex}
% \newpage

\end{document}

%% file: Introduction.tex
\section{Introduction}

The discipline corresponding to the modeling and control of infinite-dimensional partial differential equations (PDEs), has been one of significant interest with a rich history, spanning applications such as fluid flows \cite{lellouche}, protein folding and enzyme kinematics \cite{proteinfold}, microstructure dynamics \cite{Boe2002}\cite{Rui2011} \emph{etc}. To address such systems, numerical methods are generally utilized, wherein the state-space is discretized through finite-element methods, employing control strategies for the resulting complex, very high-dimensional differential equation systems. The need for reducing the complexity of these systems has driven the development of Model Order Reduction (MOR) techniques. Recent advances in the field have seen the introduction of data-driven approaches in MOR, with methods such as DMD \cite{dmd1}\cite{dmd2}, neural networks \cite{nn1}\cite{nn2}\cite{autoenc}, nonlinear auto-regressive algorithms\cite{nara1}, and sparse regression\cite{sparse1}. Among these approaches, the Proper Orthogonal Decomposition (POD) approach is especially prevelant, using simulation-derived snapshots to compute the reduced basis in order to approximate the original nonlinear PDE \cite{sirovich1987turbulence}.

% The modeling and control of infinite-dimensional partial differential equations has long been a topic of avid interest, extending into various different fields ranging from fluid flows \cite{lellouche}, protein folding and enzyme kinematics \cite{proteinfold}, to microstructure dynamics \cite{Boe2002}\cite{Rui2011}. Such problems are usually solved by employing numerical methods, discretizing the state space through finite-element methods and controlling the resulting very high-dimensional system of differential equations. The need for dimensionality reduction in such systems paved the way for Model Order Reduction (MOR) techniques, accomplished in a data-driven manner, using techniques such as DMD \cite{dmd1}\cite{dmd2}, neural networks \cite{nn1}\cite{nn2}\cite{autoenc}, nonlinear auto-regressive algorithms\cite{nara1}, and sparse regression\cite{sparse1}. Perhaps the most widely used approach is the so-called Proper Orthogonal Decomposition (POD) Method of Snapshots that finds a reduced basis for the approximation of the nonlinear PDE using snapshots from a simulation of the same\cite{sirovich1987turbulence}. 

Recent studies \cite{wang2021search}\cite{yu2019decoupled} have introduced a data-based approach for the optimal control of nonlinear systems, employing a successive identification of linear time-varying (LTV) models in conjunction with the Iterative Linear Quadratic Regulator (ILQR) technique to compute a globally optimum local feedback control policy. The ILQR algorithm \cite{li2004iterative} belongs to the class of sequential quadratic programming (SQP) methods, widely used for solving problems in nonlinear programming in the optimal control context, with theoretical guarantees for convergence to the global minimum of the optimal control problem under relatively mild conditions \cite{wang2021search}. Its versatility has been successfully demonstrated in various relatively low-dimensional robotic applications \cite{ilqr2}\cite{ilqr3}. However, the Curse of Dimensionality presents significant obstacles in controlling complex nonlinear systems. The existing extensions fail to address the limitations of the ILQR framework in complex, high-dimensional settings, especially for PDEs, which require the identification of very high-dimensional LTV models. This leads to the critical backward pass Ricatti recursion in the algorithm being intractable, thus rendering the algorithm in its current form infeasible. Consequently, the optimal control of nonlinear PDEs largely remains an unresolved issue.

% In previous work \cite{wang2021search}\cite{yu2019decoupled}, the authors proposed a data-based algorithm for optimal nonlinear control, which identifies successive linear time-varying (LTV) models combined with the so-called Iterative Linear Quadratic Regulator (ILQR) approach to yield a globally optimum local feedback law. The ILQR algorithm \cite{li2004iterative} is a sequential quadratic programming (SQP) approach, typically used to solve nonlinear programming problems, for the solution of optimal control problems, and under relatively mild conditions can be shown to converge to the global minimum of the optimal control problem \cite{wang2021search}. The ILQR is a versatile approach and has been successfully applied to a wide range of relatively low-dimensional robotic applications \cite{ilqr2}\cite{ilqr3}. 
% Unfortunately, none of these extensions address the intractability of the ILQR algorithm when dealing with high-dimensional systems, in particular PDEs that require the identification of very high-dimensional LTV models and leads to intractability of the critical backward pass Ricatti recursion in the algorithm. Hence, the optimal control of nonlinear PDEs still remains a major challenge.

\begin{figure}
    \centering
\begin{multicols}{4}
      \subfloat[Initial State]{\includegraphics[width=\linewidth]{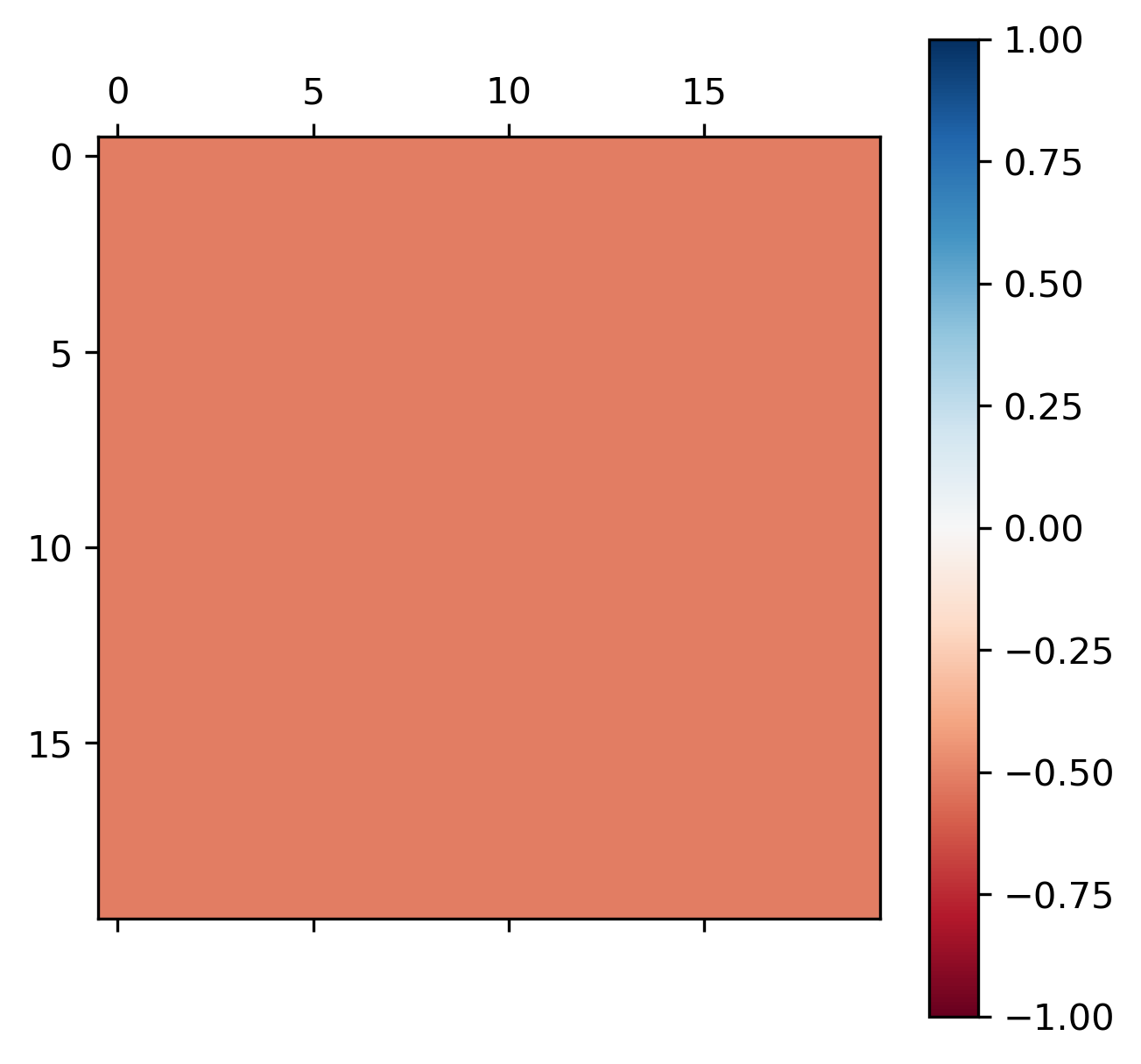}}
      \subfloat[Int. I]{\includegraphics[width=\linewidth]{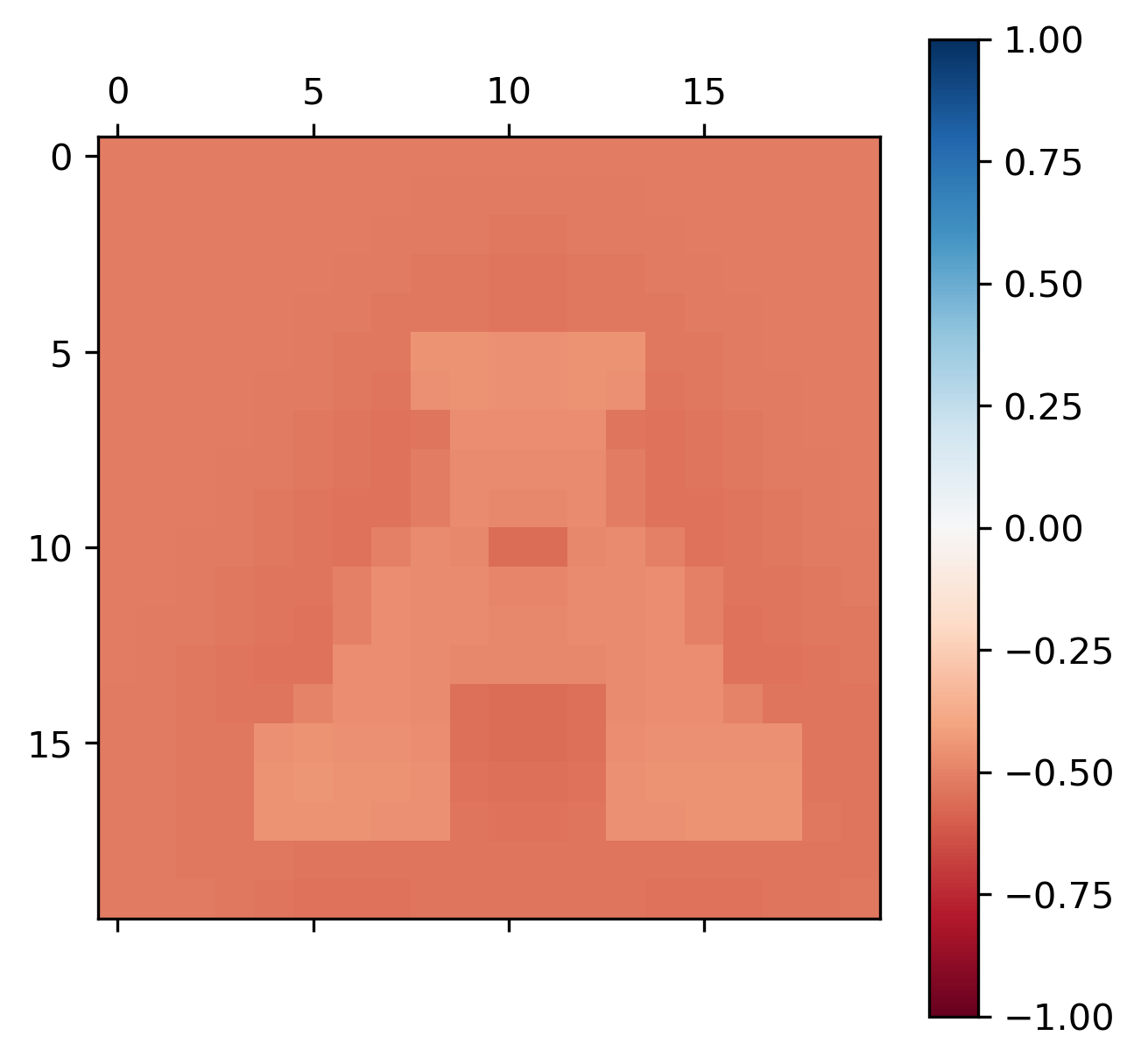}}
      \subfloat[Int. II]{\includegraphics[width=\linewidth]{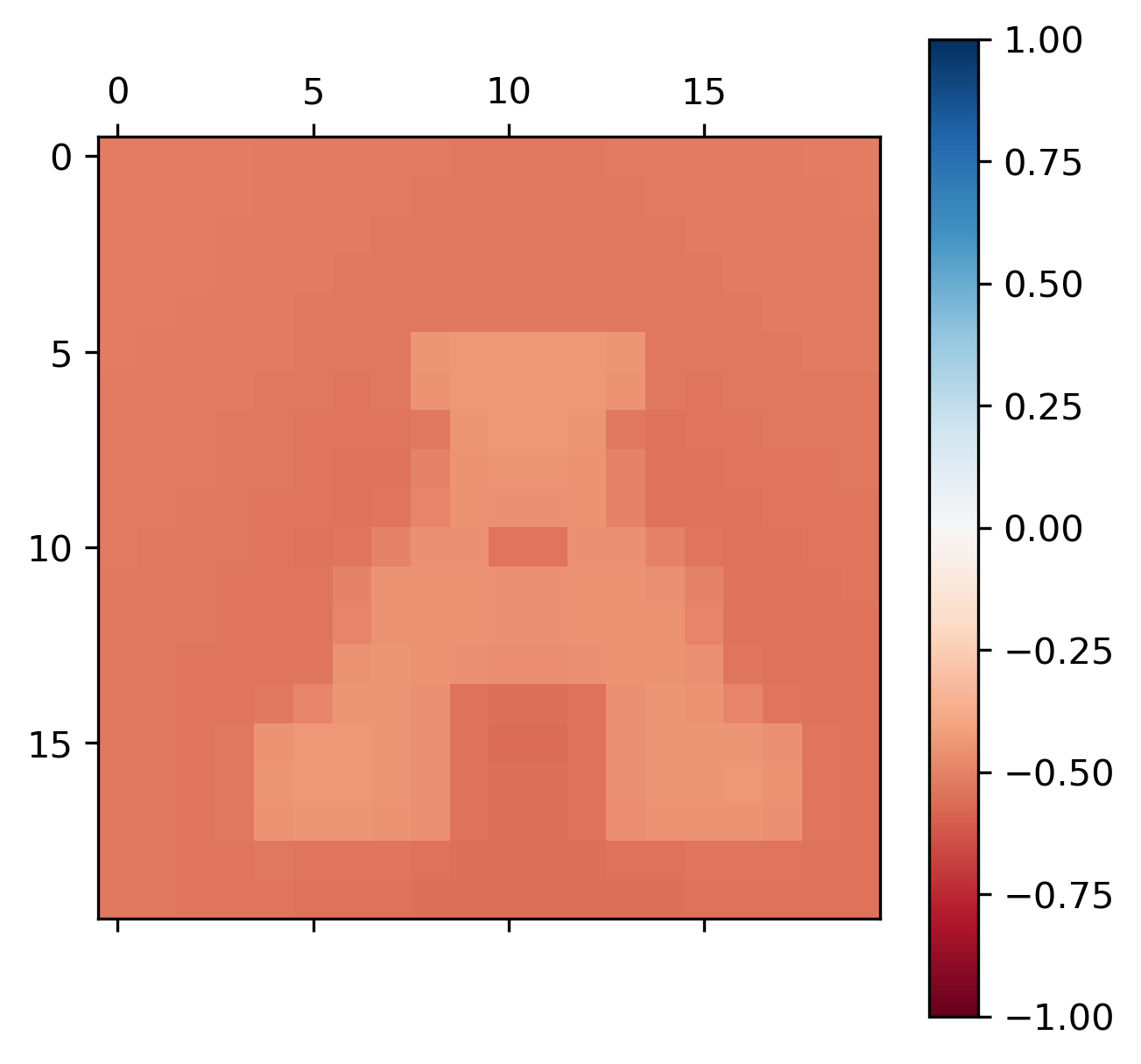}}
      \subfloat[Final State]{\includegraphics[width=\linewidth]{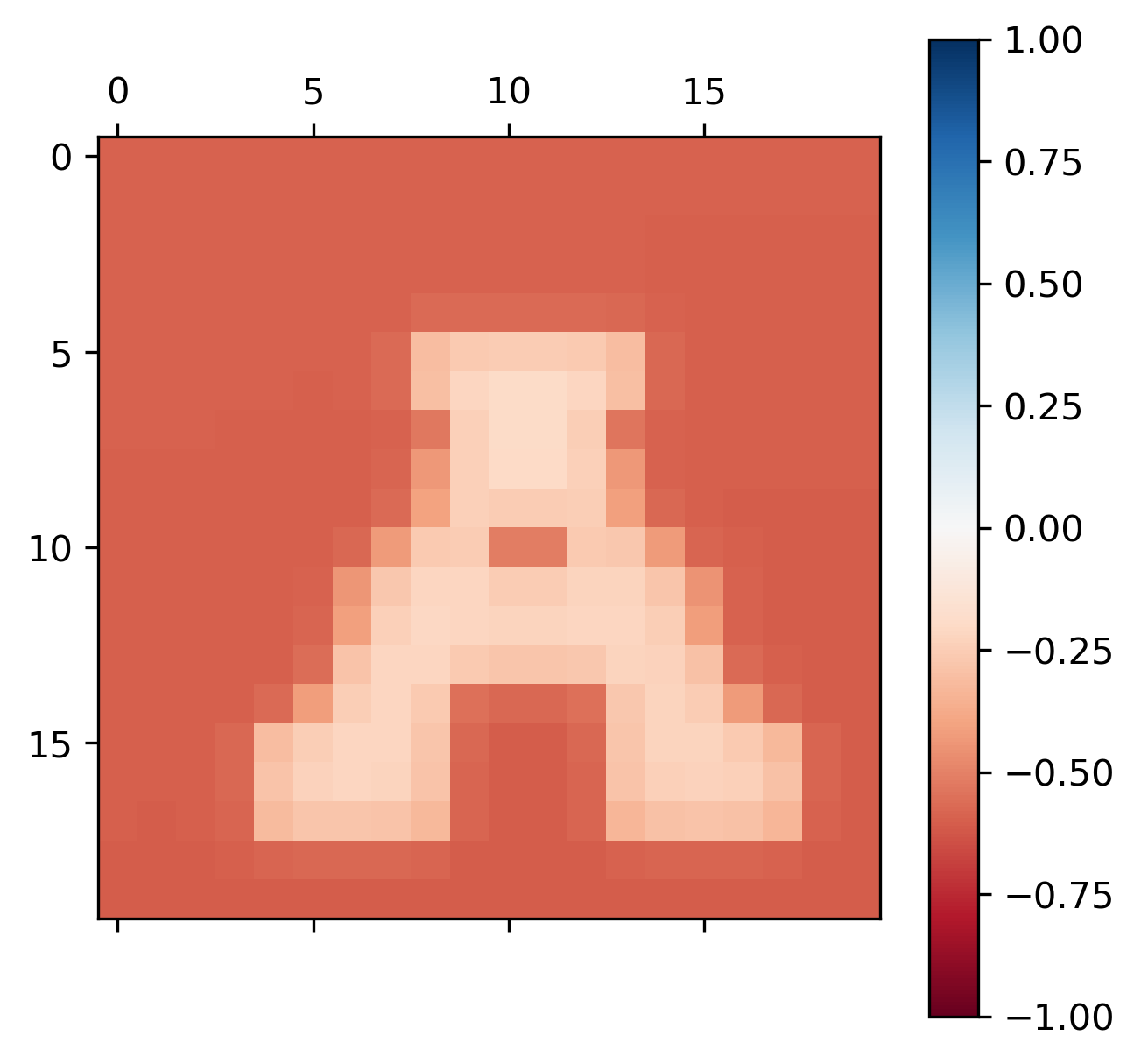}}

      %\subfloat[Gradient-descent based D2C Convergence]{\includegraphics[width=1\linewidth]{images/episodic_cost_OL_training.png}}
      
\end{multicols}
\begin{multicols}{4}

      \subfloat[1st mode]{\includegraphics[width=\linewidth]{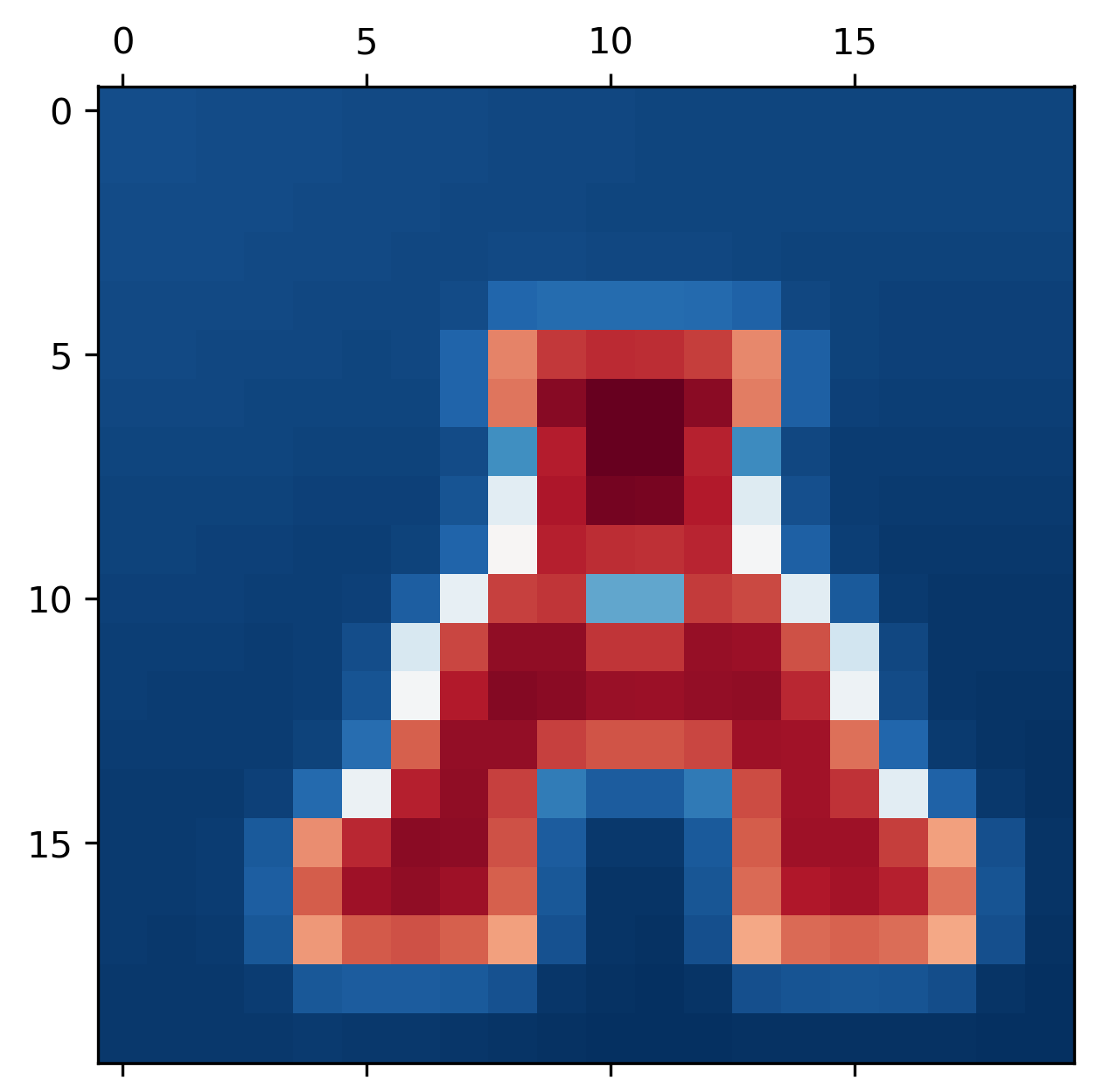}}
      \subfloat[2nd mode]{\includegraphics[width=\linewidth]{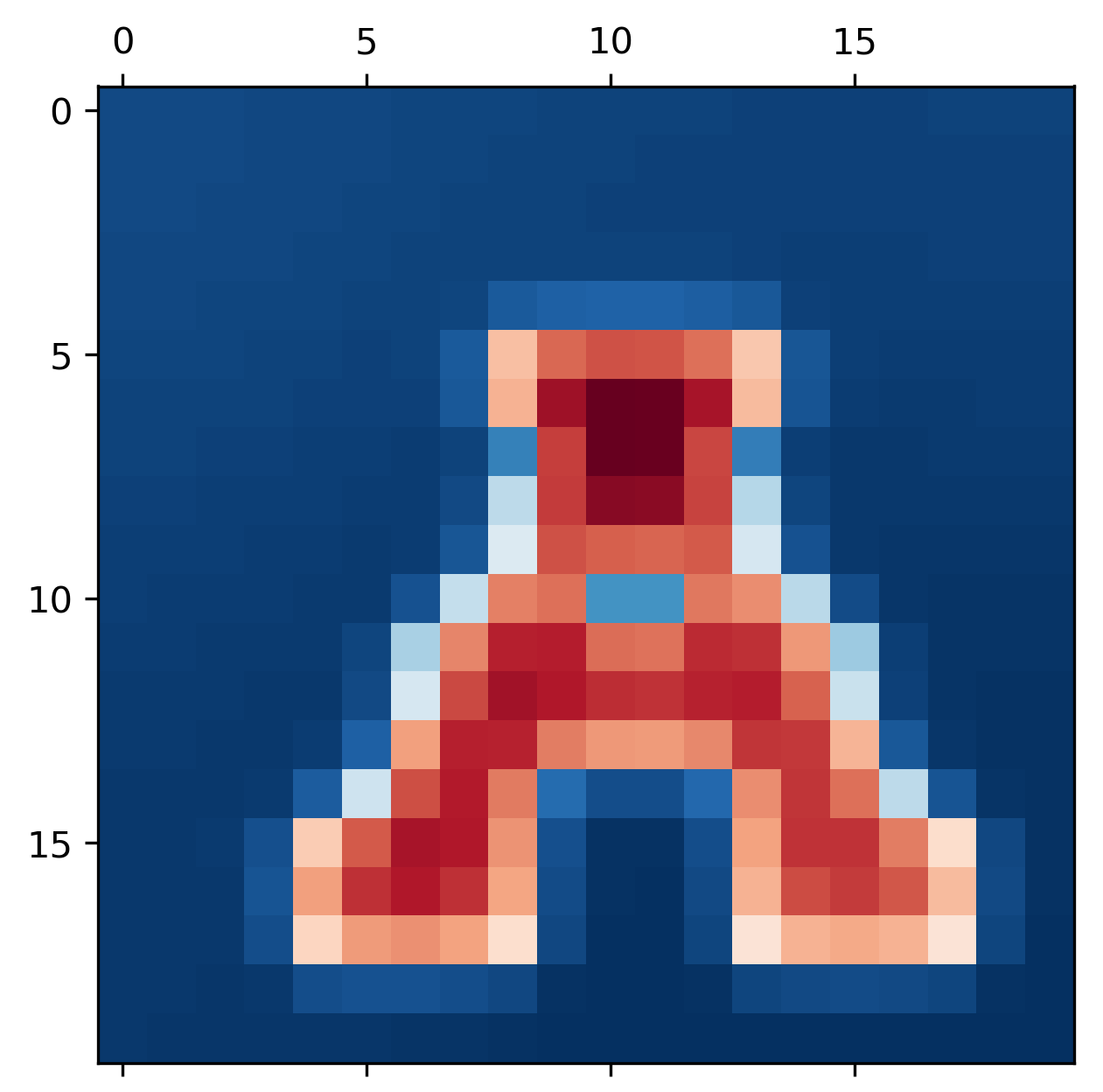}}
      \subfloat[3rd mode]{\includegraphics[width=\linewidth]{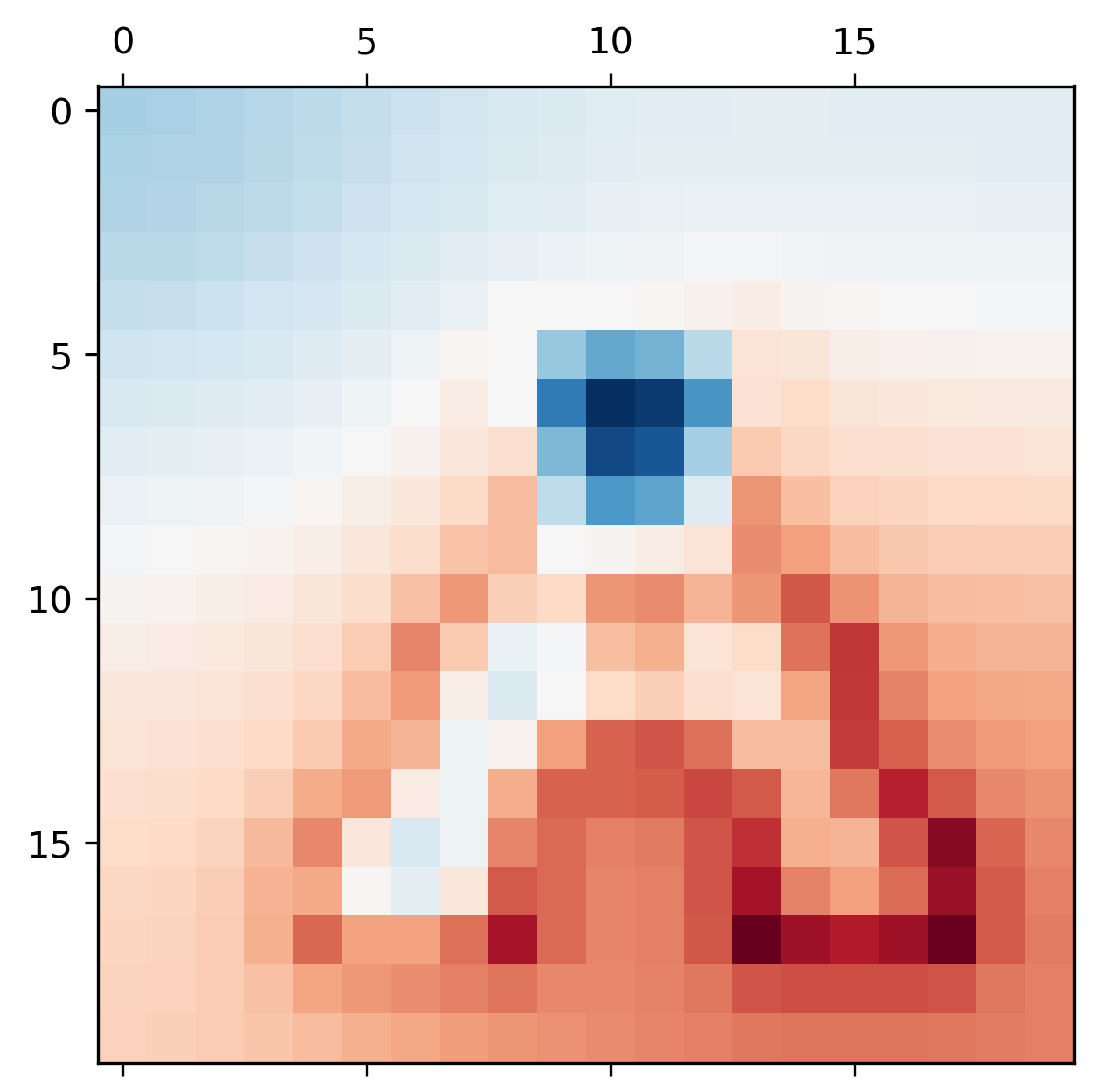}}
      \subfloat[4th mode]{\includegraphics[width=\linewidth]{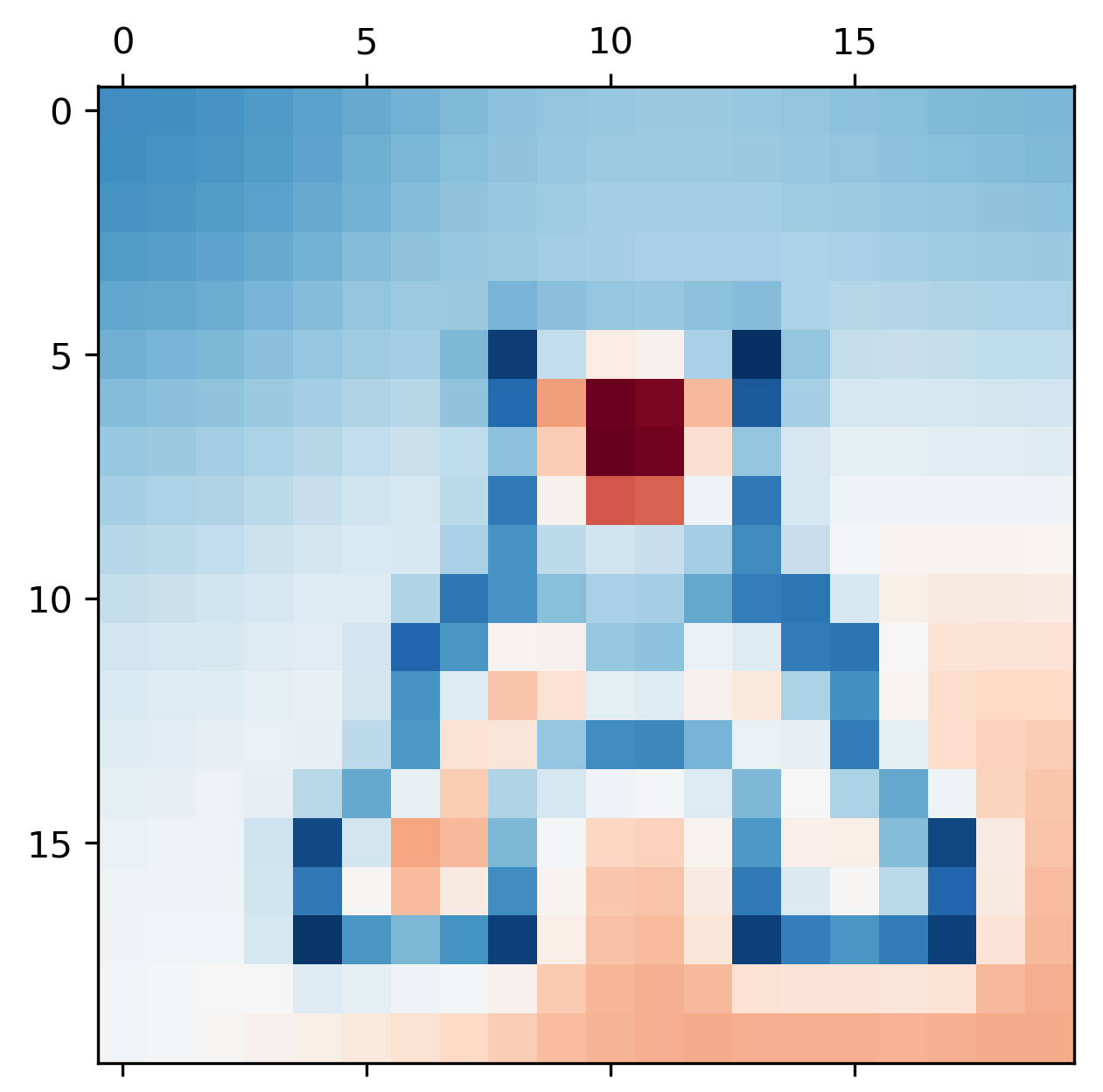}}
      %\subfloat[Gradient-descent based D2C Convergence]{\includegraphics[width=1\linewidth]{images/episodic_cost_OL_training.png}}
      
\end{multicols}

\begin{multicols}{4}
   
      \subfloat[Initial State]{\includegraphics[width=\linewidth]{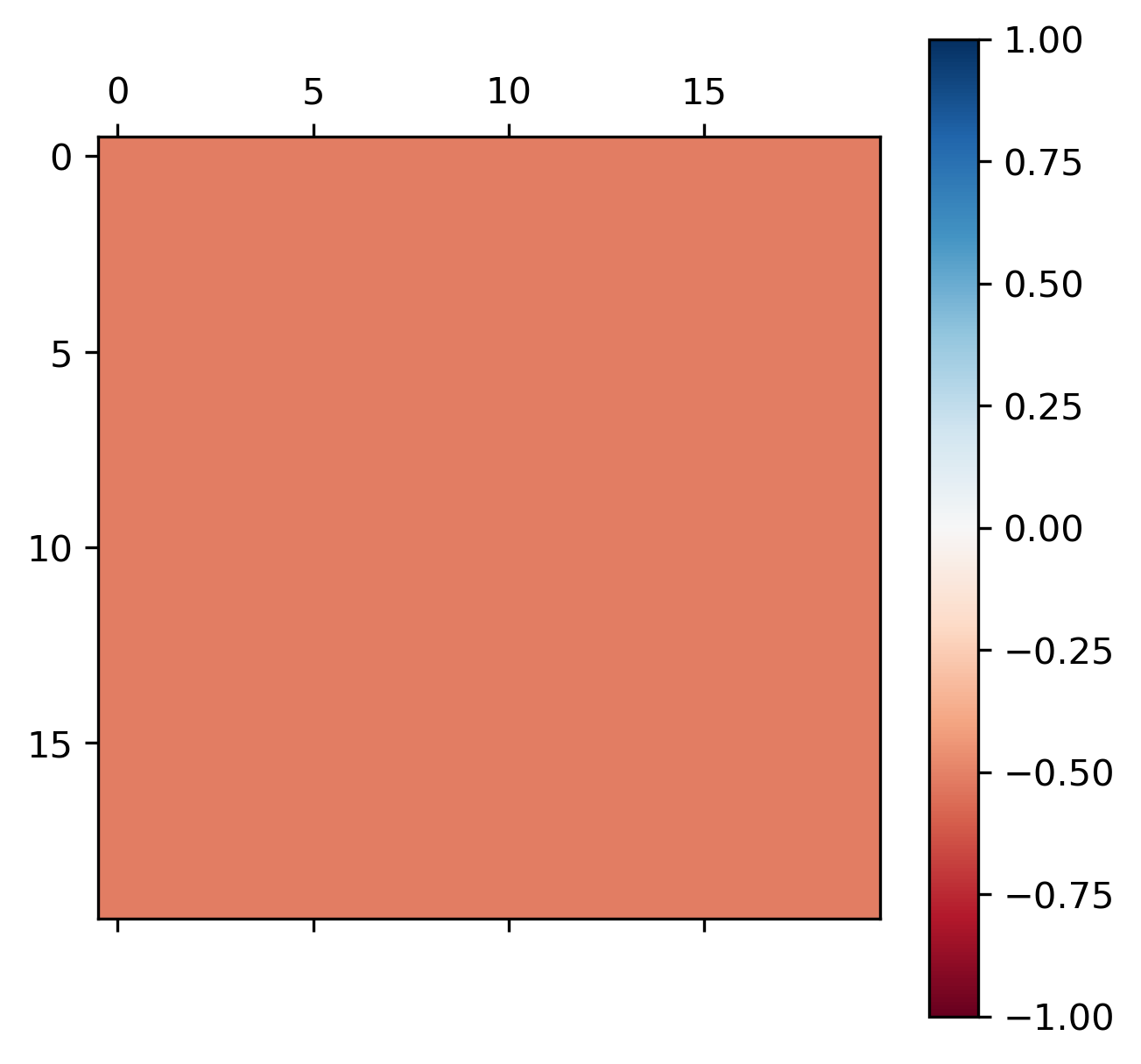}}
      \subfloat[Int. I]{\includegraphics[width=\linewidth]{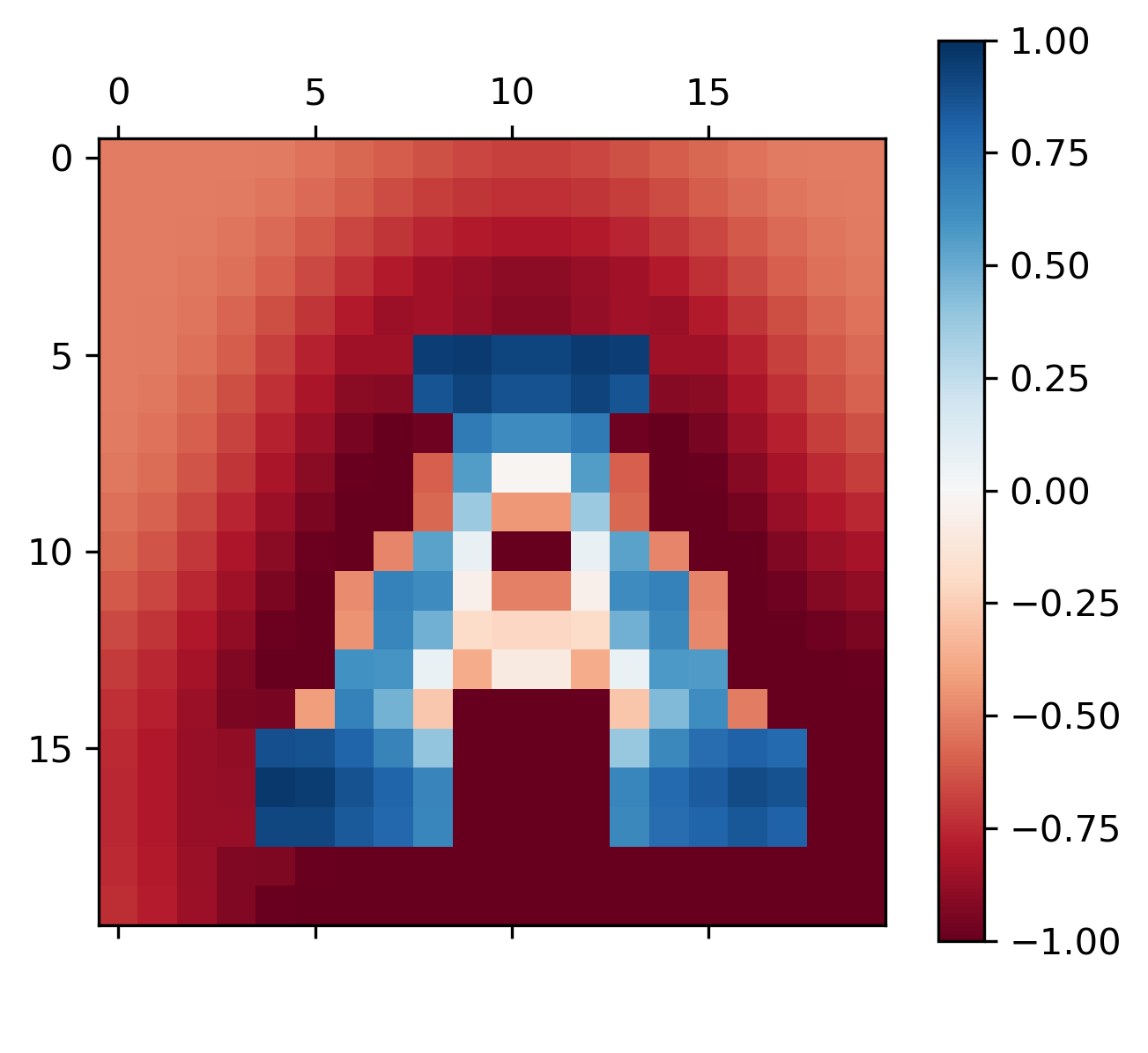}}
      \subfloat[Int.  II]{\includegraphics[width=\linewidth]{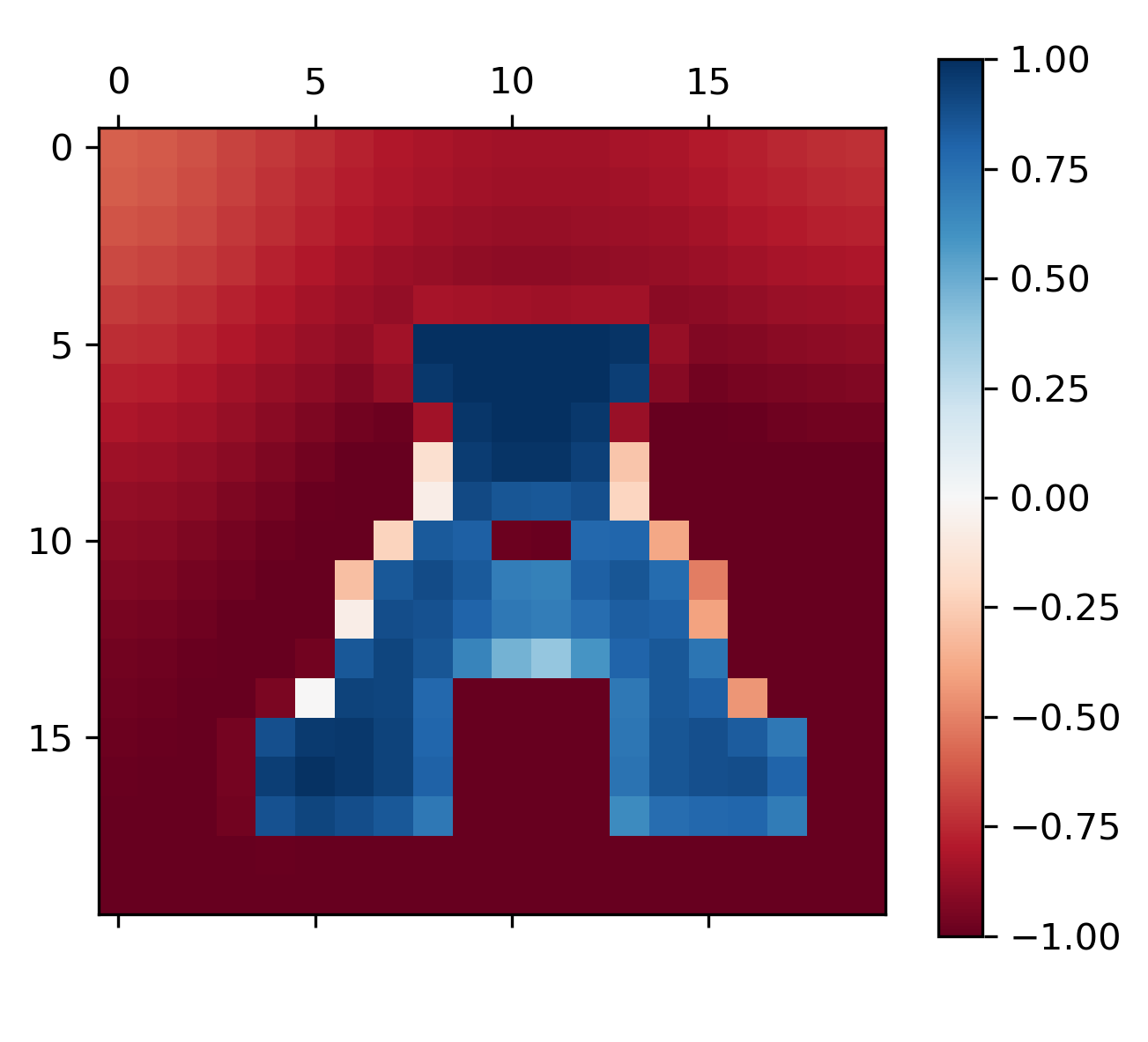}}
      \subfloat[Final State]{\includegraphics[width=\linewidth]{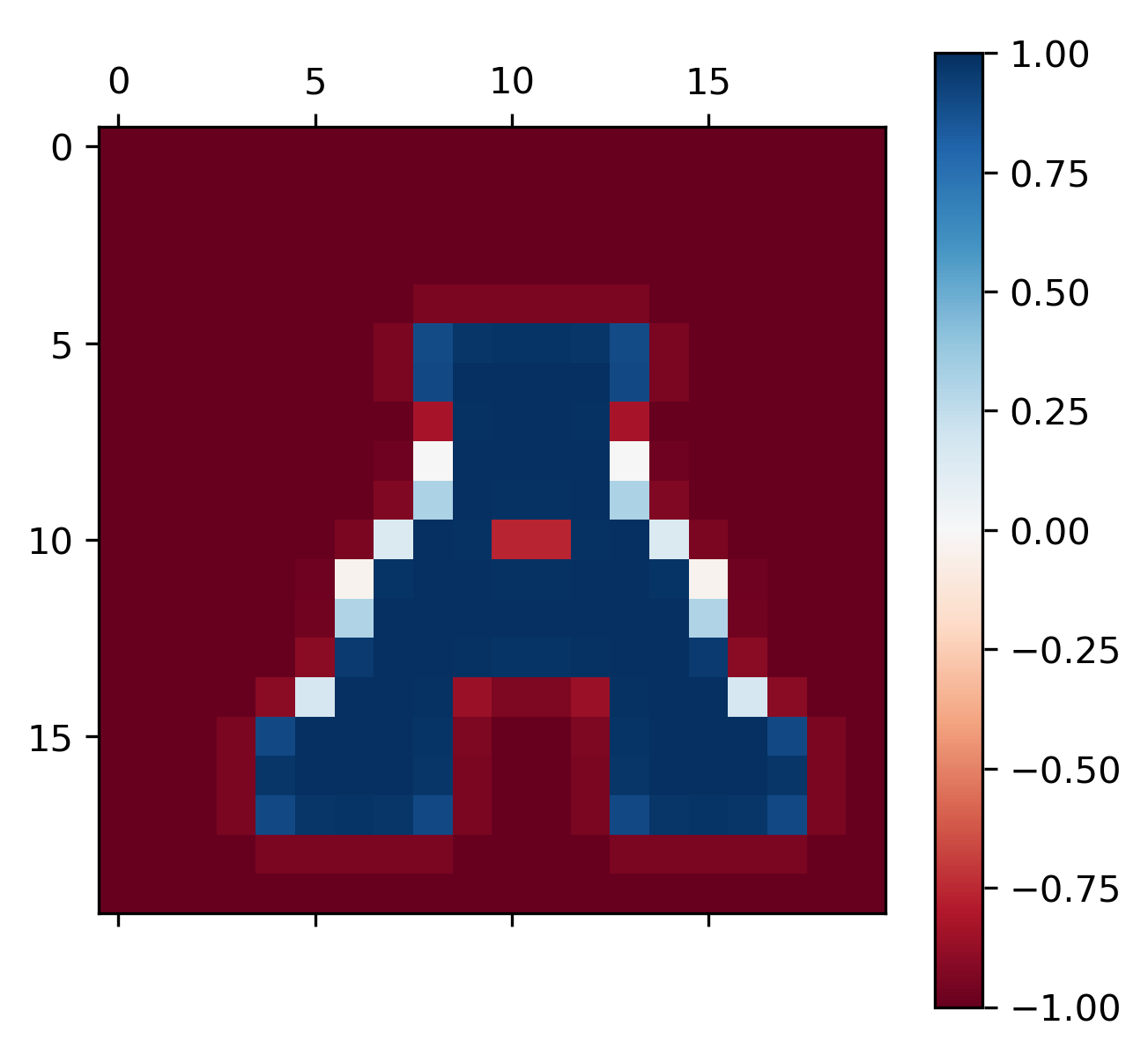}}

      %\subfloat[Gradient-descent based D2C Convergence]{\includegraphics[width=1\linewidth]{images/episodic_cost_OL_training.png}}
      
\end{multicols}

\begin{multicols}{4}

      \subfloat[1st mode]{\includegraphics[width=\linewidth]{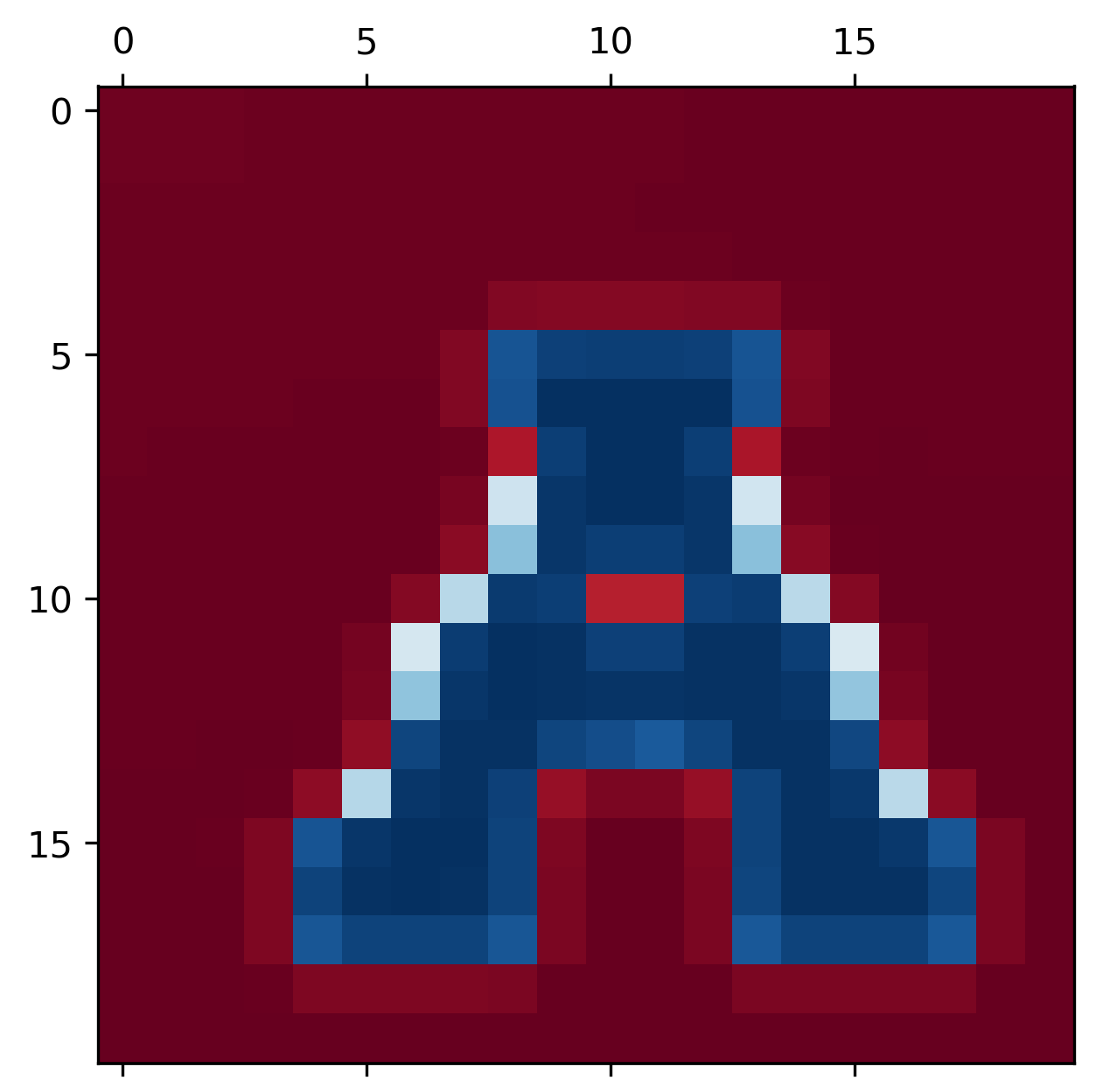}}
      \subfloat[2nd mode]{\includegraphics[width=\linewidth]{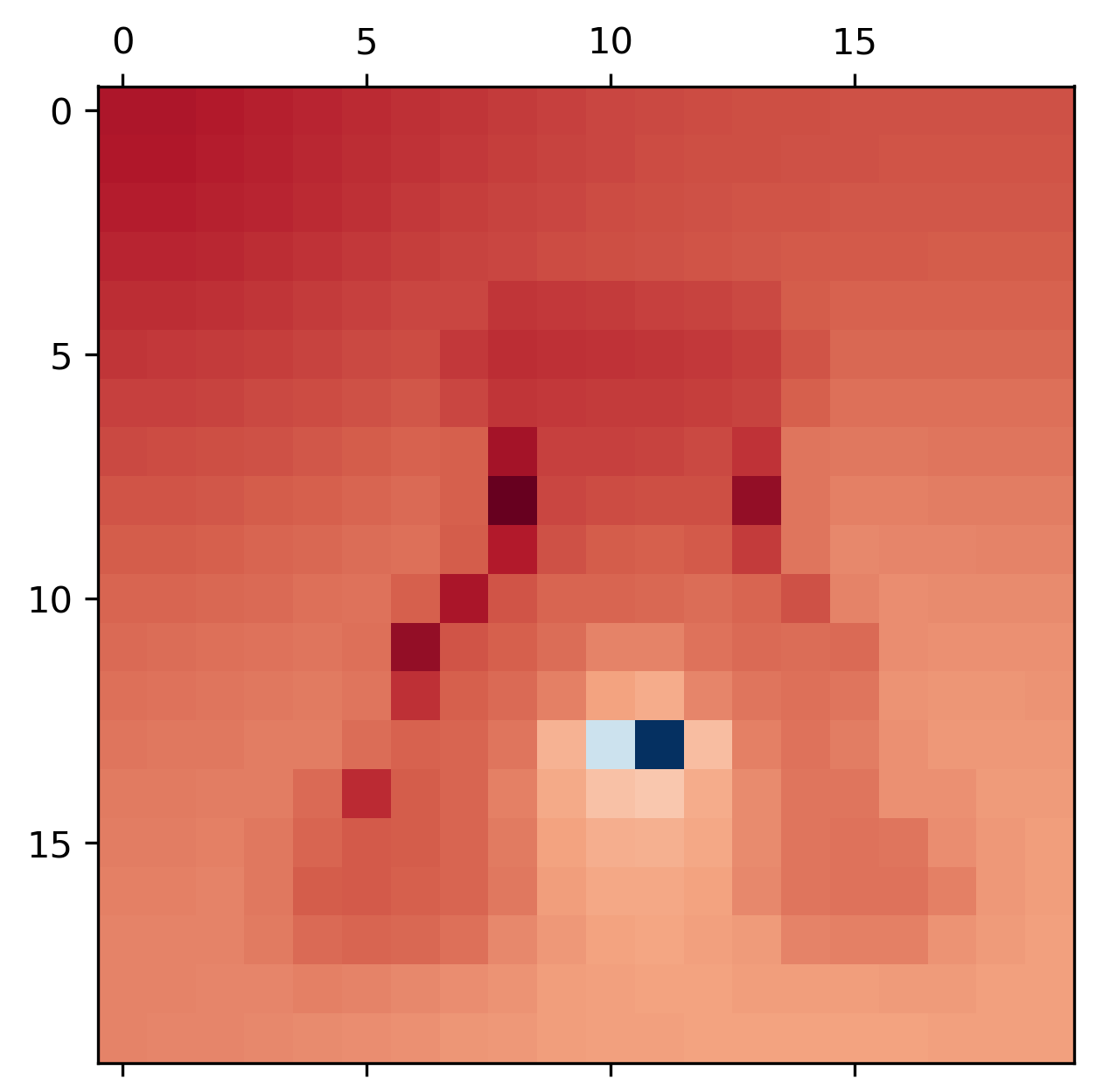}}
      \subfloat[3rd mode]{\includegraphics[width=\linewidth]{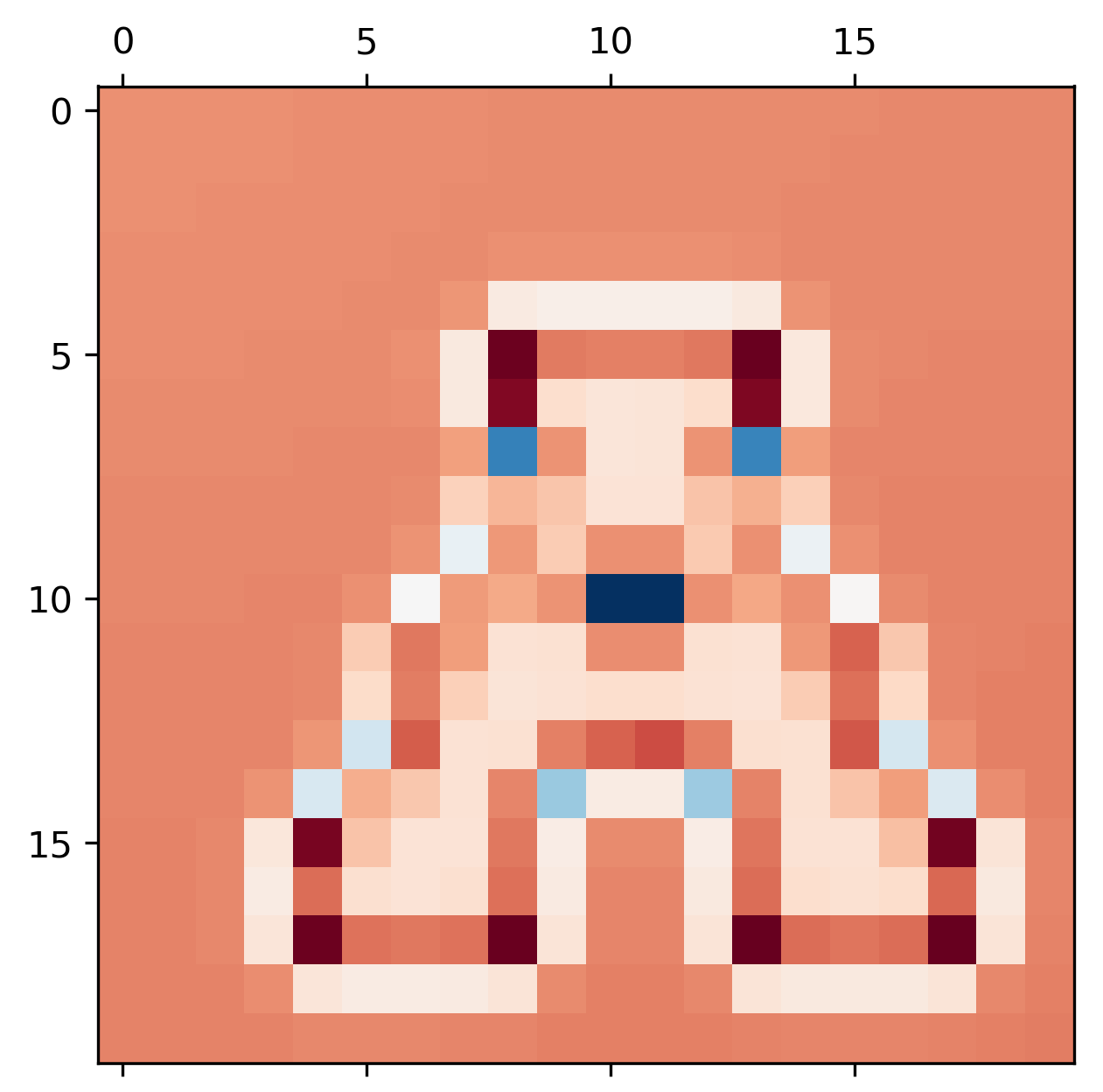}}
      \subfloat[4th mode]{\includegraphics[width=\linewidth]{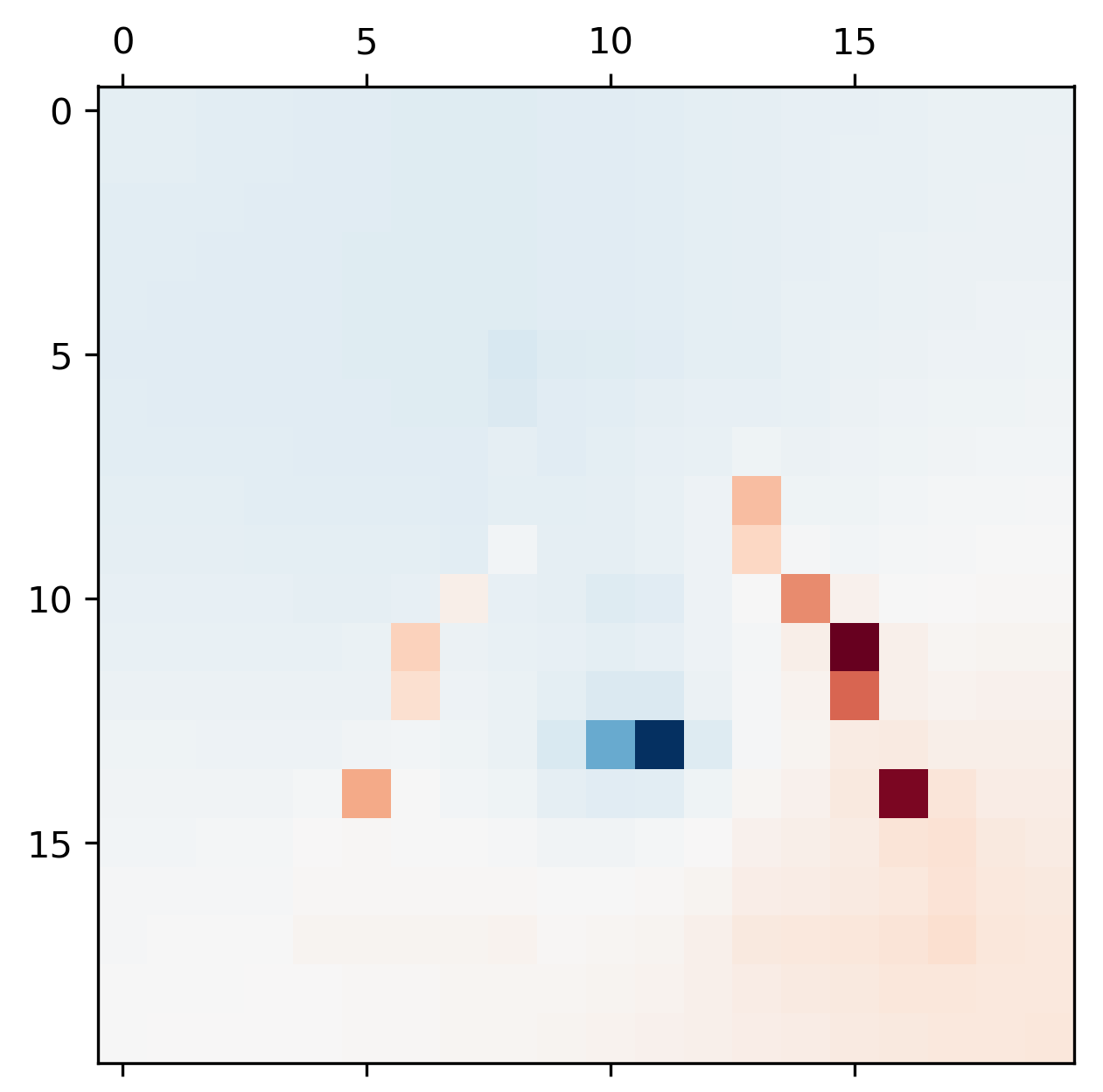}}
      %\subfloat[Gradient-descent based D2C Convergence]{\includegraphics[width=1\linewidth]{images/episodic_cost_OL_training.png}}
      
\end{multicols}

    %\caption{Comparing control cost at different time-discretizations}
    \caption{\small The initial trajectory (a)-(d) differs significantly from the final trajectory optimized by ILQR (i)-(l). Thus, the set of basis eigenfunctions for the initial trajectory (e)-(h) will only be valid locally, and differ significantly from the reduced order subspace that the optimal trajectory lies on (m)-(p). Hence, it is key to update the reduced-order basis with each successive iteration of the algorithm.}
    \label{fig:up_top}

\end{figure}

% \textcolor{blue}{Add papers on reduced order MBRL. Typically, the algorithms are split between a 'reduce-then-optimize', or an 'optimize-then-reduce' approach. These approaches suffer from model bias arising from differing basis as the optimization proceeds, leading to a need for shifting the RO basis as we iterate locally.}
The field of PDE control has a rich history, for instance, see Kunisch \cite{KunischOCPDE}, %\cite{KunischSLPDE},
 Volkwein \cite{volkweinPhaseField}, %\cite{VolkweinParabolicADR},
%Krstic \cite{smyshlyaev2010adaptive}, Alla \cite{allaSLPDE} 
and Falcone \cite{allafalconeadaptive}. %This foundational work was further built upon for application to nonlinear PDE control \cite{Breiten}\cite{bennerOCPDE}\cite{VolkweinFeedbackController}\cite{volkweinNonlinearRDE}\cite{krstic2008backstepping}.
% Further work by Luo et al \cite{cybernetics} suggests using a Karhunen-Loève decomposition followed by singular perturbation for model order reduction, and a neural network (NN)-based approach to solve the HJB-like equation in order to compute an approximate optimal control for nonlinear one-dimensional PDEs. 
More recently, there is a shift towards the application of reinforcement learning techniques for the control of such distributed parameter systems \cite{tnnls_pde_2}, in conjunction with model reduction approaches such as the Karhunen-Loève decomposition \cite{cybernetics}. These techniques are broadly classified into a ``reduce-then-control" (RTC) approach \cite{rtc_1,rtc_2}, and a ``control-then-reduce" (CTR) approach \cite{ctr_1}. However, in general, the CTR approach is infeasible, due to the intractability of finding the control for the full-order models in very high-dimensional problems (Curse of Dimensionality).
% Recent work from Kramer et al \cite{Kramer-ILQR} have suggested using a linear time invariant (LTI) model reduction technique called Balanced Truncation \cite{tacBT}
%  to obtain a reduced basis on which to project the nonlinear PDEs to obtain a reduced order nonlinear model, the optimal control for which is then accomplished using the ILQR algorithm.  However, in general, the reduced basis modes can be quite different around the optimal trajectory when compared to an initial trajectory, and an approximation using the latter’s reduced basis can lead to unacceptable performance \cite{ravindran2002}\cite{kunisch2008}. %(also see Fig. 1). 
%  Therefore, the reduced basis needs to be adaptively changed as the optimization proceeds \cite{kunisch2008}\cite{pod_volkwein}. Moreover, for nonlinear PDEs, the model reduction still results in a lower order nonlinear model, whose optimal control is itself challenging, and in general, intractable.

Recently, Kramer et al \cite{Kramer-ILQR} advocated for utilizing Balanced Truncation \cite{tacBT}, a linear time-invariant (LTI) MOR approach, for computing the reduced basis upon which the nonlinear PDE is projected to find the corresponding reduced order nonlinear system model, and further designing the optimal control for the same with the ILQR algorithm. However, in general, the corresponding reduced basis modes for the initial trajectory differ significantly as compared to the optimal trajectory. Thus, approximating the optimal control policy using the reduced order basis for the former results in unacceptable performance \cite{ravindran2002}\cite{kunisch2008}.
Consequently, there is a need to adaptively change the reduced order basis as the optimization process advances. Additionally, for nonlinear PDEs, even after model reduction, the resulting lower-order nonlinear model remains challenging and often intractable for optimal control. 
\textit{In this regard, the simple yet fundamental observation we make is that since the Proper Orthogonal Decomposition (POD) basis is local, an LTV model suffices to approximate the local behavior of the nonlinear PDE rather than a reduced nonlinear model.} 

Thus, this paper introduces a reduced-order model based RL approach, which modifies the existing Iterative Linear Quadratic Regulator (ILQR) technique to an "iterative-reduce-then-control" (IRTC) method for the optimal control of nonlinear partial differential equations. The ILQR algorithm is implemented in a data-based/ MBRL fashion, wherein we use rollouts by querying the `black-box' computational model. The system dynamics are iteratively linearized, and the cost function is quadratized around the current iterate of the optimal trajectory.
%It involves iteratively linearizing the system dynamics and quadratizing the cost function around the current estimate of the optimal trajectory. 
A time-varying LQR problem is then solved to generate an improved control sequence, and the process iterates till convergence. In this study, the POD Method of snapshots \cite{sirovich1987turbulence} is employed to derive a reduced-order LTV approximation of the nonlinear PDE in the vicinity of the current trajectory, followed by solving a reduced-order LQR problem, consequently generating the improved trajectory and an updated POD basis, and iterated till convergence, resulting in a drastic reduction of the computational burden of the ILQR method when extended to PDEs. The proposed approach is evaluated on the viscous Burger's equation, along with two phase-field models for reaction-diffusion equations related to microstructure evolution in multiphase materials, and demonstrates convergence to the global optimum that would be achieved without employing model reduction. This paper builds on recent findings presented in a conference paper \cite{roilqr}: it introduces a convergence analysis for the quality of the solution provided by the reduced order approach and presents more extensive computational results.

The rest of the paper is organized as follows: Section \ref{sec2} introduces the nonlinear PDE optimal control problem. We discuss the issues related to the control of high-dimensional PDEs, and provide suggestions to navigate around them in Section \ref{sec3}. Section \ref{sec4} briefly introduces the iLQR algorithm, followed by the proposed approach. In Section \ref{analysis}, we show a mathematical analysis for the convergence guarantees of the proposed approach. In Section \ref{sec5}, we 
demonstrate the application of the proposed approach through custom-defined test problems of varying dimensionality, along with repeatability, followed by the benchmarking of the approach with the standard ILQR.

\vspace{-0.15cm}

%% file: pde_pod.tex
% \vspace{-0.1cm}
\section{Problem Formulation}
\label{sec2}
Consider the discrete-time nonlinear dynamical system:
\begin{equation}
    x_{t+1}=f(x_t, u_t),
\label{eq:dyn}
\end{equation}
where $x_t\in \mathbb{R}^{n_x}$ and $u_t\in \mathbb{R}^{n_u}$ correspond to the state and control vectors at time $t$. The optimal control problem is to find the optimal control policy $\pi^0=\{\pi^0_0, \pi^0_1 ... \pi^0_{T-1}\}$, that minimizes the cumulative cost:
\begin{align}
    & \min_\pi J^\pi(x)=\sum_{t = 0}^{T-1} c_t(x_{t}, u_{t}) + c_T(x_T), \label{eq:FO-OCP}\tag{FO-OCP} \\
    &\text{Subject to:}\  x_{t+1}=f(x_t, u_t),
\end{align}
given some $x_0 = x$, and where $u_t=\pi_t^0(x_t)$, $c_t(\cdot)$ is the instantaneous cost function and $c_T(\cdot)$ is the terminal cost. We assume the incremental cost to be quadratic in control, such that $c_t(x_t,u_t)=l_t(x_t)+\frac{1}{2}u_t^T R u_t$.\\
In this paper, we consider dynamical systems governed by partial differential equations.
Our goal in this paper is to provide a feedback solution to such problems.

\textit{Remark:} We consider the system available to us as sufficiently finely discretized to be able to accurately represent the dynamics of the infinite-dimensional PDE system. All our claims regarding the optimality of the feedback solution are with respect to this high, albeit finite, dimensional problem.

\section{The Control of Nonlinear PDEs and the Issues Inherent}\label{sec3}
Despite there being a large number of algorithms developed for controlling dynamical systems, the control of infinite-dimensional PDEs still remains a challenge. Modeling techniques usually discretize the PDEs spatially and temporally, in order to solve them. In order for the discretization to reasonably model the dynamics of the system, we need a fine mesh, which results in a very high degree-of-freedom system, easily in the thousands. Scaling the methods for these problems leads to the so called Curse of Dimensionality \cite{bellman1957}, which makes working on such systems largely intractable. Thus we need ways to reduce the dimensionality of the model in a way that can still capture the dynamics of the system.\\
%\textcolor{red}{Give an example of a PDE and its discretization, and how the DOFs become very large, perhaps the A-C or Burger's equation.}\\
As an example, let us consider the Allen-Cahn Equation, as used in our work in the subsequent sections.
\begin{align}
    \frac{\partial \phi}{\partial t}&= -M(\frac{\partial F}{\partial\phi}-\gamma\nabla^{2}\phi),\\
    F(\phi;T,h) &= \phi^{4}+T\phi^{2}+h\phi.
\end{align}
% We spatially discretize the order-parameter $\phi(x,t)$ as a 2D grid, where the order parameter for a cell at position $(i,j)$ is represented by $\phi^{(i,j)}$, and a grid resolution of $\Delta x$. Similarly, we also do a temporal discretization. Thus, we have a 2nd order Central Difference equation as:
% \begin{align*}
%     \phi^{(i,j)}_{t+1} =\phi^{(i,j)}_t-M\Delta t \{4(\phi^{(i,j)}_t)^3+2T_t\phi^{(i,j)}_t+h_t-\\
%     \gamma (\frac{\phi^{(i,j)}_t+\phi^{(i+1,j)}_t  +\phi^{(i-1,j)}_t+\phi^{(i,j+1)}_t+\phi^{(i,j-1)}_t-4\phi^{(i,j)}_t}{(\Delta x)^2})\},
% \end{align*}
% \begin{equation*}
%   \phi^{(i,j)}_{t+1}  = f(\phi^{(i,j)}_t, T_t, h_t).
% \end{equation*}

% Thus, the state vector $x_t=\{\phi^{(i,j)}_t\}$ has the dimension $n_x=\frac{1}{(\Delta x)^2}$. To ensure that continuity holds, we need $\Delta x$ to be sufficiently small. 
%  But, we observe that $\Delta x \leq 0.02 \implies n_x \geq 2500$. Thus, the equation has very high DoF, even for relatively small-scale estimates.

For a realistic capture of the equation's behavior, we need to sufficiently discretize the equation. But, we observe that a discretization of $\Delta x \leq 0.02 \implies$ the dimensionality $n_x \geq 2500$. Thus, the equation has very high DoF, even for relatively small-scale estimates.

This issue of very high dimensionality is usually dealt by employing model order reduction techniques, which can largely be classified as simplified physics based approaches and projection based approach (\textit{e.g.} Proper orthogonal decomposition, balancing methods, nonlinear manifold methods). Here, we consider the proper orthogonal decomposition (POD) based approach to model reduction, which is one of the most prevalent approach in model reduction.  
\subsection{The POD Method of Snapshots}
\label{snapshotpod}
This section briefly reviews the application of the standard POD approach for dealing with high degree-of-freedom systems, particularly in the case of spatially discretized PDEs. In general, the POD approach computes the low rank approximation of the original system by employing the singular-value decomposition (SVD) on a general field variable $u(x,t)$, as $u(x,t)=\sum_{k=1}^{\infty}\alpha_k(t)\phi_k(x).$  

For high DOF systems, solving the resulting eigenvalue problem becomes computationally intractable owing to the denseness and high dimensionality of the spatially-discretized state matrix.  However, with access to the system’s dynamics in the form of discrete ``snapshot” state vectors, the ``Method of Snapshots” \cite{sirovich1987turbulence} offers a feasible approach to compute a low-order approximation of the system for a specified trajectory.

Given the initial state vector $\bar{x}_0$, and control inputs $\{u_t\}_{t=0}^{T-1}$, we propagate the system forward in time to get trajectory $\{x_t\}_{t=0}^T$. Consequently, we get the matrix $\bar{X}=\begin{bmatrix}x_0 &x_1 &\hdots &x_{T-1} &x_{T} \end{bmatrix}$ of dimension $(N\times T)$, where $T \ll N$.

Let the singular value decomposition of $\bar{X}=U\Sigma V^T$.
Unlike the direct POD approach, the method of snapshots operates on the assumption of the corresponding left- and right-singular vectors being related, and computes the eigen-decomposition of the smaller matrix $\bar{X}^T \bar{X}$ (of dimension $T\times T$), as opposed to the significantly larger $\bar{X}\bar{X}^T$. Thus, $\bar{X}^T \bar{X}=V\Sigma V^T$.
We can then retrieve the corresponding POD bases by $U=XV\Sigma^{-1/2}.$

Further, we consider the first $p$ modes that capture $99.999\%$ of the relative energy, defined as $E_{rel}=(\sum_{i=1}^{p}\lambda_i)/(\sum_{j=1}^{T}\lambda_j)\geq 0.99999$,
where $\lambda_1\geq...\geq\lambda_T$ are the non-zero singular values. With this criterion, we construct the reduced order basis matrix $\Phi$ by concatenating the singular vectors corresponding to the singular values satisfying the energy criteria $\Phi = \begin{bmatrix}\phi_{1} &\phi_{2} &\hdots &\phi_{p}
  \end{bmatrix}$.

Thus, the system can be compactly described by projecting the system onto the reduced-order basis as $\alpha = \Phi^T \bar{X}$.

\vspace{-0.3cm}

\subsection{Reduced Order Problem Formulation}
Let us assume the system dynamics to be as given in Eq.~\ref{eq:dyn}.

% \begin{equation}
%     x_{t+1}=f(x_t,u_t).
% \label{eq:dyn}
% \end{equation}

Given a trajectory, we can use the snapshot POD to get the reduced order dynamics as $x_t\approx \Phi \alpha_t$,
% \begin{equation*}
%     x_t\approx \Phi \alpha_t, 
% \end{equation*}
where recall that $\Phi$ denotes the POD bases obtained from the Method of Snapshots and $dim(\alpha_t) \ll dim(x_t)$.
We can now write the reduced order dynamics as 
\begin{equation}
    \alpha_{t+1}=\Phi^T f(\Phi \alpha_t,u_t).
    \label{rodyn}
\end{equation}

Thus, the  reduced order optimal control problem changes to
\begin{align}
    & \min_\pi J^\pi(\alpha)=\sum_{t = 0}^{T-1} \bar{c}_t(\alpha_{t}, u_{t}) + \bar{c}_T(\alpha_T) | \alpha_0 = \alpha, \label{eq:RO-OCP}\tag{RO-OCP}\\
    & \text{Subject to:}\  \alpha_{t+1}=\Phi^T f(\Phi \alpha_t,u_t), \nonumber
\end{align}
where $\bar{c}_t(\alpha_t)=c_t(\Phi \alpha_t)$.

\subsection{Control of the Reduced Order Model}
Projecting the system onto a reduced-order basis can drastically reduce the size of the state variables, but the reduced order system is still nonlinear and the nonlinear projections involved in forming the reduced order model can be very expensive. Techniques like SINDY \cite{SINDY} %\cite{pysindy},
and POD-DEIM \cite{nguyen2020pod} \textit{etc.} are able to give an efficient estimate of the nonlinear reduced-order model, but the identified system is nonetheless strongly dependent on the POD basis. However, the POD basis is obtained from a particular trajectory of the system, and thus, we can expect that the validity of the nonlinear reduced model is restricted to a local region around the trajectory used to generate the POD basis. \\
In lieu, in order to accurately capture the local dynamics of a nonlinear model around a trajectory, we can use a local Linear Time-Varying (LTV) model which is computationally far cheaper. 
% \textcolor{red}{where are the plots below? Are you even proofreading the paper? what about the references you mention below?} 
% The LTV model can be seen to accurately estimate the dynamics of a nonlinear model for small perturbations around a given trajectory, as observed in Fig.~\ref{fig:ltv_vs_gt}.
Since large perturbations can lead to a significant change in the reduced order POD basis, invalidating most nonlinear reduced order models obtained, an LTV model instead suffices for describing the behavior of a system locally around a trajectory.\\

% \begin{figure}[!htp]
% \vspace{-0.8cm}
% \begin{multicols}{2}
%     %\hspace{1.4cm} 
%     \subfloat[]{\includegraphics[width=1\linewidth]{images/roilqr plots/ltv_vs_gt_5pc-1.png}}    
%       \subfloat[]{\includegraphics[width=1\linewidth]{images/roilqr plots/ltv_vs_gt_10pc-1.png}}
    
%     % \includegraphics[width=\linewidth]{images/roilqr plots/ltv_vs_gt_10pc.pdf}  

% \end{multicols}
    
%     \caption{\small Comparing the trajectory generated by the identified LTV system with the ground-truth from the nonlinear dynamics of the Allen-Cahn PDE, with perturbations at (a)5\% and (b)10\% of max control input.}
%     \label{fig:ltv_vs_gt}
% \end{figure}
\vspace{-0.3cm}
Given an LTV estimate around a local trajectory, the Iterative Linear Quadratic Regulator (ILQR) algorithm \cite{li2004iterative} can be applied to find the optimal control for such problems. The ILQR algorithm can be adapted directly to control PDEs, modeled as a high DoF system of coupled ODEs but this results in very large LTV systems that are computationally intractable for fine discretizations \cite{wang2021search}.  
Thus, the question arises whether we can find a way to incorporate model reduction into the ILQR algorithm. The following section proposes a method capable of doing so.

%% file: Algorithm2.tex
% \twocolumn
\section{The Reduced Order MBRL Algorithm}%ILQR Algorithm}
\label{sec4}
% \textcolor{red}{This Section should be before the analysis section.}
In this section, we briefly introduce the ILQR algorithm followed by the modifications proposed for the reduced-order approach that is highly efficient for the control of PDEs.

\subsection{Iterative Linear Quadratic Regulator (ILQR)}
The ILQR Algorithm \cite{li2004iterative} iteratively solves the nonlinear optimal control problem posed in Section~\ref{sec2} using the following steps:

Given a nominal control sequence $\{u_t\}_{t=0}^{T-1}$ and initial state vector $x_0$, the state is propagated in time in accordance with the dynamics (\textit{Forward Pass}). Now, we find the corresponding local LTV system around the trajectory using input-output perturbation and a Linear Least Squares method. After getting the local LTV parameters about the system, the ILQR algorithm computes a local optimal control by solving the discrete time Riccati Equation (\textit{Backward Pass}). Now, given the gains from the backward pass, we can update the nominal control. This sequence is iterated till convergence.

\begin{remark}
The primary drawback of the ILQR scheme outlined above for PDE control is that the LTV system around a trajectory of the PDE is very high dimensional since it is $O(n_x + n_u)$, and thus, estimating the LTV system requires a very large amount of data which is intractable. Furthermore, given such a large LTV system, the backward pass is intractable as well since it is $O(n_x^3)$. Thus, in the next section, we propose a reduced order approach to iLQR that finds a reduced order LTV system using the POD method of snapshots and recursively updates the reduced order bases as the trajectory changes.
\end{remark}
\vspace{-0.15cm}
\subsection{The Reduced Order ILQR Algorithm}

\subsubsection{Forward Pass}
Given a nominal control sequence $\{u_t\}_{t=0}^{T-1}$ and initial state vector $x_0$, the state is propagated in time in accordance with the full-order nonlinear dynamics to generate the nominal trajectory $(\bar{x}_t,\bar{u}_t)$. Here, we incorporate the snapshot-POD method to get the reduced-order basis and the corresponding projections of the trajectory: $x_t \approx \Phi \alpha_t$.
% \begin{equation}
%     x_t \approx \Phi \alpha_t
% \end{equation}
Thus, we have the resulting trajectory $(\bar{\alpha}_t,\bar{u}_t)$ and the corresponding basis $\Phi$.

\input{Reduced_sys_id}

Thus, the system identification problem reduces from $O(n_x+n_u)$ to $O(n_\alpha+n_u)$, where $n_\alpha \ll n_x$. 
%\textcolor{red}{give examples of this reduction, from the examples in the paper.}\\
This can be seen with the PDE systems studied in this paper. For example, the Allen-Cahn equation, when discretized, has a 2500-dimensional state with 4 control inputs whereas the reduced order system lies almost entirely in a 3-dimensional subspace (also see Fig. 5).

\textit{Performance Comparison with Full Order LTV:}
The reduced order LTV system entails some loss in information of the dynamics due to projection on a reduced order basis. But from control perturbation experiments, we observe that the reduced-order LTV system performs almost as well as the full-order LTV system as seen in Fig. \ref{ltv_podvsfo}, despite a significantly cheaper computation cost.  
% \textcolor{red}{Why not compare with the full nonlinear simulation as well?}
%\vspace{-0.5cm}
\begin{figure}[!htpb]
\begin{multicols}{2}
    %\hspace{1.4cm} 
    \subfloat[]{\includegraphics[width=1.1\linewidth]{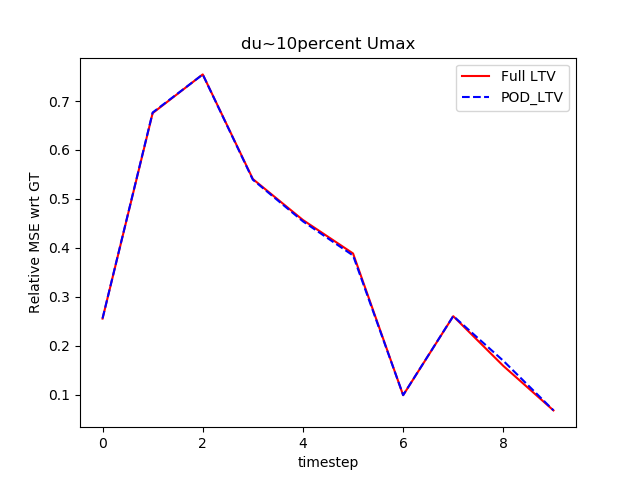}}    
      \subfloat[]{\includegraphics[width=1.1\linewidth]{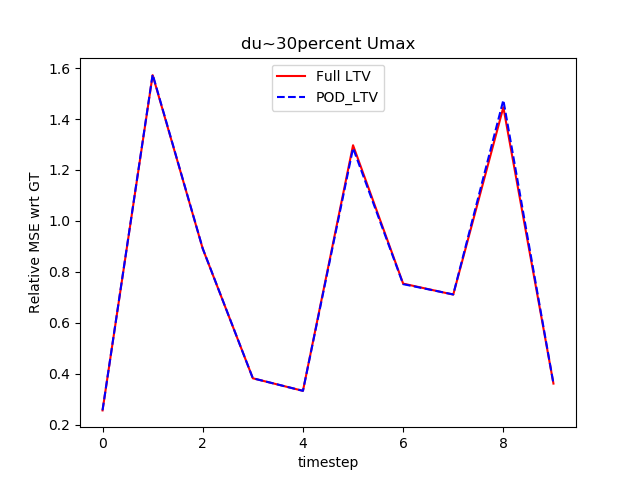}}

\end{multicols}
    
% \begin{multicols}{2}
%     %\hspace{1.4cm} 
%       \subfloat[Initial state]{\includegraphics[width=\linewidth]{images/roilqr plots/sys_id_10pc_comp.png}}    
%       \subfloat[Initial state]{\includegraphics[width=\linewidth]{images/roilqr plots/sys_id_30pc_comp.png}}
      
% \end{multicols}
\caption{\small Relative performance of the Full-order vs Reduced-order LTV systems identified w.r.t. the ground truth, benchmarked on the Allen-Cahn Equation. The trajectory is appended with a Gaussian noise of std 10\% and 30\% of the max control input.}
\label{ltv_podvsfo}
% \vspace{-0.5cm}
\end{figure}

\subsubsection{Backward pass}
For the reduced order model, we can compute the backward pass in a similar fashion as ILQR. Given the terminal conditions $v_T(\alpha_T)=\frac{\partial \bar{c}_T}{\partial \alpha}|_{\alpha_T}=\Phi^T \frac{\partial \bar{c}_T}{\partial x}|_{x_T}$ and $V_T(\alpha_T)=\nabla ^2_{\alpha \alpha}c_T|_{\alpha_T}=\Phi ^T V_T(x_T)\Phi$, the local optimal control in the reduced space is given by
\begin{equation*}
    \delta u_t=-(R+\hat{B}^T_t V_{t+1}\hat{B}_t)^{-1}(R\bar{u}_t+\hat{B}_t^T v_{t+1}+\hat{B}_t^T V_{t+1}\hat{A}_t\delta \alpha_t),
\end{equation*}
% $%$\begin{align}
%     \delta u_t
%     %&=R^{-1}f_{u_t}'(-v_{t+1}-V_{t+1}\delta x_{t+1})-\bar{u}_t \nonumber \\
%     %=R^{-1}\hat{B}_t^T(-v_{t+1}-V_{t+1}(\hat{A}_t\delta \alpha_t + \hat{B}_t\delta u_t))-\bar{u}_t ,%\nonumber \\
%     =-(R+\hat{B}^T_t V_{t+1}\hat{B}_t)^{-1}(R\bar{u}_t+\hat{B}_t^T v_{t+1}%\nonumber \\
%     +\hat{B}_t^T V_{t+1}\hat{A}_t\delta \alpha_t),$
%\end{align}
which can be written in the linear feedback form 
\begin{equation*}
   \delta u_t =-\hat{k}_t-\hat{K}_t\delta \alpha_t, 
\end{equation*}
where
\begin{align*}
    \hat{k}_t&=(R+\hat{B}_t^T V_{t+1}\hat{B}_t)^{-1}(R\bar{u}_t+\hat{B}_t^T v_{t+1})\\
    \hat{K}_t&=(R+\hat{B}_t^T V_{t+1}\hat{B}_t)^{-1}\hat{B}^T_t V_{t+1}\hat{A}_t.
\end{align*}
%$\delta u_t =-k_t-K_t\delta \alpha_t$, where $k_t=(R+\hat{B}_t^T V_{t+1}\hat{B}_t)^{-1}(R\bar{u}_t+\hat{B}_t^T v_{t+1})$ and $K_t=(R+\hat{B}_t^T V_{t+1}\hat{B}_t)^{-1}\hat{B}^T_t V_{t+1}\hat{A}_t$. 
Since $c_t(x_t)=\bar{c}_t(\alpha_t)$, the corresponding equations for $v_t$ and $V_t$ are
\begin{subequations}
\begin{align}
    v_t &= l_{t,\alpha}+\hat{A}_t^T v_{t+1}-\hat{A}_t^T V_{t+1}\hat{B}_t(R+\hat{B}_t^T V_{t+1}\hat{B}_t)^{-1}\nonumber \\
    &\cdot(\hat{B}_t^T v_{t+1}+R\bar{u}_t) \label{eq:vtr},\\
    V_t &= l_{t,\alpha \alpha}+\hat{A}_t^T (V_{t+1}^{-1}+\hat{B}_t R^{-1}\hat{B}_t^T)^{-1}\hat{A}_t,\nonumber \\
    &=l_{t,\alpha \alpha}+\hat{A}_t^T V_{t+1}\hat{A}_t-\hat{A}_t^T V_{t+1}\hat{B}_t(R+\hat{B}_t^T V_{t+1}\hat{B}_t)^{-1} \label{eq:Vtr} \nonumber \\
    &\cdot \hat{B}_t^T V_{t+1}\hat{A}_t,
\end{align}
\end{subequations}
where $l_{t,\alpha}=\Phi^T l_{t,x}$ and $l_{t,\alpha\alpha}=\Phi^T l_{t,xx}\Phi$. Again, from Eqs. \ref{eq:vtr} and \ref{eq:Vtr} we see that, given the terminal conditions and the local LTV model parameters $(\hat{A}_t, \hat{B}_t)$, we can do a backward-in-time sweep to compute all values of $v_t$ and $V_t$, and thus, the corresponding optimal control for the reduced order trajectory.

\subsubsection{Update trajectory}
The trajectory update step remains the same as in ILQR, except that the gains from the backward pass are in the reduced-order space as well. Thus, 
\begin{align*}
    \bar{u}_t^{k+1} &= \bar{u}_t^k + \hat{k}_t + \hat{K}_t(\alpha_t^{k+1}-\alpha_t^k), \\
    x^{k+1}_0 &= x^k_0.
\end{align*}
% $\bar{u}_t^{k+1} = \bar{u}_t^k + \hat{k}_t + \hat{K}_t(\alpha_t^{k+1}-\alpha_t^k)$, and $x^{k+1}_0 = x^k_0$.
% \begin{align*}
%     &\bar{u}_t^{k+1} = \bar{u}_t^k + \hat{k}_t + \hat{K}_t(\alpha_t^{k+1}-\alpha_t^k),\\
%     &x^{k+1}_0 = x^k_0 .
% \end{align*}

The updated cost is, thus, computed and, if the convergence criteria is not met, we iterate the process by repeating steps (1) to (4). Note that the new forward pass with the improved control sequence now results in a new reduced order basis as outlined in Step 1, and then, the process is repeated with the new reduced basis.\\

The algorithm is summarized and presented in Algorithm~\ref{roilqr_algo}.

\begin{algorithm}
  \caption{\strut Reduced-Order ILQR (RO-ILQR) Algorithm}
  % $\Rightarrow$ \textit{\textbf{Open-loop trajectory optimization}}\\
   {\bf Initialization:} Set state $x = x_0$, initial guess $\bar{U}^0 = \bar{u}^0_{0:T-1}$, line search parameter $\alpha = 1$, %regularization $\mu = 10^{-6}$, iteration counter $m = 0$, 
   convergence coefficient $\Gamma = 0.0001$.%, line search threshold $\sigma_1=0.3$.\\
  \While {$cost_{m}/cost_{m-1} < 1 - \Gamma$}{
  % \CommentSty{/* backward pass */}\\
  % \{$k^{k}_{0:T-1}, K^{k}_{0:T-1}\} \gets backward\_pass().$\\
  % \CommentSty{/* Forward pass */}\\
  1. $(\bar{X}^m,\bar{U}^m) \gets forward\_dynamics(x_0, \bar{u}_{0:T-1}^m)$.

  2. Apply the method of snapshots (Sec.~\ref{snapshotpod}) to get the reduced order basis $\Phi$ and the corresponding reduced order trajectory $(\bar{\alpha}^m, \bar{U}^m)$.

  % \CommentSty{/* Reduced Order LTV System Identification */}\\

  3. Run N rollouts to generate data matrices X and Y (Eq.~\ref{leastsq1}). Do the reduced-order LTV system identification as $[\hat{A_t} ~|~ \hat{B_t}]= Y X \T(X X \T)^{-1}.$
  
  % \CommentSty{/* Backward pass */}\\
4. Using Eqs.~\ref{eq:vtr} and \ref{eq:Vtr} and given terminal conditions for $v_T$ and $V_T$, compute ILQR gains $\{\hat{k}_{0:T-1}, \hat{K}_{0:T-1}\}$ through a backward-in-time sweep.
  % \begin{equation*}
  %     \delta u_t^m = -\hat{k}_t -\hat{K}_t \delta \alpha_t.
  % \end{equation*}

  5. Trajectory update: $\bar{u}_t^{m+1} = \bar{u}_t^m + \hat{k}_t + \hat{K}_t(\alpha_t^{m+1}-\alpha_t^m)$.
  
  % Initialize $z=0$.\\
  %     \While{$z < \sigma_1$}{
  %           Reduce $\alpha$, \\
  %           $u^{k+1}_{0:T-1}, cost_k, \Delta cost(\alpha) \gets forward\_pass(u^{k}_{0:T-1}, \{k^{k}_{0:T-1}, K^k_{0:T-1}\}, \alpha),$\\
  %           $z=(cost_k-cost_{k-1})/\Delta cost(\alpha),$\\
  %       }
  %       $k = k + 1.$
  %     }{
  %     $\bar{u}_{0:T-1} = u^{k+1}_{0:T-1}.$\\
  }
  % $\Rightarrow$ \textit{\textbf{The closed-loop feedback design}}\\
  % 1. $A_{t}, B_{t} \gets LLSCD(\bar{u}_{0:T-1},\bar{x}_{0:T-1}).$\\
  % 2. Calculate feedback gain $K_{0:T-1}$ from the Backward pass algorithm \ref{model_free_DDP_OL_BP}.\\
  % 3. Full closed-loop control policy: $u^*_t = \bar{u}_t + K_t \delta x_t$,\\where $\delta x_t$ is the state deviation from the nominal trajectory.
  \label{roilqr_algo}
\end{algorithm}

%% file: Reduced_sys_id.tex
\subsubsection{Reduced Order LTV System Identification}

We propose to identify an LTV model in the reduced order space, which drastically reduces the computational burden of the LTV identification and the subsequent backward pass. Using Eq. \ref{rodyn}, 
and taking perturbations about the trajectory $(\bar{\alpha}_t,\bar{u}_t)$, we can linearize the system to get the corresponding linear time-varying approximation of the nonlinear reduced-order dynamics about the nominal trajectory:
\begin{align*}
    \delta \alpha_{t+1} =\hat{A}_t\delta \alpha_t + \hat{B}_t\delta u_t,
\end{align*}
where $\hat{A}_t=\hat{f}_{\alpha_t}\in \mathbb{R}^{n_\alpha\times n_\alpha}$ and $\hat{B}_t=\hat{f}_{u_t}\in \mathbb{R}^{n_\alpha \times n_u}$. To compute the reduced order model in a data-based fashion, we make use of the standard least-squares method, given the input-output experiment data.\\
% \textcolor{blue}{To compute these in a data-based fashion, we make use of the standard least-squares method, given the input-output experiment data.}\\

\noindent \textit{Least Squares and Sample Efficiency:}
\label{sys_id_solve}

%All the perturbations are zero-mean, i.i.d, Gaussian noise with covariance matrix $\sigma I$. $\sigma$ is a $o(u)$ small value selected by the user.

Run N simulations for each step and collect the input-output data: $Y = [\hat{A_t} ~|~ \hat{B_t}] X$
% \begin{align}
% \label{sysmat}
% & Y = [\hat{A_t} ~|~ \hat{B_t}] X,
% \end{align}
and write out the components:
\begin{align}
& Y = \begin{bmatrix}
\delta \alpha_{t+1}^{(1)}& \delta \alpha_{t+1}^{(2)} &\cdots& \delta \alpha_{t+1}^{(N)} \\
\end{bmatrix},
\nonumber\\
& X=\begin{bmatrix}
  \delta \alpha_t^{(1)} & \delta \alpha_t^{(2)}& \cdots & \delta \alpha_t^{(N)} \\
  \delta u_t^{(1)} &  \delta u_t^{(2)} & \cdots & \delta u_t^{(N)} \\
\end{bmatrix},
\label{leastsq1}
\end{align}
where $\delta u_{t}^{(n)}$ is the control perturbation vector we feed to the system at step $t$ of the $n$\textsuperscript{th} simulation, and $\delta \alpha_{t}^{(n)}$ is the reduced-order state perturbation vector that we get from running forward simulations/rollouts with the above control perturbations, and projecting the observed state perturbation vector $\delta x_t^{(n)}$ onto the reduced order subspace, \textit{i.e.}
\begin{equation*}
    \delta \alpha_t ^{(n)} = \Phi^T \delta x_t^{(n)}. 
\end{equation*}
All the perturbations are zero-mean, i.i.d, Gaussian noise with covariance matrix $\sigma I$. $\sigma$ is a $o(u)$ small value selected by the user. $\delta \alpha_{t+1}^{(n)}$ denotes the deviation of the output state vector  from the nominal state projected in the reduced order subspace, after propagating for one step.

Finally, using the standard least square method, the linearized system parameters are estimated as 
\begin{equation*}
    [\hat{A_t} ~|~ \hat{B_t}]=Y X \T(X X \T)^{-1}.
\end{equation*}
We are free to choose the distribution of $\delta {\phi_t}$ and $\delta {u_t}$. Given that we perform rollouts of the system, where $\{\delta u_t^{(i)}\}$ is a Gaussian white noise sequence for all rollouts $i$, the reduced-order state perturbations $\delta \alpha_t^{(i)}$ are also Gaussian and independent for the different rollouts $i$, for any given time $t$. This ensures that $\delta X_t \delta X_t'$ is very well conditioned.
To see this, let us consider the terms in the matrix $\delta X_t \delta X_t'=\begin{bmatrix}  \delta {\Lambda_t} \delta {\Lambda_t}' & \delta {\Lambda_t} \delta {U_t}' \\ \delta {U_t} \delta {\Lambda_t}'  & \delta {U_t} \delta {U_t}' \end{bmatrix}$,~$\delta {\Lambda_t} \delta {\Lambda_t}' = \sum_{i=1}^{n_s} \delta {\alpha_t}^{(i)} {\delta {\alpha_t}^{(i)}}'$, where `$n_s$' is the number of samples for each of the random variables, $\delta {\alpha_t}$ and $\delta {u_t}$, and we denote the random samples as $\delta {\Lambda_t} = \begin{bmatrix} \delta {\alpha_t^1}& \delta {\alpha_t^2}& \ldots &\delta {\alpha_t^{n_s}}\end{bmatrix}$, $\delta {U_t} = \begin{bmatrix} \delta {u_t^1} &\delta {u_t^2}& \ldots& \delta {u_t^{n_s}}\end{bmatrix}$ and $\delta X_t = \begin{bmatrix} \delta \Phi_t & \delta U_t \end{bmatrix}$. Similarly, $\delta {U_t} \delta {U_t}' = \sum_{i=1}^{n_s} \delta {u_t}^{(i)} {\delta {u_t}^{(i)}}'$, $\delta {U_t} \delta {\Lambda_t}' = \sum_{i=1}^{n_s} \delta {u_t}^{(i)} {\delta {\alpha_t}^{(i)}}'$ and $\delta {\Lambda_t} \delta {U_t}' = \sum_{i=1}^{n_s} \delta {\alpha_t}^{(i)} {\delta {u_t}^{(i)}}'$. 

From the definition of sample variance, for a large enough $n_s$, we can write the above matrix as:
\begin{equation*}
\begin{split}
    \delta X_t \delta X_t' &= \begin{bmatrix} \sum_{i=1}^{n_s} \delta {\alpha_t}^{(i)} {\delta {\alpha_t}^{(i)}}' & \sum_{i=1}^{n_s} \delta {\alpha_t}^{(i)} {\delta {u_t}^{(i)}}' \\ \sum_{i=1}^{n_s} \delta {u_t}^{(i)} {\delta {\alpha_t}^{(i)}}' & \sum_{i=1}^{n_s} \delta {u_t}^{(i)} {\delta {u_t}^{(i)}}'
    \end{bmatrix}\\ &\approx \begin{bmatrix} \sigma^2(n_s - 1) {\text I_{n_\alpha}} & {\text 0_{n_\alpha \times n_u}} \\ 0_{n_u \times n_\alpha} & \sigma^2 (n_s - 1) {\text I_{n_u}}\end{bmatrix} \\
    &= \sigma^2 (n_s - 1){\text I}_{(n_\alpha+n_u) \times (n_\alpha+n_u)}
\end{split}
\end{equation*}

%% file: Preliminaries.tex
% \onecolumn

\section{Convergence analysis of the reduced order solution}
\label{analysis}
In this section, we provide a mathematical analysis to bound the sub-optimality of the approach introduced
 in Section \ref{sec4}. We first introduce the perturbed LQR problem about the nominal trajectory, and under certain assumptions show that the solutions to the full-order and reduced-order perturbed LQR problems at each iteration are uniformly bounded. Next, we determine the compact set to which the algorithm descends monotonically, and consequently, determine the convergence guarantees of the proposed approach.

\subsection{Preliminaries}
Given the optimization problem \ref{eq:FO-OCP}, we can write the corresponding trajectories $\Bar{X}=\{x_0,x_1,..,x_T\}$ in the reduced order subspace $\Phi$, computed as given in Sec.~\ref{snapshotpod}, as 
\begin{equation*}
    \Bar{X} = \Phi \Bar{\alpha},
\end{equation*}

where $\Bar{\alpha}=\{\alpha_0,\alpha_1,...,\alpha_T\}$. 

% \textit{i.e.}, $x_t \approx \Phi \alpha_t$ $\forall t$. Thus, we can write the reduced order optimization problem as

% \begin{align*}
%     & \min_{u_t} \hat{J}=\sum_{t = 0}^{T-1} \bar{c}_t(\alpha_{t}, u_{t}) + \bar{c}_T(\alpha_T) | \alpha_0 = \alpha,\\
%     & \text{Subject to:}\  \alpha_{t+1}=\Phi^T f(\Phi \alpha_t,u_t),
% \end{align*}
% which, by selecting $\hat{Q}=\Phi^T Q \Phi$, $\hat{R}=R$, can be written as
% \begin{align}
%     & \min_{u_t} \hat{J}(\Phi \alpha_t)=\sum_{t = 0}^{T-1} c_t(\Phi \alpha_{t}, u_{t}) + c_T(\Phi \alpha_T) \label{eq:RO-OCP}\tag{RO-OCP} \\
%     & \text{Subject to:}\  \alpha_{t+1}=\Phi^T f(\Phi \alpha_t,u_t).
% \end{align}
% \textcolor{red}{I don't think the above problem formulation is necessary, we never use it in the following.}

Let the cost functions associated with \ref{eq:FO-OCP} and \ref{eq:RO-OCP} be $J(\cdot)$ and $\Hat{J}(\cdot)$ respectively. If there was no loss in information when projecting upon the reduced order subspace, \textit{i.e.}, $x_t = \Phi \alpha_t$, the functions $J(\cdot)$ and $\hat{J}(\cdot)$ coincide, and thus correspond to the same optimal solution. But, in general, neglecting the higher order POD modes results in some loss of information, which is reflected as deviation in the trajectories of the reduced order system as compared to the original system. Thus, our optimizations will operate on different functions $J$ and $\hat{J}$, and hence, we may lose guarantees regarding the optimality of the reduced order solution.

Heuristically, if $x_t\approx \Phi \alpha_t$, the performance of the solution of the reduced LQR problem should be close to the performance of the full order solution.\\

% \textcolor{red}{These assumptions should come after you have introduced the LQR problem at the $k^{th}$ iteration below and should have the superscript $k$. Write A1-3 together. }

\subsubsection{Perturbed LQR problem}
% \textcolor{red}{Change the cost to quadratic in the RO-LQR and FO-LQR below?}
Given the current nominal trajectory $(\bar{X}^{(k)},\bar{U}^{(k)})$ at the $k^{th}$ iterate, and the initial condition $x(0)=x_0$, we can linearize around the trajectory $(X^{(k)}, U^{(k)})$ to get the linear time-varying dynamics
\begin{equation}
    \delta x^{(k)}_{t+1}=A_t \delta x^{(k)}_t+ B_t \delta u^{(k)}_t. \label{eq:FO-LTV}
\end{equation}

% From A\ref{A2}, we can write $\delta x_t = \Phi \delta \alpha_t$. Thus, we have the reduced order LTV dynamics
% \begin{equation}
%     \delta \alpha_{t+1}=\Phi^T A_t \Phi \delta \alpha_t+ \Phi^T B_t \delta u_t. \label{eq:RO-LTV}
% \end{equation}

We can also write the perturbed cost about the nominal as
\begin{align*}
    \delta J^{(k)} = &\sum_{t=0}^{T-1} ((c^{(k)}_x)^T \delta x^{(k)}_t + (\delta x^{(k)}_t)^T c^{(k)}_{xx} \delta x^{(k)}_t + (c^{(k)}_u)^T \delta u^{(k)}_t \\
    &+ (\delta u^{(k)}_t)^T c^{(k)}_{uu} \delta u^{(k)}_t) + (c^{(k)}_{T,x})^T \delta x^{(k)}_T \\
    &+ (\delta x^{(k)}_T)^T c^{(k)}_{T,xx}\delta x^{(k)}_T.
\end{align*}
For simplicity, in the following analysis, we consider a cost function which is quadratic in state and control. However, it is important to note that the methods and results derived can be extended to general cost functions. Thus, we can write 
$c^{(k)}_{x} (\Bar{x}^{(k)}_t,\Bar{u}^{(k)}_t) = (Q^{(k)}_t)^T \bar{x}^{(k)}_t,c^{(k)}_{xx} (\cdot)= Q^{(k)}_t, c^{(k)}_{u}(\cdot) = (R^{(k)})^T \Bar{u}^{(k)}_t$, and $c^{(k)}_{uu}(\cdot) = R^{(k)}$. Then, the \ref{eq:FO-OCP} can be expressed as a perturbed LQR problem about the nominal, \textit{i.e.},

\begin{align}
    & \min_{\delta u_t^{(k)}} \delta J^{(k)} = \sum_{t=0}^{T-1} \Bigg((\Bar{x}^{(k)}_t)^T Q^{(k)}_t \delta x^{(k)}_t + (\delta x^{(k)}_t)^T Q^{(k)}_t \delta x^{(k)}_t+ \nonumber \\ &~~~~~~(\bar{u}_t^{(k)})^T R^{(k)} \delta u^{(k)}_t 
    +(\delta u^{(k)}_t)^T R^{(k)} \delta u^{(k)}_t\Bigg)+\nonumber \\ &~~~~~~ \Bigg((\Bar{x}^{(k)}_T)^T Q^{(k)}_T \delta x^{(k)}_T + (\delta x^{(k)}_T)^T Q^{(k)}_T \delta x^{(k)}_T \Bigg)
    \label{eq:FO-LQR}\tag{FO-LQR} \\
    & \text{Subject to:}\  \delta x^{(k)}_{t+1}=A_t \delta x^{(k)}_t + B_t \delta u^{(k)}_t. \nonumber
\end{align}

Similarly, we can write the perturbed reduced-order LQR problem as 

\begin{align}
    & \min_{\delta u_t^{(k)}} \delta \hat{J}^{(k)} = \sum_{t=0}^{T-1} \Bigg((\Phi \bar{\alpha}_t^{(k)})^T Q_t^{(k)} \Phi \delta \alpha^{(k)}_t + \nonumber \\
    &~~~~~ (\Phi\delta \alpha^{(k)}_t)^T Q_t^{(k)} (\Phi\delta \alpha^{(k)}_t) + (\bar{u}_t^{(k)})^T R^{(k)} \delta u^{(k)}_t + \nonumber \\
    &~~~~~(\delta u^{(k)}_t)^T R^{(k)} \delta u^{(k)}_t \Bigg) + (\Phi \bar{\alpha}_T^{(k)})^T Q_T^{(k)} (\Phi\delta \alpha^{(k)}_T) + \nonumber \\
    &~~~~~ (\Phi\delta \alpha^{(k)}_T)^T Q_T^{(k)}(\Phi\delta \alpha^{(k)}_T),\label{eq:RO-LQR}\tag{RO-LQR} \\
    & \text{Subject to:}\  \delta \alpha^{(k)}_{t+1}=\Phi^T A_t \Phi \delta \alpha^{(k)}_t + \Phi^T B_t \delta u^{(k)}_t. \nonumber
\end{align}

We now state the assumptions under which we establish guarantees regarding the performance of the reduced-order algorithm. 
\begin{assumption}{(\textbf{A1})}\label{A1}
    We assume that $x^{(k)}_t$ is close to $\Phi \alpha^{(k)}_t$, for all iterates $k$ and
    % where $\Phi$ is the reduced order basis as determined in Sec.~\ref{dynamics}, 
    for all snapshots along the trajectory. Also, for the linearized systems above, the same holds, \textit{i.e.}, 
    \begin{align*}
      &||x^{(k)}_t - \Phi \alpha^{(k)}_t||\leq \epsilon~~ \forall t,  \\
      &||\delta x^{(k)}_t-\Phi \delta \alpha^{(k)}_t|| \leq 2\epsilon~~\forall t.
    \end{align*}
    % $||x^{(k)}_t - \Phi \alpha^{(k)}_t||\leq \epsilon$ $\forall t$. 
    
\end{assumption}

% \begin{assumption}{(A2)}\label{A2}
%     For small perturbations around the trajectory, the basis does not vary significantly, and can be approximated in a linear sense, \textit{i.e.},
%     $||\delta x_t-\Phi \delta \alpha_t|| \leq 2\epsilon$ $\forall t$. \textcolor{red}{Why $2 \epsilon$? Since at this point, you would have already introduced the delta variables, just write the equation above. You may want to lump A1 and 2 together. }
% \end{assumption}

\begin{assumption}{\textbf{(A2)}}\label{A3}
% \textcolor{red}{Why $Q_t^k$ etc, are you assuming that the incremental cost is quadratic?}.
    The cost functions $c^{(k)}_t(\cdot)$ and $c^{(k)}_T(\cdot)$ at the $k^{th}$ iterate are chosen such that the states $x^{(k)}_t$, $\delta x^{(k)}_t$ and $\Phi \delta \alpha^{(k)}_t$ are uniformly bounded, \textit{i.e.}, $||Q^{(k)}_t x_t^{(k)}||,||Q^{(k)}_t \delta x_t^{(k)}||, ||Q^{(k)}_t \Phi \delta \alpha_t^{(k)}||\leq C(k) \leq \Bar{C}$ for all $t=\{0,1,...,T-1\}$, and $||Q^{(k)}_T x_T^{(k)} ||,||Q^{(k)}_T \delta x_T^{(k)}||,||Q^{(k)}_T \Phi \delta \alpha_T^{(k)}||\leq C_T(k) \leq \Bar{C}$.
   
    % \begin{align*}
    %     & ||Q^{(k)}_t x_t^{(k)}||\leq C(k) \leq \Bar{C}, \\
    %     & ||Q^{(k)}_t \delta x_t^{(k)}||\leq C(k) \leq \Bar{C}, \\
    %     & ||Q^{(k)}_t \Phi \delta \alpha_t^{(k)}||\leq C(k) \leq \Bar{C} 
    % \end{align*}
    % for all $t=\{0,1,...,T-1\}$, and
    % \begin{align*}
    %     & ||Q^{(k)}_T x_T^{(k)} ||\leq C_T(k) \leq \Bar{C}, \\
    %     & ||Q^{(k)}_T \delta x_T^{(k)}||\leq C_T(k) \leq \Bar{C}, \\
    %     & ||Q^{(k)}_T \Phi \delta \alpha_T^{(k)}||\leq C_T(k) \leq \Bar{C}. 
    % \end{align*}
\end{assumption}
\begin{assumption}{\textbf{(A3)}}\label{A4}
    The smallest singular value of the Hessian of the cost function is uniformly bounded away from zero at every iterate, \textit{i.e.}, $\sigma_{min}(\frac{\partial^2 \delta J}{\partial (\delta U)^2}|_{\delta U^{(k)}}) > \bar{\sigma}>0$.
    % We assume that the solution $\delta U^*$ lies in a set $U$ such that $\sigma_{min}(\frac{\partial^2 \delta J}{\partial (\delta U)^2}|_{\delta U^*}) > \bar{\sigma}>0$.
\end{assumption}
% \textcolor{red}{this assumption should be with A1 and A2. We are saying that the lowest singular value of the Hessian of the cost function at every iterate is uniformly bounded away from zero.}

With the preliminaries stated above, we introduce the following results. Note that in the subsequent development, we drop the use of the superscript $k$ to simplify notation.\\
% \textcolor{red}{In the following development, you can drop the superscript $k$ to simplify notation but you should mention the same at this point.}

\begin{lemma}\label{L1}
    Under assumptions \textbf{A\ref{A1}} and \textbf{A\ref{A3}}, given the same control sequence $\delta U$, the costs of the \ref{eq:FO-LQR} and \ref{eq:RO-LQR} satisfy:
    \begin{equation*}
     |\delta J(\delta U)-\delta \hat{J}(\delta U)|\leq \bar{C}_1\epsilon~~\forall \delta U,   
    \end{equation*}
    % $|\delta J(\delta U)-\delta \hat{J}(\delta U)|\leq \bar{C}_1\epsilon$ $\forall \delta U$, 
    where $\Bar{C}_1=7(T+1)\Bar{C}$ and $\delta U = \{\delta u_t\}_{t=0}^{T-1}$  (Fig.~\ref{fig:costs_close_samectrl}).
\end{lemma}
% \textcolor{red}{Call $\delta u$, $\delta U$.}\\

\begin{proof}
    See Sec~\ref{l1proof}.
\end{proof}

\begin{figure}
\centering
\def\svgwidth{1\columnwidth}
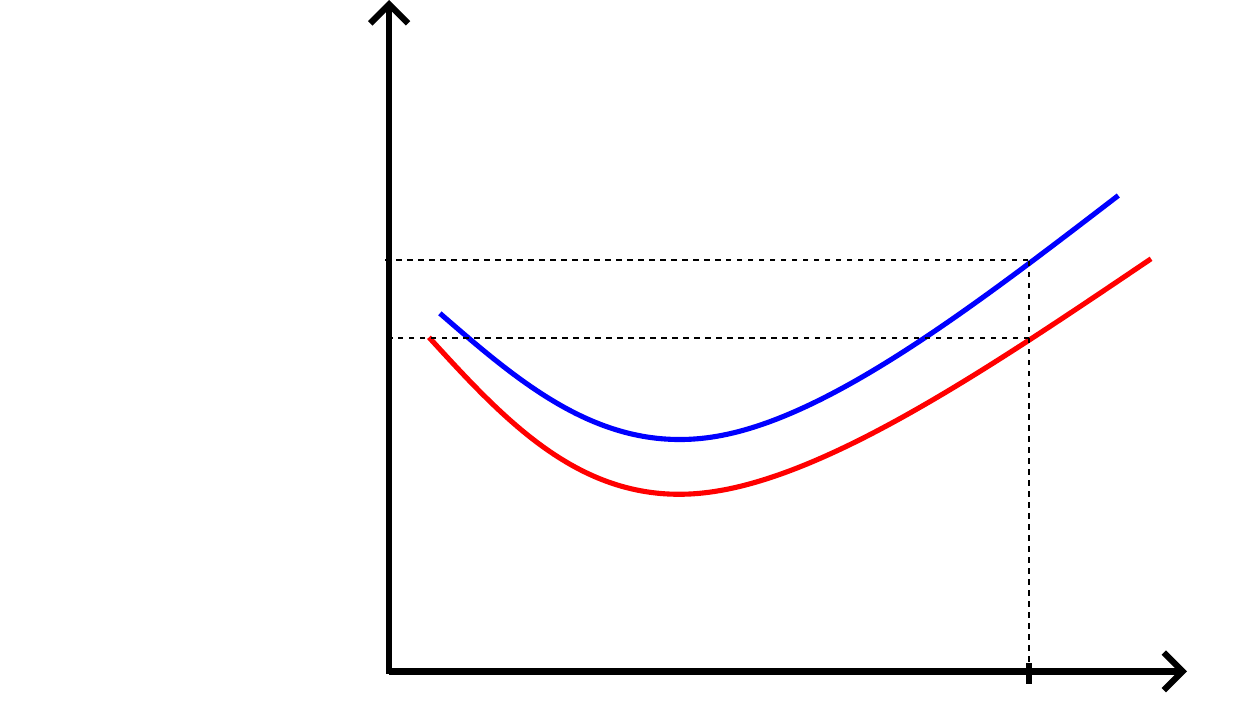
\caption{\small The cost functions $\delta J(\cdot)$ and $\delta\hat{J}(\cdot)$ are close to each other given the same control input.}
\label{fig:costs_close_samectrl}
\end{figure}

\begin{lemma}\label{L2}
    Let the optimal control inputs corresponding to the minima of the cost functions $\delta J$ and $\delta \hat{J}$ be $\delta U^*$ and $\delta \hat{U}^*$ respectively. Then, under assumption \textbf{A\ref{A1}} and Lemma \ref{L1}, the minima of the \ref{eq:FO-LQR} and \ref{eq:RO-LQR} satisfy: 
    \begin{equation*}
       |\delta J(\delta U^*)-\delta \hat{J}(\delta \hat{U}^*)|\leq \Bar{C}_1 \epsilon. 
    \end{equation*}
\end{lemma}

\begin{proof}
    See Sec~\ref{l2proof}.
\end{proof}

Now, given Lemmas \ref{L1} and \ref{L2}, we can discuss the comparisons between the corresponding optimal control. \\

\begin{lemma}\label{L3}
    The optimal control sequences $\delta U^*$ and $\delta \hat{U}^*$ obtained from solving the \ref{eq:FO-LQR} and \ref{eq:RO-LQR} satisfy: 
    \begin{equation*}
        ||\delta U^* -\delta \hat{U}^*||\leq \delta,
    \end{equation*}
    % $||\delta U^* -\delta \hat{U}^*||\leq \delta$, 
    where $\delta = \sqrt{\frac{2\Bar{C}_1 \epsilon}{\Bar{\sigma}}}$ (Fig.~\ref{fig:ctrl_same}). 
    % \textcolor{red}{What is $\delta$, define it.}
\end{lemma}

\begin{proof}
    See Sec~\ref{l3proof}.
\end{proof}

Thus, we know that, under the specified assumptions, the solution to the \ref{eq:FO-LQR} and \ref{eq:RO-LQR} are close to each other for each iteration. 

\begin{figure}
\centering
\transparent{1.0}
\def\svgwidth{1\columnwidth}
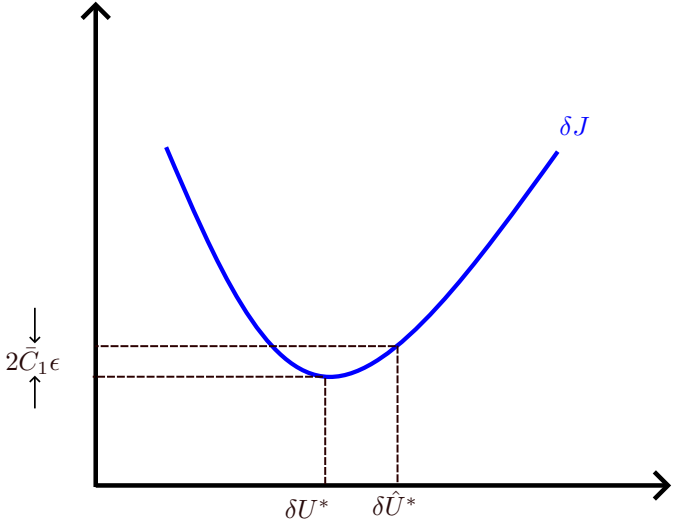            
\caption{\small The solutions of the full and reduced order perturbed LQR problems are close to each other, under the specified assumptions.}
\label{fig:ctrl_same}
\end{figure}

\subsection{Convergence of the Reduced Order Formulation}
\begin{lemma}\label{L4}
     Let $S_\infty = \{\delta U: ||[\nabla^2 J]^{-1}\nabla J||\leq \delta \}$, then for any $\delta U\notin S_\infty$, the \ref{eq:RO-LQR} will always move in a direction such that the cost decreases. 
    
    % \textcolor{red}{the definition of $S_{\infty}$ should be = \{ $\delta U$ s.t..\}}.
\end{lemma}
\begin{proof}
    See Sec~\ref{l4proof}.
\end{proof}

We now study the behavior of the optimization as it proceeds. Let $S_\infty$ be enclosed between two sub-level sets of the cost function, say $\bar{S}$ and $\underline{S}$, such that $S_\infty$ completely encloses $\underline{S}$, while being completely enclosed by $\Bar{S}$ (Fig.~\ref{S_infty_compact}). Let the costs for any control policy in the sub-level sets $\bar{S}$ and $\underline{S}$ be upper-bounded by the costs $\bar{J}$ and $\underline{J}$ respectively.  

\begin{figure}
\centering
\def\svgwidth{0.6\columnwidth}
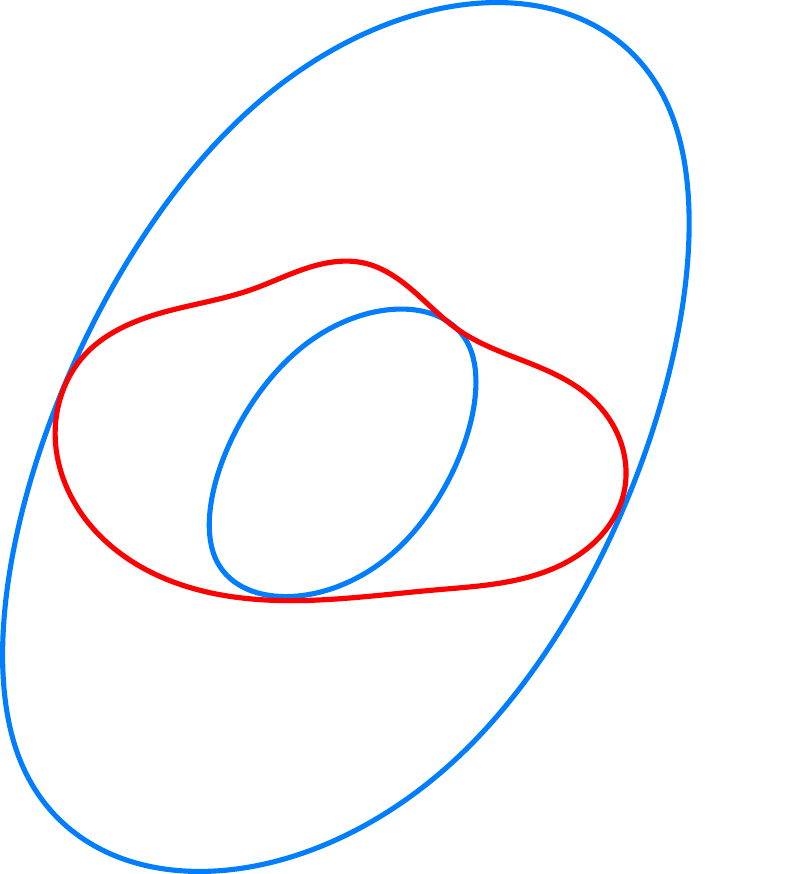

\caption{\small The set $S_\infty$ is compact, and enclosed between sub-level sets $\underline{S}$ and $\Bar{S}$.}
\label{S_infty_compact}
\end{figure}

% \textcolor{red}{What is $\bar{S}_{\infty}$?}\\

% \textcolor{red}{I think there is a confusion between the values $\bar{J}, \underline{J}$ and the sub-level sets $\bar{S}, \underline{S}$.}
We introduce the following assumption to ensure that the algorithm never increases the cost:

\begin{assumption}{\textbf{(A4)}}\label{A5}
    We enforce the cost to be non-increasing through all iterations of the algorithm, \textit{i.e.}, $J^{k+1}-J^{k}\leq 0, \, \forall \, k$.
\end{assumption}

Let us assume, for some $(X_0,U_0)$ at $t=0$, the initial cost $J_0>\Bar{J}$. But we know that the cost is always non-negative. Since $J$ is monotonically decreasing, then for some $t=k$, 
\begin{align*}
    J^k \leq \Bar{J} \text{ s.t. } U^{(k)}\in S_\infty.
\end{align*}

From Eq.~\ref{eq:monotonic},
\begin{equation*}
    J^{k+1}-J^{k} \leq \beta_2,
\end{equation*}

where $\beta_2=\beta (||[\nabla^2 J]^{-1}\nabla J||-||w||)\geq 0$. This results in three possible cases (Fig.~\ref{fig:Jk_threecases}): 
\begin{figure}
\centering
\def\svgwidth{0.8\columnwidth}
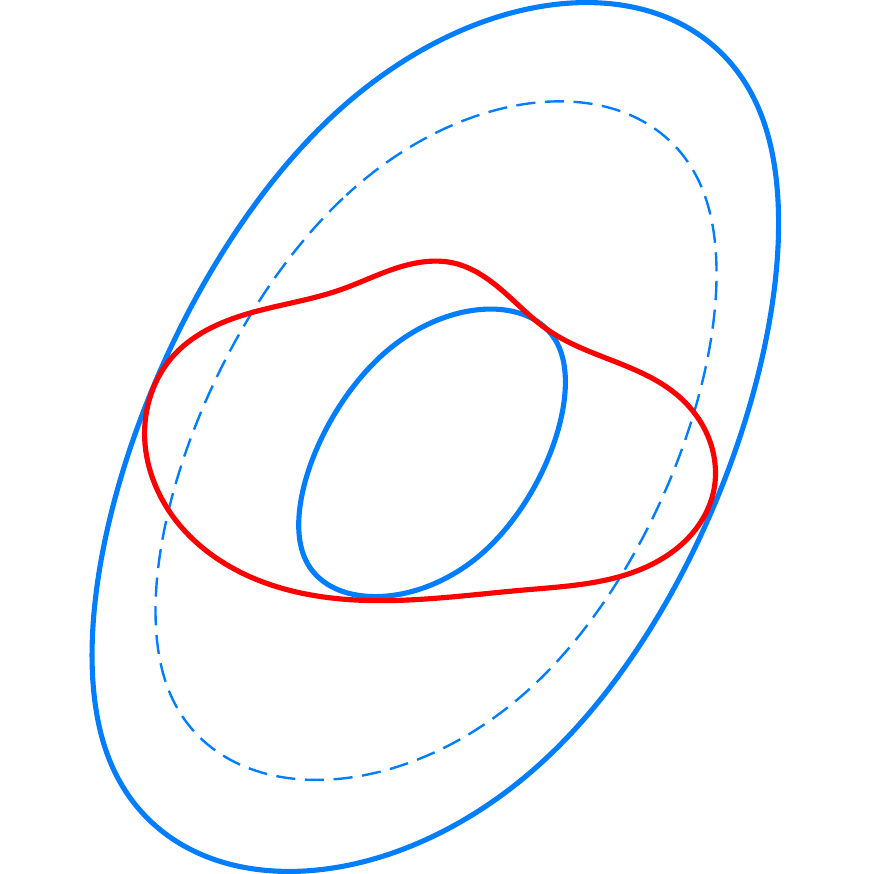            
\caption{\small The cost lying within the sublevel set $\Bar{S}$, with the trajectory being in $S_\infty$, results in three possible cases: the subsequent estimate may either jump out of $S_\infty$ but still descend, stay in $S_\infty$ and descend, or terminate at the same cost. The algorithm's progress after it hits $S_{\infty}$ is some sequence of these three cases till the algorithm terminates due to the first case or reaches the set $\underline{S}$, after which it cannot leave $S_{\infty}$. Thus, the algorithm will eventually converge to the set $S_\infty$ and stay there, however, this is the strongest guarantee of convergence. Thus, that the reduced algorithm converges to $S_\infty$ is all we can guarantee. This is the price we pay for the computational efficiency obtained from model reduction.}
\label{fig:Jk_threecases}
\end{figure}

\textbf{\textit{Case 1 :}} $\beta_2 \geq J^{k+1}-J^k \geq 0$

Due to A\ref{A5}, this is infeasible and thus, the algorithm terminates at this cost, and remains within the set $S_\infty$.

\textbf{\textit{Case 2 :}} $J^{k+1}-J^k <0$ and $U^{k+1}\in S_\infty$.

Thus, despite being in $S_\infty$, the cost may still descend. %At the next iteration, this will, again, either result in the cost terminating within the set $S_\infty$ (Case 1), or the cost will keep descending until $J<\underline{J}$.

\textbf{\textit{Case 3 :}} $J^{k+1}-J^k <0$, but $U^{k+1}\notin S_\infty$, \textit{i.e.}, the algorithm jumps out of $S_\infty$.
Then, from L\ref{L4}, for any $t\geq k+1$, $J^{t+1}-J^t<0$ until $U^t\in S_\infty$, i.e., the algorithm will return to the set $S_{\infty}$. 

At the next iteration, Case 2 or Case 3 will either result in the iteration terminating (Case 1), the cost staying within $S_\infty$ and descending (Case 2), or jumping out again, and monotonically decreasing till it re-enters $S_\infty$ (Case 3). This process can only continue until $J\leq \underline{J}$ at which point Case 3 is no longer feasible and the algorithm remains within $S_{\infty}$.

% \textbf{\textit{Case 1 :}} $J^k \leq \Bar{J}$, but $(X^k, U^k)\in \Bar{S}_\infty$. Then, for any $t\geq k$, $J^{t+1}-J^t<0$ until $(X^t, U^t)\in S_\infty$.

% \textbf{\textit{Case 2 :}} $J^k \leq \Bar{J}$, but $(X^k, U^k)\in S_\infty$. From Eq.~\ref{eq:monotonic},

% \begin{equation*}
%     J^{k+1}-J^k \leq \beta_2,
% \end{equation*}

 % Again, this results in two sub-cases:

% \textbf{\textit{Case 2.1 :}} $J^{k+1}-J^k <0$.
% Thus, despite being in $S_\infty$, the cost may still descend.\\

% \textbf{\textit{Case 2.2 :}} $\beta_2 \geq J^{k+1}-J^k \geq 0$. 

% Along with A\ref{A5}, this is only valid if $J^{k+1}-J^k=0$. Thus, the iteration terminates at this cost.

% \begin{remark}\label{remark_1}
%     If the cost descends further in Case 2, and jumps out of $S_\infty$, it will descend monotonically until it re-enters $S_\infty$, wherein it may either get lucky and keep descending, or terminate somewhere within $S_\infty$ (Fig.~\ref{fig:S_infty_conv}).
% \end{remark}
% \begin{remark}
%     If the cost keeps descending without entering $S_\infty$, then it will enter the set $\underline{J}$, upon which we lose any further guarantees, as it may descend further or terminate, depending upon the gradient and errors.
% \end{remark}

\begin{remark}
    Upon reaching the sub-level set $\underline{S}$, the iteration may either terminate, or keep descending until it reaches the optimal trajectory, depending on the gradient and errors. Unfortunately, there are no guarantees within this region, and, thus, the best we can say is that the algorithm converges to some point within the set $\underline{S}$. 
\end{remark}

% \begin{figure}
% \centering
% \includegraphics[width=0.8\linewidth]{figure/S_infty_conv_3.png}
            
% \caption{\small The algorithm will definitively converge to the set $S_\infty$, upon which we lose all guarantees of convergence. Thus, that the reduced algorithm converges to $S_\infty$ is all we can guarantee.}
% \label{fig:S_infty_conv}
% \end{figure}
The above development can be summarized in the following result. 
\begin{theorem}    
    Under assumptions A\ref{A1}, A\ref{A3}, A\ref{A4} and A\ref{A5}, the reduced order formulation of the ILQR is guaranteed to converge to the set $S_\infty$.
\end{theorem}

\begin{remark}
    Note that the bounds computed above can be highly conservative, owing to the application of worst case bounds in Lemmas 1-3. In practice, the bound $\delta = \sqrt{\frac{2\bar{C}_1\epsilon}{\bar{\sigma}}}$, parametrizing the algorithm's limit set $S_{\infty}$ is much smaller, and consequently, the reduced order solution is much closer to the true solution as shall be seen from our empirical results later in the paper.
\end{remark}
% \textcolor{red}{I think the analysis is almost done, however, I feel we need to have a discussion of the $S_{infty}$ set and how $\delta$ determines it. }

%% file: figure/fig_L1.pdf_tex
%% Creator: Inkscape 1.2.1 (9c6d41e410, 2022-07-14), www.inkscape.org
%% Accompanies image file 'fig_L1.pdf' (pdf, eps, ps)
\begingroup%
  \makeatletter%
  \providecommand\color[2][]{%
    \errmessage{(Inkscape) Color is used for the text in Inkscape, but the package 'color.sty' is not loaded}%
    \renewcommand\color[2][]{}%
  }%
  \providecommand\transparent[1]{%
    \errmessage{(Inkscape) Transparency is used (non-zero) for the text in Inkscape, but the package 'transparent.sty' is not loaded}%
    \renewcommand\transparent[1]{}%
  }%
  \providecommand\rotatebox[2]{#2}%
  \newcommand*\fsize{\dimexpr\f@size pt\relax}%
  \newcommand*\lineheight[1]{\fontsize{\fsize}{#1\fsize}\selectfont}%
  \ifx\svgwidth\undefined%
    \setlength{\unitlength}{595.51871713bp}%
    \ifx\svgscale\undefined%
      \relax%
    \else%
      \setlength{\unitlength}{\unitlength * \real{\svgscale}}%
    \fi%
  \else%
    \setlength{\unitlength}{\svgwidth}%
  \fi%
  \global\let\svgwidth\undefined%
  \global\let\svgscale\undefined%
  \makeatother%
  \begin{picture}(1,0.5810119)%
    \lineheight{1}%
    \setlength\tabcolsep{0pt}%
    \put(0,0){\includegraphics[width=\unitlength,page=1]{fig_L1.pdf}}%
    \put(0.81293904,0.00266394){\color[rgb]{0,0,0}\makebox(0,0)[lt]{\lineheight{1.25}\smash{\begin{tabular}[t]{l}$\delta U$\end{tabular}}}}%
    \put(0.9145842,0.44040058){\color[rgb]{0,0,0}\makebox(0,0)[lt]{\lineheight{1.25}\smash{\begin{tabular}[t]{l}$\delta J$\end{tabular}}}}%
    \put(0.93392545,0.38501256){\color[rgb]{1,0,0}\makebox(0,0)[lt]{\lineheight{1.25}\smash{\begin{tabular}[t]{l}$\delta \hat{J}$\end{tabular}}}}%
    \put(0.00703989,0.33197246){\color[rgb]{0,0,0}\makebox(0,0)[lt]{\lineheight{1.25}\smash{\begin{tabular}[t]{l}$|\delta J-\delta \hat{J}|\leq \bar{C}_1 \epsilon$\end{tabular}}}}%
  \end{picture}%
\endgroup%

%% file: figure/fig_L3_2.pdf_tex
%% Creator: Inkscape 1.2.1 (9c6d41e410, 2022-07-14), www.inkscape.org
%% Accompanies image file 'fig_L3_2.pdf' (pdf, eps, ps)
\begingroup%
  \makeatletter%
  \providecommand\color[2][]{%
    \errmessage{(Inkscape) Color is used for the text in Inkscape, but the package 'color.sty' is not loaded}%
    \renewcommand\color[2][]{}%
  }%
  \providecommand\transparent[1]{%
    \errmessage{(Inkscape) Transparency is used (non-zero) for the text in Inkscape, but the package 'transparent.sty' is not loaded}%
    \renewcommand\transparent[1]{}%
  }%
  \providecommand\rotatebox[2]{#2}%
  \newcommand*\fsize{\dimexpr\f@size pt\relax}%
  \newcommand*\lineheight[1]{\fontsize{\fsize}{#1\fsize}\selectfont}%
  \ifx\svgwidth\undefined%
    \setlength{\unitlength}{443.42470276bp}%
    \ifx\svgscale\undefined%
      \relax%
    \else%
      \setlength{\unitlength}{\unitlength * \real{\svgscale}}%
    \fi%
  \else%
    \setlength{\unitlength}{\svgwidth}%
  \fi%
  \global\let\svgwidth\undefined%
  \global\let\svgscale\undefined%
  \makeatother%
  \begin{picture}(1,0.77945579)%
    \lineheight{1}%
    \setlength\tabcolsep{0pt}%
    \put(0,0){\includegraphics[width=\unitlength,page=1]{fig_L3_2.pdf}}%
    \put(0.83259928,0.57916925){\color[rgb]{0,0,1}\transparent{0}\makebox(0,0)[lt]{\lineheight{1.25}\smash{\begin{tabular}[t]{l}$\delta J$\end{tabular}}}}%
    \put(0.42022349,0.00421693){\color[rgb]{0.15686275,0.04313725,0.04313725}\transparent{0}\makebox(0,0)[lt]{\lineheight{1.25}\smash{\begin{tabular}[t]{l}$\delta U^*$\end{tabular}}}}%
    \put(0.55128482,0.00797461){\color[rgb]{0.15686275,0.04313725,0.04313725}\transparent{0}\makebox(0,0)[lt]{\lineheight{1.25}\smash{\begin{tabular}[t]{l}$\delta \hat{U}^*$\end{tabular}}}}%
    \put(0.0006789,0.2263815){\color[rgb]{0.15686275,0.04313725,0.04313725}\transparent{0}\makebox(0,0)[lt]{\lineheight{1.25}\smash{\begin{tabular}[t]{l}$2\bar{C}_1 \epsilon$\end{tabular}}}}%
    \put(0,0){\includegraphics[width=\unitlength,page=2]{fig_L3_2.pdf}}%
  \end{picture}%
\endgroup%

%% file: figure/base_test.pdf_tex
%% Creator: Inkscape 1.2.1 (9c6d41e410, 2022-07-14), www.inkscape.org
%% Accompanies image file 'base_test.pdf' (pdf, eps, ps)
\begingroup%
  \makeatletter%
  \providecommand\color[2][]{%
    \errmessage{(Inkscape) Color is used for the text in Inkscape, but the package 'color.sty' is not loaded}%
    \renewcommand\color[2][]{}%
  }%
  \providecommand\transparent[1]{%
    \errmessage{(Inkscape) Transparency is used (non-zero) for the text in Inkscape, but the package 'transparent.sty' is not loaded}%
    \renewcommand\transparent[1]{}%
  }%
  \providecommand\rotatebox[2]{#2}%
  \newcommand*\fsize{\dimexpr\f@size pt\relax}%
  \newcommand*\lineheight[1]{\fontsize{\fsize}{#1\fsize}\selectfont}%
  \ifx\svgwidth\undefined%
    \setlength{\unitlength}{382.23364212bp}%
    \ifx\svgscale\undefined%
      \relax%
    \else%
      \setlength{\unitlength}{\unitlength * \real{\svgscale}}%
    \fi%
  \else%
    \setlength{\unitlength}{\svgwidth}%
  \fi%
  \global\let\svgwidth\undefined%
  \global\let\svgscale\undefined%
  \makeatother%
  \begin{picture}(1,1.09754267)%
    \lineheight{1}%
    \setlength\tabcolsep{0pt}%
    \put(0,0){\includegraphics[width=\unitlength,page=1]{base_test.pdf}}%
    \put(0.87541297,0.81381222){\color[rgb]{0,0.49411765,1}\makebox(0,0)[lt]{\lineheight{1.25}\smash{\begin{tabular}[t]{l}$\bar{S}$\end{tabular}}}}%
    \put(0.55022659,0.43756512){\color[rgb]{0,0.49411765,1}\makebox(0,0)[lt]{\lineheight{1.25}\smash{\begin{tabular}[t]{l}$\underline{S}$\end{tabular}}}}%
    \put(0.69186794,0.64028897){\color[rgb]{1,0,0}\makebox(0,0)[lt]{\lineheight{1.25}\smash{\begin{tabular}[t]{l}$S_\infty$\end{tabular}}}}%
  \end{picture}%
\endgroup%

%% file: figure/convergence_s_inf.pdf_tex
%% Creator: Inkscape 1.2.1 (9c6d41e410, 2022-07-14), www.inkscape.org
%% Accompanies image file 'convergence_s_inf.pdf' (pdf, eps, ps)
\begingroup%
  \makeatletter%
  \providecommand\color[2][]{%
    \errmessage{(Inkscape) Color is used for the text in Inkscape, but the package 'color.sty' is not loaded}%
    \renewcommand\color[2][]{}%
  }%
  \providecommand\transparent[1]{%
    \errmessage{(Inkscape) Transparency is used (non-zero) for the text in Inkscape, but the package 'transparent.sty' is not loaded}%
    \renewcommand\transparent[1]{}%
  }%
  \providecommand\rotatebox[2]{#2}%
  \newcommand*\fsize{\dimexpr\f@size pt\relax}%
  \newcommand*\lineheight[1]{\fontsize{\fsize}{#1\fsize}\selectfont}%
  \ifx\svgwidth\undefined%
    \setlength{\unitlength}{425.16834989bp}%
    \ifx\svgscale\undefined%
      \relax%
    \else%
      \setlength{\unitlength}{\unitlength * \real{\svgscale}}%
    \fi%
  \else%
    \setlength{\unitlength}{\svgwidth}%
  \fi%
  \global\let\svgwidth\undefined%
  \global\let\svgscale\undefined%
  \makeatother%
  \begin{picture}(1,0.9867097)%
    \lineheight{1}%
    \setlength\tabcolsep{0pt}%
    \put(0,0){\includegraphics[width=\unitlength,page=1]{convergence_s_inf.pdf}}%
    \put(0.88799412,0.73163115){\color[rgb]{0,0.49411765,1}\makebox(0,0)[lt]{\lineheight{1.25}\smash{\begin{tabular}[t]{l}$\bar{S}$\end{tabular}}}}%
    \put(0.59564599,0.39337856){\color[rgb]{0,0.49411765,1}\makebox(0,0)[lt]{\lineheight{1.25}\smash{\begin{tabular}[t]{l}$\underline{S}$\end{tabular}}}}%
    \put(0.72298399,0.57563077){\color[rgb]{1,0,0}\makebox(0,0)[lt]{\lineheight{1.25}\smash{\begin{tabular}[t]{l}$S_\infty$\end{tabular}}}}%
    \put(0,0){\includegraphics[width=\unitlength,page=2]{convergence_s_inf.pdf}}%
    \put(0.47669775,0.49189398){\color[rgb]{0,0,0}\transparent{0}\makebox(0,0)[lt]{\lineheight{1.25}\smash{\begin{tabular}[t]{l}$J^*$\end{tabular}}}}%
    \put(0,0){\includegraphics[width=\unitlength,page=3]{convergence_s_inf.pdf}}%
  \end{picture}%
\endgroup%

%% file: Simulations_new.tex
\section{Empirical Results}\label{sec5}

% % In this section, we show the performance of the RO-ILQR approach on several nonlinear PDEs and compare its performance with the standard ILQR approach.

\subsection{Optimal Control Problems Analyzed} \label{applications}

This section presents the partial differential equations (PDEs) used to evaluate the performance of the reduced-order ILQR against the standard ILQR method applied to a full-order, high-degree of freedom (DOF) system.

\subsubsection{1D Viscous Burger's Equation} Consider the viscous Burger's Equation given by: \begin{equation*} \frac{\partial u}{\partial t} + u \frac{\partial u}{\partial x} = \nu \frac{\partial^2 u}{\partial x^2}, \end{equation*} with external control inputs at the boundaries corresponding to blowing and suction, defined as $u(0,t)=U_1(t)$ and $u(L,t)=U_2(t)$.

\subsubsection{Material Microstructure Evolution - Allen-Cahn and Cahn-Hilliard Equations} The dynamics of material microstructures can be described using two types of partial differential equations: the Allen-Cahn (A-C) equation \cite{ALLEN19791085}, which models the evolution of a non-conserved quantity, and the Cahn-Hilliard (C-H) equation \cite{Cahn1958}, which represents the evolution of a conserved quantity. The general form of the Allen-Cahn equation is: \begin{align} \frac{\partial \phi}{\partial t} = -M \left( \frac{\partial F}{\partial \phi} - \gamma \nabla^2 \phi \right), \end{align} while the Cahn-Hilliard equation is expressed as: \begin{align} \frac{\partial \phi}{\partial t} = \nabla \cdot M \nabla \left( \frac{\partial F}{\partial \phi} - \gamma \nabla^2 \phi \right). \end{align} In these equations, $\phi = \phi(x,t)$ is the order parameter, which represents the state of the system in control theory and is infinite-dimensional. It indicates the proportion of each phase within the material system. For the two-phase system examined in this study, $\phi = -1$ corresponds to one pure phase, while $\phi = 1$ denotes the other; values of $\phi \in (-1, 1)$ signify a mixed state at the boundary between the two pure phases. The parameter \textit{M} relates to the material's mobility and is assumed to be constant in this analysis. The energy function \textit{F} has a nonlinear dependence on $\phi$, and $\gamma$ is a gradient coefficient that governs the level of diffusion or the thickness of the boundary interface.

In this research, we utilize the following general form of the energy density function: \begin{align} F(\phi;T,h) = \phi^4 + T\phi^2 + h\phi. \end{align} Here, the parameters \textit{T} and \textit{h} are not determined by the A-C and C-H equations; instead, they can be set externally to any value and may vary spatially and temporally.

\subsection{Structure and Task}
We simulated the dynamics in Python, through calling an explicit, second-order solver subroutine in FORTRAN. The system and its tasks are defined as follows: 
\begin{figure}[!htpb]
\vspace{-0.5cm}
\begin{multicols}{3}
    %\hspace{1.4cm} 
      \subfloat[Initial state]{\includegraphics[width=\linewidth]{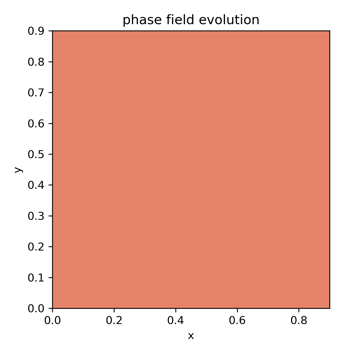}}    
    %   \subfloat[Goal state-I]{\includegraphics[width=\linewidth]{images/Final_state.PNG}}
      \subfloat[Goal state-I]{\includegraphics[width=1.1\linewidth]{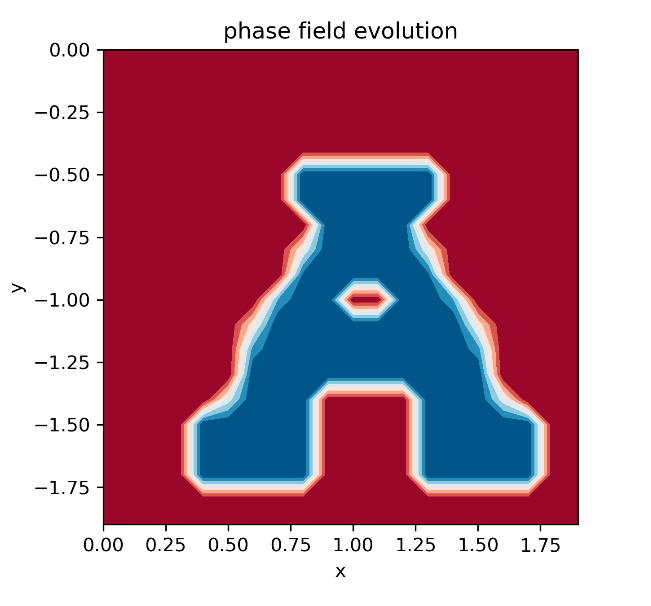}}
      \subfloat[Goal state-II]{\includegraphics[width=\linewidth]{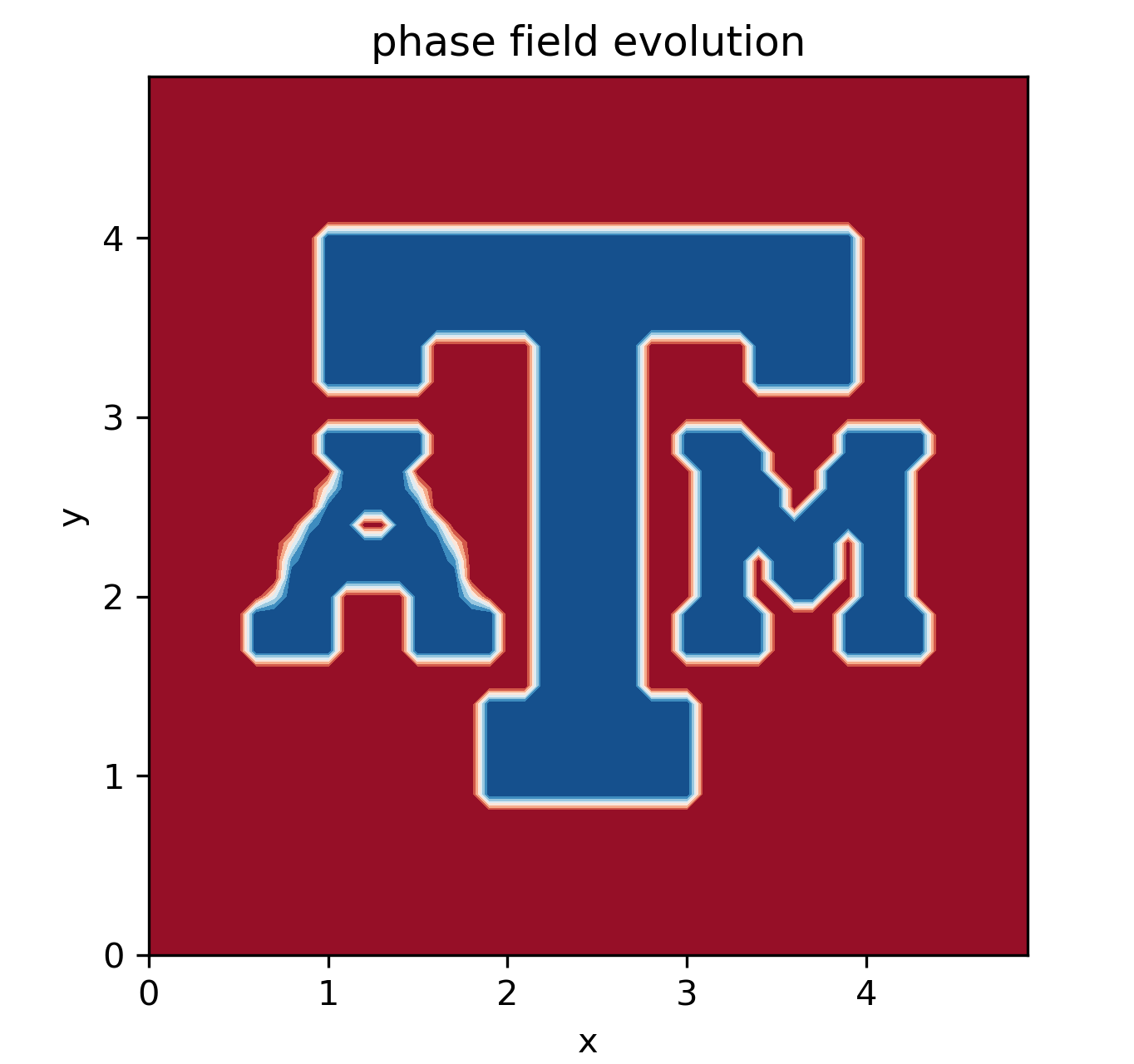}}

\end{multicols}
% \begin{multicols}{2}
%     %\hspace{1.4cm} 
          
%       \subfloat[Goal state-II]{\includegraphics[width=0.5\linewidth]{images/Final_state_2.PNG}}
%       \subfloat[Goal state-III]{\includegraphics[width=0.5\linewidth]{images/Final_state_4.PNG}}

% \end{multicols}
\caption{{Model simulated in Python}}
% \vspace{-0.4cm}
\label{figinit}
\end{figure}

\paragraph{Material Microstructure} The material model employed in this simulation comprises a two-dimensional grid with dimensions of $20 \times 20$ and $50 \times 50$, resulting in 400 and 2500 fully observable states, respectively. The order parameter at each grid point can take values within the range of $[-1, 1]$. The model is solved using an explicit, second-order central-difference scheme. Control inputs $(T, h)$ are applied such that all grid points targeting an order parameter of +1 receive the same control inputs, while a different set of $(T, h)$ is assigned to grid points with a target of -1. The objective is to manipulate the material dynamics to achieve a specified phase distribution, involving 2500 state variables and four control channels.

The initial and desired final states of the model are illustrated in Fig. \ref{figinit}.

\paragraph{1D Viscous Burger's Equation} The Burgers equation is represented as a one-dimensional grid discretized with a finite difference scheme, utilizing 100 equally spaced grid points that correspond to 100 fully observable states. This model is also solved using an explicit, second-order central-difference method. The system incorporates two control inputs for suction and blowing at its boundaries, leading to two control variables.

Starting with an initial velocity profile defined as $u(x,0) = \sin(\pi x)$ for $x \in [-1,1]$, the goal of the control task is to attain the target state $u(x,T) = -0.5$, which represents a constant velocity profile.

\begin{remark}
    Note that, for the purpose of the study, the dimensionality of the systems considered were kept within $O(10^3)$, which is smaller than realistic grid scales. This is done to facilitate a comparison with the full-order solution since the standard, full-order ILQR becomes intractable as we scale the dimensionality of the problem. Thus, in order to establish a benchmark with the ground truth (here, the solution from the ILQR), the experiments were restricted to dimensions where the full-order approach can still yield a feasible solution.
\end{remark}
% as given in Fig.~\ref{bginitfin} (a) and (b). 
% \begin{figure}[!htp]
% %\captionsetup{width=1\linewidth}

% \begin{multicols}{2}
%     % \hspace{3.2cm}
%       \subfloat[Initial Velocity Profile]{\includegraphics[width=\linewidth]{images/Burgers/burgers_init.png}}    
%       \subfloat[Target Velocity Profile]{\includegraphics[width=\linewidth]{images/Burgers/burgers_goal.png}}

% \end{multicols}

% %\begin{multicols}{6}
% %\captionsetup{width=.75\textwidth}
%     \caption{{Initial and goal state for Burgers Equation modeled in Python}} 
% %\end{multicols}

% \label{bginitfin}
% \end{figure}
% The model starts at an initial configuration of all states at $\phi=-1$. The final state should converge to alternating bands of $\phi=0$ (red) and $\phi=1$ (blue), with each band containing 2 columns of grid-points. 

% \subsection{Algorithms Tested}
% % \input{Hybrid.tex}
% %\subsubsection{Model Based Control}
% %% REMOVE THIS LINE FOR JOURNAL
% \paragraph{ILQR} One paragraph introduction to ILQR, add Algorithm 

%% REMOVE THIS LINE FOR JOURNAL
%\input{MaterialControl.tex}
%\input{DDPGAlgorithm.tex}

\subsection{Training and Testing}
\subsubsection{Open-loop trajectory convergence}

The open-loop training plots illustrated in Fig.~\ref{d2c_2_training_testing} depict the cost curve throughout the training process. After the convergence of the cost curves, we acquire the optimal control sequence capable of guiding the systems to fulfill their tasks. The training parameters and results are summarized in (Tables \ref{d2c_comparison_table1}, \ref{params_comp}). The optimal trajectories achieved for the Allen-Cahn and Burgers PDEs are presented in Fig.~\ref{AC_traj} and ~\ref{bg_traj}, respectively.

\begin{figure*}[!htpb]
\centering
% \begin{multicols}{3}
   
%       \subfloat[DDPG Convergence]{\includegraphics[width=1\linewidth]{images/DDPG_Plots/ddpg_banded_h_100.png}}    
%       \subfloat[Gradient-descent based D2C Convergence]{\includegraphics[width=1\linewidth]{images/Training/Banded_episodic_cost_OL_training.png}}
%       \subfloat[ILQR-based D2C Convergence]{\includegraphics[width=1\linewidth]{images/Training/Banded_D2C_2.png}}\\

% \end{multicols}
% \begin{multicols}{3}
   
%       \subfloat[DDPG Convergence]{\includegraphics[width=1\linewidth]{images/DDPG_Plots/ddpg_checkerboard_h_10_10000r.png}}    
%       \subfloat[Gradient-descent based D2C Convergence]{\includegraphics[width=1\linewidth]{images/Training/episodic_cost_OL_training_20x20.png}}
%       %\subfloat[Gradient-descent based D2C Convergence]{\includegraphics[width=1\linewidth]{images/episodic_cost_OL_training.png}}
%       \subfloat[ILQR-based D2C Convergence]{\includegraphics[width=1\linewidth]{images/Training/Figure_20X20.png}}\\

% \end{multicols}
% \centering
\begin{multicols}{3}
\ 
      \subfloat[Allen-Cahn Convergence\\(GS-II)]{\includegraphics[width=1\linewidth]{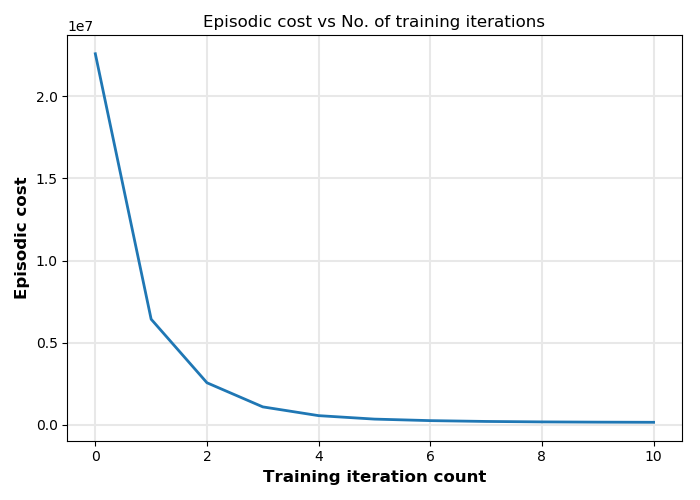}}
      \subfloat[Cahn-Hilliard Convergence\\(GS-I)]{\includegraphics[width=1\linewidth]{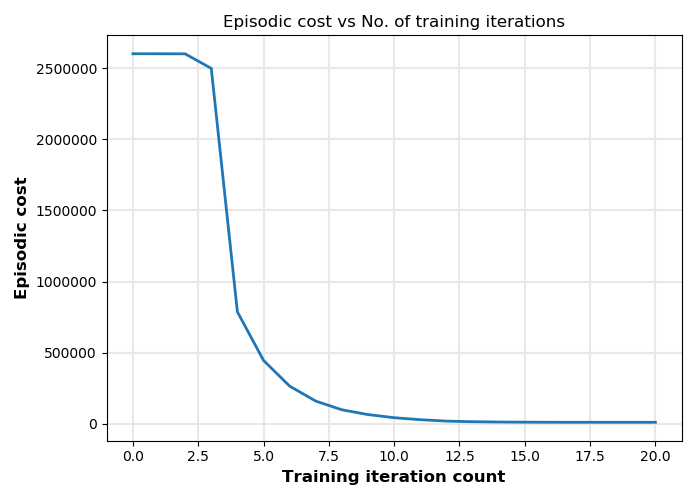}}
      %/DDPG_Plots/ddpg_checkerboard_h_10_10000r.png}}
      \subfloat[Burgers Convergence]{\includegraphics[width=1\linewidth]{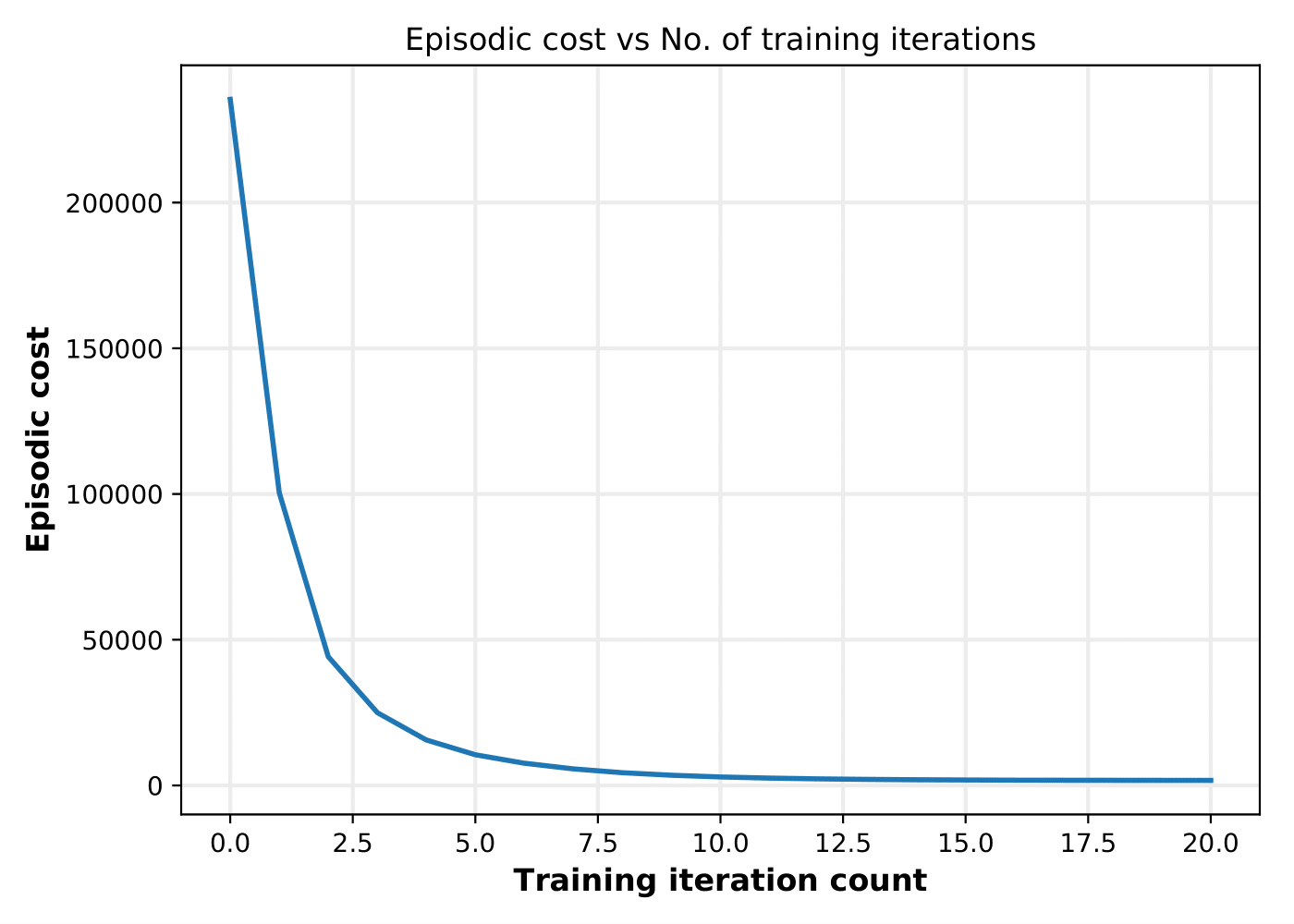}}\\

\end{multicols}

\vspace{-1cm}
\begin{multicols}{3}
   
      \subfloat[Allen-Cahn]{\includegraphics[width=1\linewidth]{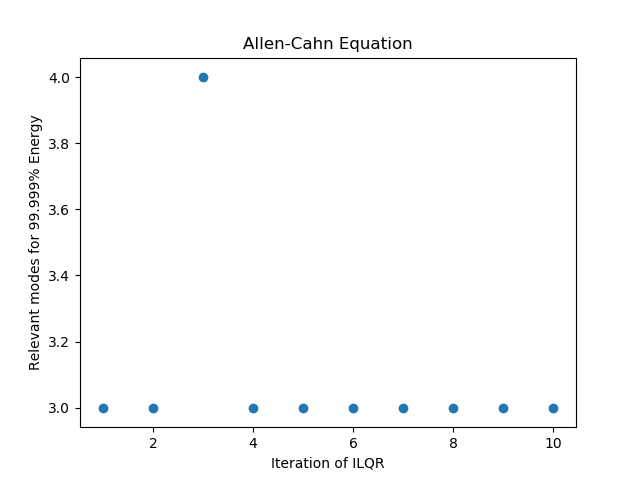}}    
      \subfloat[Cahn-Hilliard]{\includegraphics[width=1\linewidth]{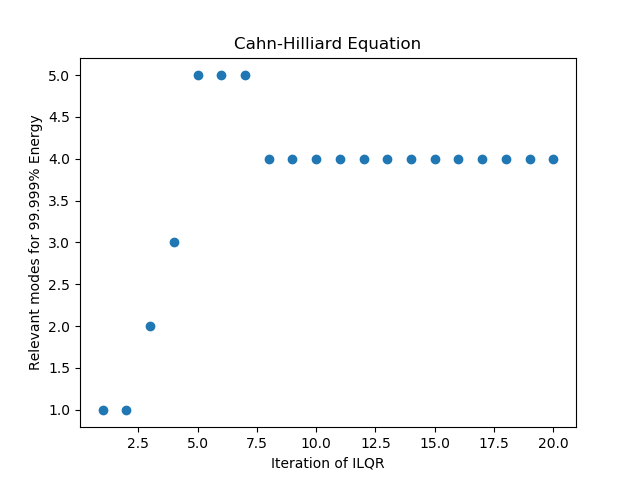}}
      %\subfloat[Gradient-descent based D2C Convergence]{\includegraphics[width=1\linewidth]{images/episodic_cost_OL_training.png}}
      \subfloat[Burgers]{\includegraphics[width=1\linewidth]{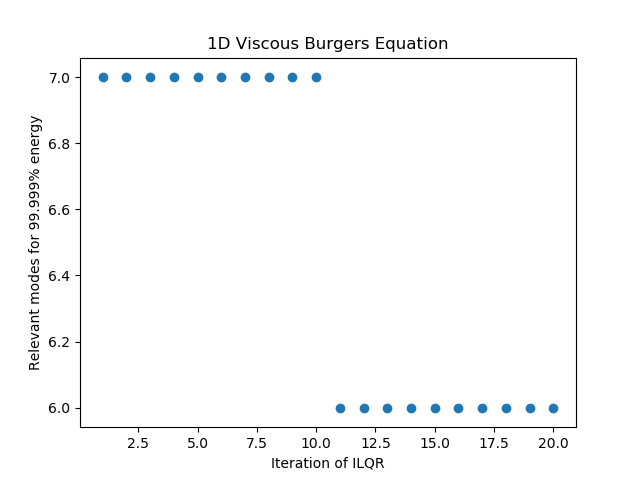}}\\

\end{multicols}

\caption{\small Convergence of Episodic cost for (a) Allen-Cahn, (b) Cahn Hilliard and (c) Viscous Burgers PDE (Top Row), and corresponding variation of number of modes for the 99.999\% energy cutoff with ILQR iteration (Bottom Row)}
% (Note that the `filtered' line in DDPG plots is the moving average over 10 episodes)%Bottom row: Terminal state MSE during testing in D2C vs DDPG. The solid line in the plots indicates the mean and the shade indicates the standard deviation of the corresponding metric.}
\label{d2c_2_training_testing}
% \vspace{-0.8cm}
\end{figure*}

\begin{table}[h]
%\begin{center}
\parbox{\linewidth}{
\caption{Comparison of the training outcomes of ILQR with the reduced order ILQR.}
\centering
\begin{tabular}{|c|c|c|}
\hline
& \multicolumn{2}{|c|}{\bf Training time (in sec.)}\\\cline{2-3}
{\bf PDE system} &{\bf ILQR} &{\bf roILQR}\\
\hline
{\bf Allen-Cahn($50 \times 50$)} &2292.38 &53.21 \\
\hline
%{\bf Goal-II(20X20)}&12327.36& 55.206 &4419.36*\\
{\bf Cahn-Hilliard($20 \times 20$)}&5295.21  &209.83\\
\hline
{\bf Viscous Burgers($100 \times 1$)} &146.99 &13.01\\
\hline

\end{tabular}
\label{d2c_comparison_table1}
\begin{tablenotes}
\item[1] The open-loop training is run on a laptop with a 2-core CPU@2.9GHz and 12GB RAM. No multi-threading at this point.
\end{tablenotes}
%\end{center}
}
\end{table}

\begin{table}[ht]
\vspace{-0.5cm}
\caption{Parameter size comparison for LTV system identification between ILQR and roILQR}
\label{parasize}
\centering
\vspace{0.1in}
\begin{threeparttable}
\setlength{\tabcolsep}{0.6mm}{
\begin{tabular}{|c|c|c|c|c|c|}
\hline
System& No. of  & No. of & No. of  &Dimensionality&Dimensionality\\
&steps&actuators&observed&for&for\\
&&($n_u$)&states($n_x$)& ILQR& roILQR\\
&&&&$ (n_x+n_u)$& $(l+n_u)$\\
\hline
Allen-Cahn& 10&4 &2500&2504 &7\\
\hline
Cahn-Hilliard& 10 &4& 400 &404&9\\
\hline
Burgers & 20&2 & 100&102 &9\\
\hline
\end{tabular}
}
\vspace{0.1in}
\end{threeparttable}
\label{params_comp}
\end{table}

\begin{figure}[!htpb]
\vspace{-0.3cm}
\centering
\begin{multicols}{2}
    \subfloat[Initial state]{\includegraphics[width=\linewidth]{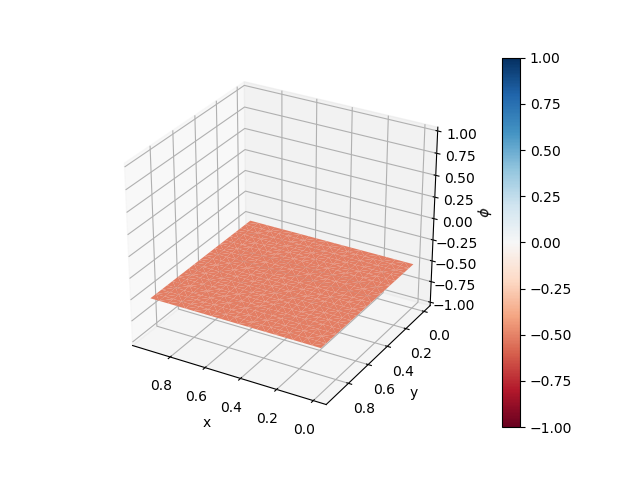}}
    % \subfloat[Intermediate-I)]{\includegraphics[width=\linewidth]{images/roilqr_plots/3d_int_2.png}}
    \subfloat[Intermediate-I]{\includegraphics[width=\linewidth]{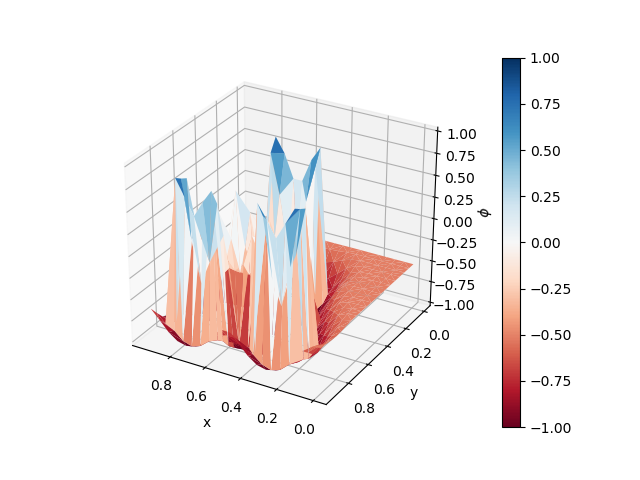}}\\

\end{multicols}
% \vspace{-1cm}

\begin{multicols}{2}
      \subfloat[Intermediate-II]{\includegraphics[width=\linewidth]{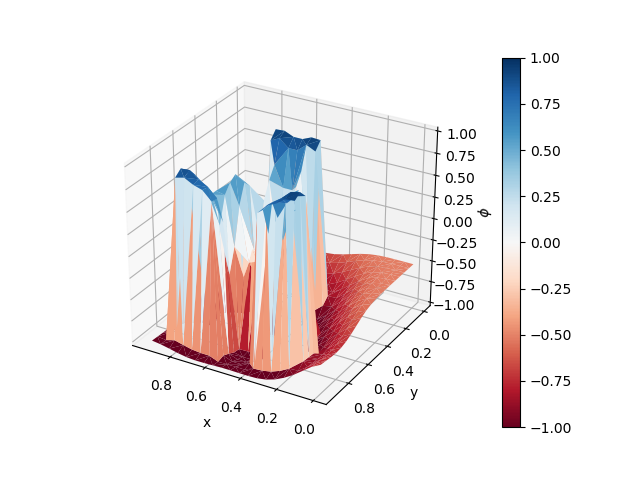}}
    %   \subfloat[Intermediate-IV]{\includegraphics[width=\linewidth]{images/roilqr_plots/3d_1_8.png}}
      \subfloat[Final State]{\includegraphics[width=\linewidth]{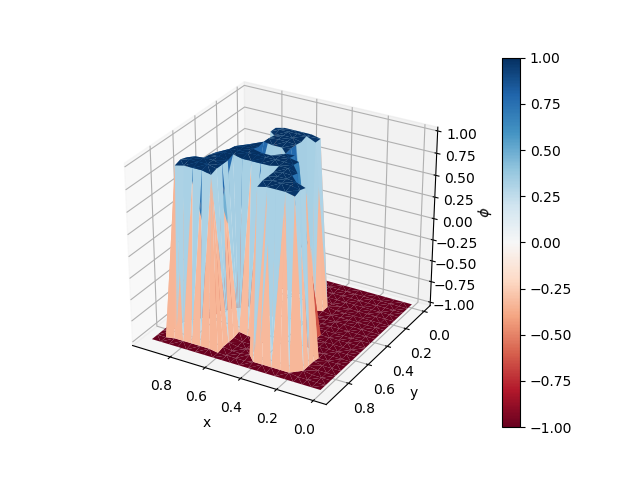}}

\end{multicols}

% \begin{multicols}{2}
   
%     %   \subfloat[Intermediate-III]{\includegraphics[width=\linewidth]{images/roilqr plots/3d_0_5.png}}
%       \subfloat[Intermediate-II]{\includegraphics[width=\linewidth]{images/roilqr plots/3d_1_2.png}}
%       \subfloat[Final State]{\includegraphics[width=\linewidth]{images/roilqr plots/3d_2_6.png}}\\

% \end{multicols}
      
\caption{Optimal trajectory for the Cahn-Hilliard PDE.}
% (Note that the `filtered' line in DDPG plots is the moving average over 10 episodes)%Bottom row: Terminal state MSE during testing in D2C vs DDPG. The solid line in the plots indicates the mean and the shade indicates the standard deviation of the corresponding metric.}
% \vspace{-0.5cm}
\label{AC_traj}
\end{figure}

\begin{figure}
% \vspace{-0.5cm}
% \begin{multicols}{1}
   
\includegraphics[width=1\linewidth]{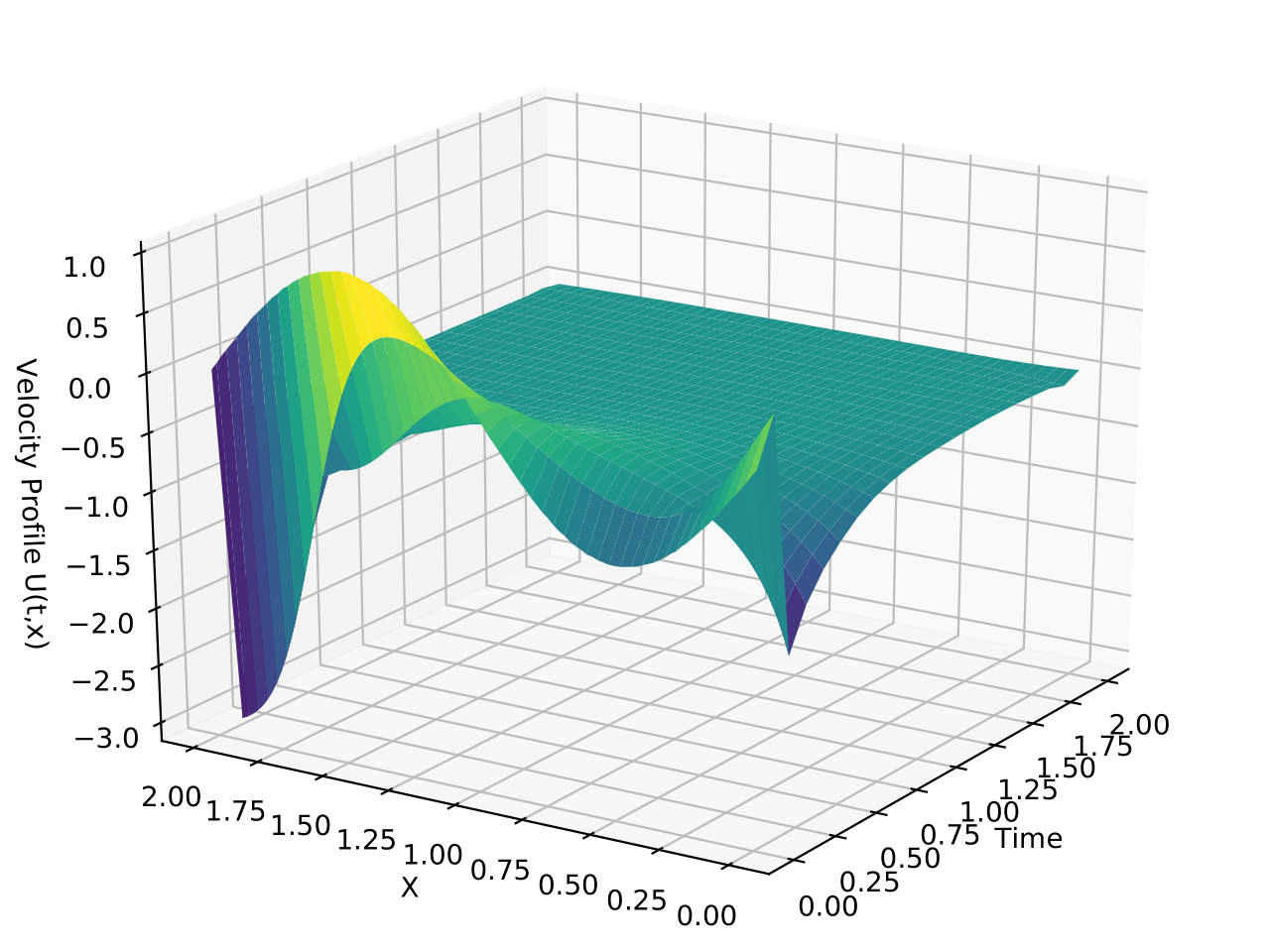}

% \end{multicols}

\caption{\small Optimal trajectory for the Burgers PDE.}
\label{bg_traj}
\end{figure}

\subsubsection{Repeatability/Variance in Training}
The graphs presented in Fig. \ref{d2c_2_conv} illustrate the average and standard deviation of the cost curves based on various initial guesses, computed over 100 iterations. This analysis reveals that the algorithm consistently converges to the same minima.

\subsubsection{Quality of the solution}    
% \textcolor{red}{Aren't the solutions the same now for both now?}
The plots in Fig. \ref{d2c_2_comp} compare the iterations of the full-scale ILQR vs the proposed reduced-order modification. 
We observe that the RO-ILQR converges to a solution with a cost within $14\%$ of the true optimal. This is in agreement with the convergence analysis we develop in Section \ref{analysis}, wherein we are guaranteed to converge to the set $S_\infty$. It is observed that this convergence bound is, in practice, tighter than the conservative estimate provided by the mathematical analysis. Moreover, the RO-ILQR approach is far cheaper computationally when compared to the full order ILQR. However, we also note that in the case of Burger's equation, the reduced order solution is almost of the same quality of the full order solution.
% We observe that both methods converge to the same true optimal solution. Furthermore, the RO-iLQR approach is far cheaper computationally when compared to the full order iLQR. 

\begin{figure}[!htpb]
\vspace{0.0cm}
\centering

\begin{multicols}{2}
   
      \subfloat[Allen-Cahn (GS-II)]{\includegraphics[width=1\linewidth]{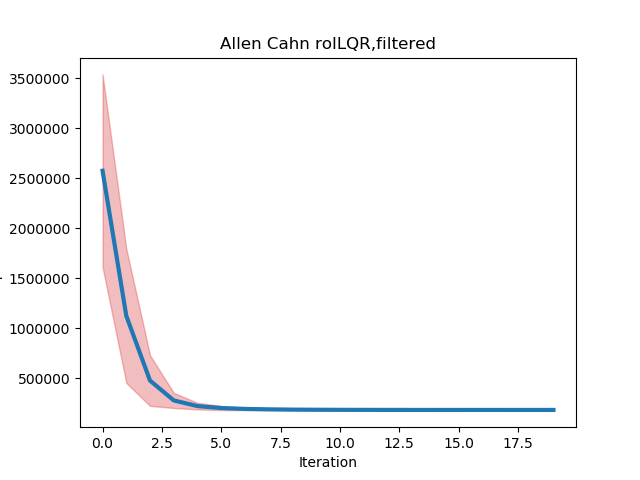}}
      \subfloat[Cahn-Hilliard (GS-I)]{\includegraphics[width=1\linewidth]{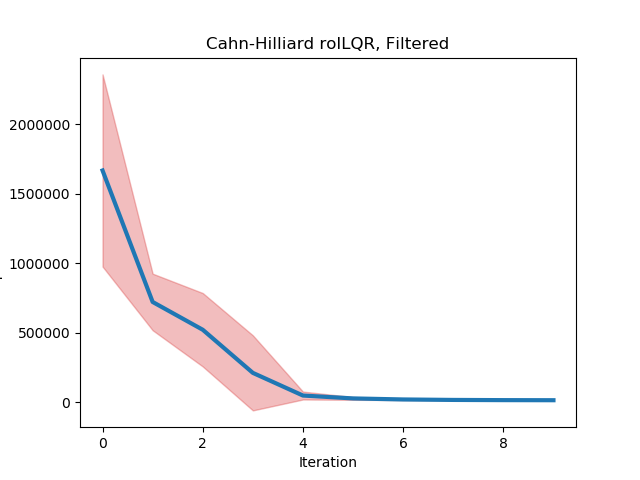}}
      %/DDPG_Plots/ddpg_checkerboard_h_10_10000r.png}}
    %   \subfloat[Burgers Convergence\\(GS-III)]{\includegraphics[width=1\linewidth]{images/roilqr plots/CH_roilqr_cost2.png}}\\%FILLER

\end{multicols}
      
\caption{\small Convergence of Episodic cost for (a) Allen-Cahn, (b) Cahn Hilliard and (c) Viscous Burgers PDE averaged over 100 runs with different initial guesses}
\label{d2c_2_conv}
% \vspace{-0.05cm}
\end{figure}

\begin{figure}[!htpb]
\setlength{\belowcaptionskip}{-15pt}
% \vspace{-0.5cm}
\centering

\begin{multicols}{2}
   
      \subfloat[Allen-Cahn Comparison\\(GS-II)]{\includegraphics[width=1\linewidth]{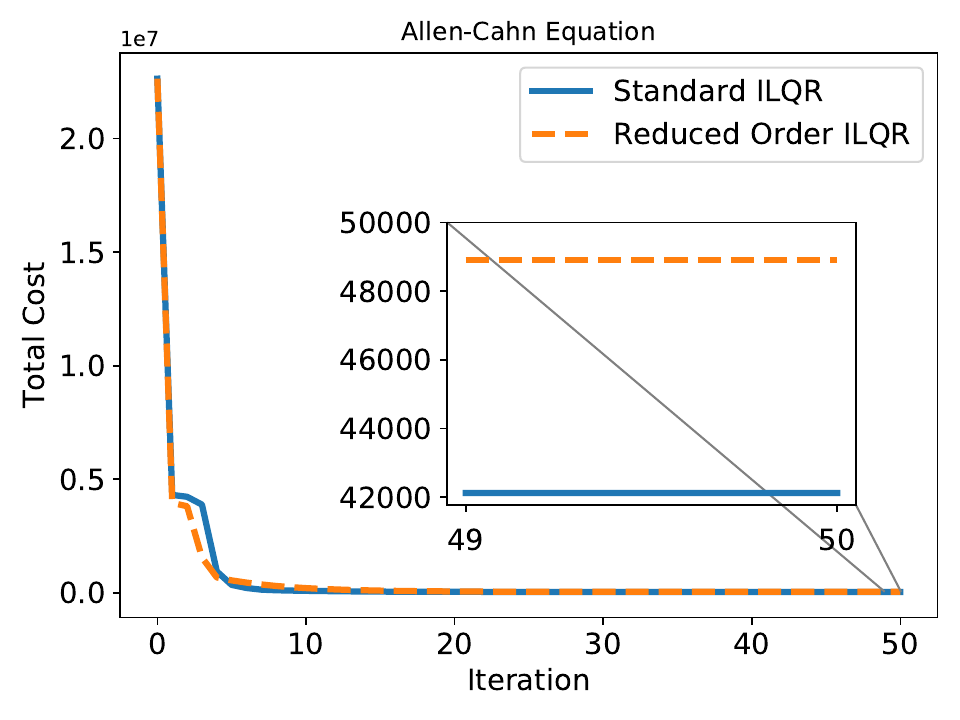}}
    %   \subfloat[Cahn-Hilliard Comparison\\(GS-II)]{\includegraphics[width=1\linewidth]{images/roilqr plots/burgers_roilqr_comp.png}}%FILLER
    %   %/DDPG_Plots/ddpg_checkerboard_h_10_10000r.png}}
      \subfloat[Cahn-Hilliard Comparison\\(GS-I)]%\\(GS-III)]
      {\includegraphics[width=1\linewidth]{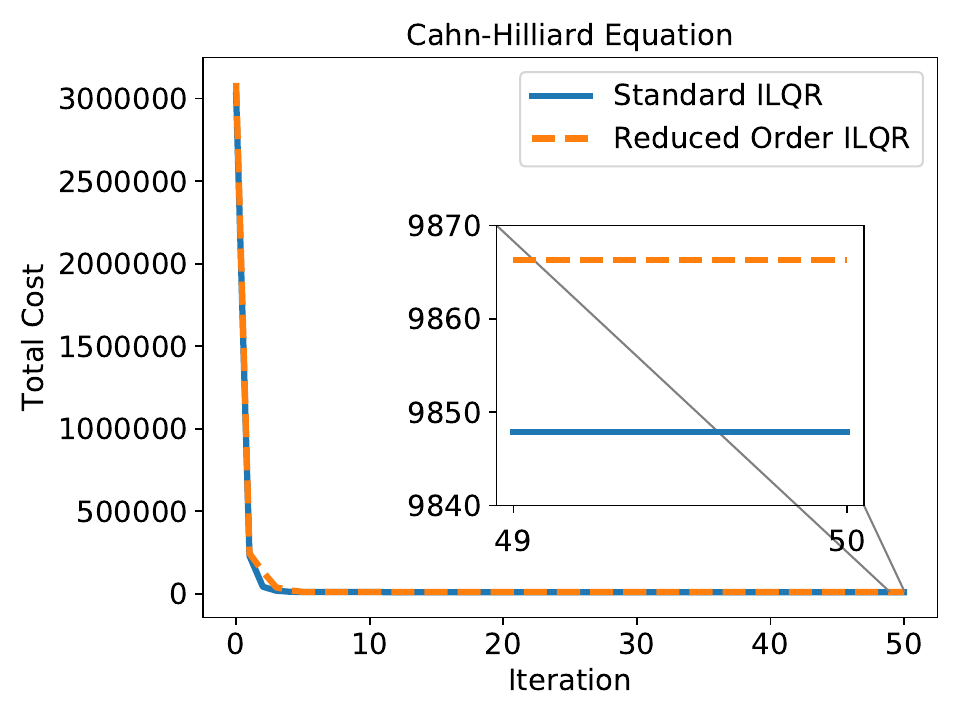}}\\%FILLER

\end{multicols}

\subfloat[Burgers Comparison]{\includegraphics[width=0.5\linewidth]{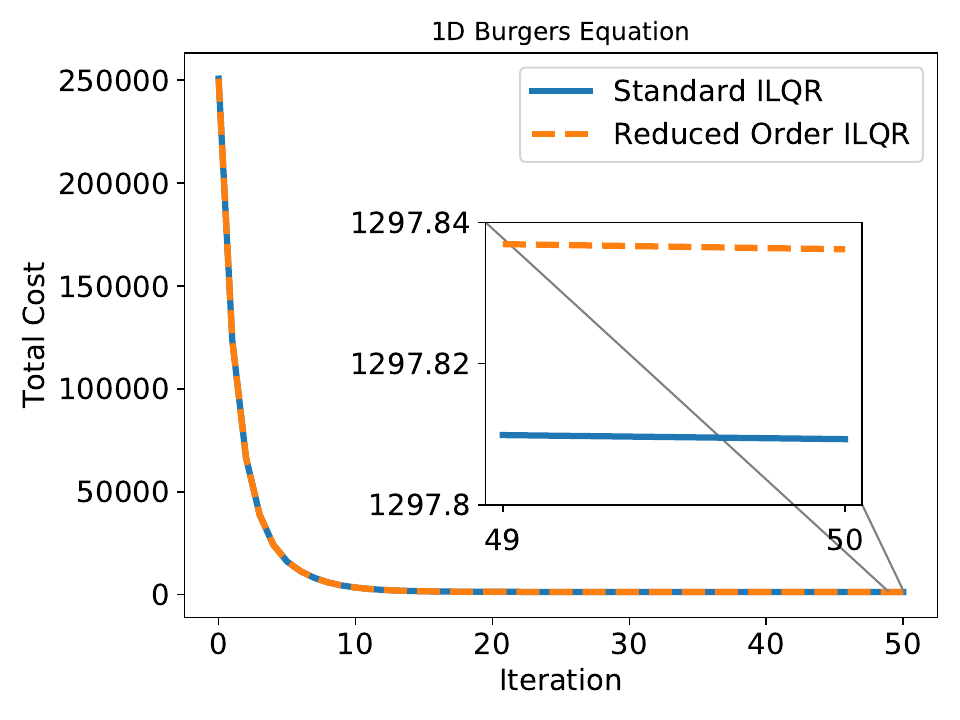}}

\caption{\small Comparing convergence of Episodic cost for (a) Allen-Cahn and (b) Viscous Burgers PDE, starting from same initial guess}
\label{d2c_2_comp}
\end{figure}

% \begin{figure}
%     \centering
%     \import{figure/}{fig_L1.pdf_tex}
%     \caption{}
%     \label{fig}
% \end{figure}

\subsubsection{Comparison With Deep RL}
In order to benchmark our algorithm, we have previously compared a data-based implementation of the standard, full order ILQR algorithm with the Deep Deterministic Policy Gradient (DDPG) algorithm, for the Allen-Cahn and Cahn-Hilliard Equations discussed in Sec.~\ref{applications} \cite{material_control_journal}. It was observed that DDPG fails to converge to the target microstructure for the scale of the problems considered here, due to the overparametrization that stems from the increasing dimensionality of the resulting system. Thus, using DDPG or similar techniques to design a feedback law is infeasible in the context of such problems. Moreover, the structure of the DDPG algorithm, or similar model-free RL techniques in general, does not allow for the implementation of model order techniques that we implement in RO-ILQR, making the comparison unfair.

%% file: Conclusions.tex
\section{CONCLUSIONS}
In this article, we have proposed an efficient approach for controlling nonlinear partial differential equations by employing model order reduction in tandem with a well known nonlinear optimal control algorithm ILQR. We have shown the viability and repeatability of the approach, and have compared it with a standard ILQR without model reduction to show that the proposed approach is highly efficient without sacrificing performance. We have also theoretically characterized the quality of the solution achieved here, and the sub-optimality thereof, and empirically shown how the convergence bounds are tighter practically than the conservative estimates.
In future work, we plan to apply the methodology to larger scale models, and scale the approach to real-world systems. Another future direction of work is the implications of the POD based reduced order LTV identification introduced here for nonlinear model reduction problems.

% In this article, we have provided an overview of the modeling of multi-phase micro-structures in accordance with the Allen-Cahn and Cahn-Hilliard equations. We have also compared learning/ data-based approaches to the control of such structures, outlining their relative merits and demerits. In future work, we shall further develop a hybrid model-data based method which can exploit a priori knowledge of the dynamics to significantly reduce computation time. We also plan to extend the methodology for higher-dimensional materials, leading up to an implementation on a realistic full-scale model.

%% file: Appendix.tex
% \newpage

\section{Appendix}

\subsection{Proof for Lemma 1}
\label{l1proof}

% \begin{proof}
    Given a perturbed control sequence $\delta U=\{\delta u_0, \delta u_1,..., \delta u_{T-1}\}$ about the nominal, let us apply the same input to both full-order and reduced-order perturbed dynamics. Thus,
    \begin{align*}
        &\delta J(\delta U)-\delta \hat{J}(\delta U) = \sum_{t=0}^{T-1} \Bigg(\Bar{x}_t^T Q_t \delta x_t + \delta x_t^T Q_t \delta x_t \\
        &~~~~-   (\Phi \Bar{\alpha}_t)^T Q_t (\Phi \delta \alpha_t)-(\Phi \delta \alpha_t)^T Q_t (\Phi \delta \alpha_t) \Bigg) + \\
        &~~~~\Bigg(\Bar{x}_T^T Q_T \delta x_T + \delta x_T^T Q_T \delta x_T\\
        &~~~~-(\Phi \Bar{\alpha}_T)^T Q_T (\Phi \delta \alpha_T)-(\Phi \delta \alpha_T)^T Q_T (\Phi \delta \alpha_T) \Bigg) 
    \end{align*}

    Adding and subtracting $\bar{x}_t^T Q_t (\Phi \delta \alpha_t)$, $\delta x_t^T Q_t (\Phi \delta \alpha_t)$, $\bar{x}_T^T Q_T (\Phi \delta \alpha_T)$ and $\delta x_T^T Q_T (\Phi \delta \alpha_T)$, we can rewrite the difference as $\delta J(\delta U)-\delta \hat{J}(\delta U) = \sum_{t=0}^{T-1} \Bigg(\Bar{x}_t^T Q_t (\delta x_t - \Phi \delta \alpha_t) + \delta x_t^T Q_t (\delta x_t-\Phi \delta \alpha_t)-   (\Bar{x}_t - \Phi \Bar{\alpha}_t)^T Q_t (\Phi \delta \alpha_t)-(\delta x_t-\Phi \delta \alpha_t)^T Q_t (\Phi \delta \alpha_t) \Bigg) +\Bigg(\Bar{x}_T^T Q_t (\delta x_T - \Phi \delta \alpha_T) + \delta x_T^T Q_T (\delta x_T-\Phi \delta \alpha_T)-   (\Bar{x}_T - \Phi \Bar{\alpha}_T)^T Q_T (\Phi \delta \alpha_T)-
        (\delta x_T-\Phi \delta \alpha_T)^T Q_T (\Phi \delta \alpha_T) \Bigg)$. Thus, 
    \begin{align*}
        % &\delta J(\delta U)-\delta \hat{J}(\delta U) = \sum_{t=0}^{T-1} \Bigg(\Bar{x}_t^T Q_t (\delta x_t - \Phi \delta \alpha_t) +\\
        % &~~~~~~ \delta x_t^T Q_t (\delta x_t-\Phi \delta \alpha_t)-   (\Bar{x}_t - \Phi \Bar{\alpha}_t)^T Q_t (\Phi \delta \alpha_t)\\
        % &~~~~~~-(\delta x_t-\Phi \delta \alpha_t)^T Q_t (\Phi \delta \alpha_t) \Bigg) +\\ &~~~~~~\Bigg(\Bar{x}_T^T Q_t (\delta x_T - \Phi \delta \alpha_T) + \delta x_T^T Q_T (\delta x_T-\Phi \delta \alpha_T)\\
        % &~~~~~~-   (\Bar{x}_T - \Phi \Bar{\alpha}_T)^T Q_T (\Phi \delta \alpha_T)\\
        % &~~~~~~-
        % (\delta x_T-\Phi \delta \alpha_T)^T Q_T (\Phi \delta \alpha_T) \Bigg) \\
        &\delta J(\delta U)-\delta \hat{J}(\delta U)\leq \sum_{t=0}^{T-1} \Bigg( (||Q_t \Bar{x}_t||+||Q_t \delta x_t||)||\delta x_t-\Phi \delta \alpha_t|| \\
        &~~~~+ ||Q_t \Phi \delta \alpha_t||(||\Bar{x}_t-\Phi \alpha_t||+||\delta x_t-\Phi \delta \alpha_t||) \Bigg) \\
        &~~~~+ \Bigg( (||Q_T \Bar{x}_T||+||Q_T \delta x_T||)||\delta x_T-\Phi \delta \alpha_T|| \\ &~~~~+ ||Q_T \Phi \delta \alpha_t||(||\Bar{x}_T-\Phi \Bar{\alpha}_T||+||\delta x_T-\Phi \delta \alpha_T||) \Bigg)\\
        &~~~~~~~~~~~~~~~~~~~~~~~\leq 7(T+1) \Bar{C}\epsilon=\Bar{C}_1 \epsilon.
    \end{align*}
    
% \end{proof}

\subsection{Proof for Lemma 2}
\label{l2proof}
% \begin{proof}
    From Lemma \ref{L1}, we know that
    \begin{align*}
        &|\delta J(\delta U)-\delta \hat{J}(\delta U)|\leq \Bar{C}_1 \epsilon\\ &\implies -\Bar{C}_1 \epsilon \leq \delta J(\delta U)-\delta \hat{J}(\delta U) \leq \Bar{C}_1 \epsilon.
    \end{align*}
    Also, we know that $\delta U^*$ is a sub-optimal policy for $\delta \hat{J}(\cdot)$, while $\delta \hat{U}^*$ is the optimal. Similarly, $\delta \hat{U}^*$ is a sub-optimal policy for $\delta J(\cdot)$, while $\delta U^*$ is the optimal. Hence,
    \begin{align}
        &\delta \hat{J}(\delta \hat{U}^*)\leq\delta \hat{J}(\delta U^*) \implies \delta \hat{J}(\delta U^*)-\delta \hat{J}(\delta \hat{U}^*) \geq 0, \label{subopt-a}\\
        &\delta J(\delta U^*)\leq\delta J(\delta \hat{U}^*) \implies \delta J(\delta U^*)-\delta J(\delta \hat{U}^*) \leq 0. \label{subopt-b}
    \end{align}

    From Eqs.~\ref{subopt-a}, \ref{subopt-b} and Lemma \ref{L1}, we can write
    \begin{align*}
        &\delta J(\delta U^*)-\delta \hat{J}(\delta \hat{U}^*) = (\delta J(\delta U^*)-\delta \hat{J}(\delta U^*) )+
        \\ &~~~~~~~~~~~~~~~~~~~~~~~~~~~(\delta \hat{J}(\delta U^*)-\delta \hat{J}(\delta \hat{U}^*)) \geq -\Bar{C}_1 \epsilon, \text{ and} \\
        &\delta J(\delta U^*)-\delta \hat{J}(\delta \hat{U}^*) = (\delta J(\delta U^*)-\delta J(\delta \hat{U}^*) )+\\ &~~~~~~~~~~~~~~~~~~~~~~~~~~~(\delta J(\delta \hat{U}^*)-\delta \hat{J}(\delta \hat{U}^*)) \leq \Bar{C}_1 \epsilon.
    \end{align*}

    Thus, 
    \begin{align*}
        |\delta J(\delta U^*)-\delta \hat{J}(\delta \hat{U}^*)|\leq \Bar{C}_1 \epsilon.
    \end{align*}
% \end{proof}

\subsection{Proof for Lemma 3}
\label{l3proof}
% \begin{proof}
    From Lemma \ref{L2}, 
    \begin{align}
        &|\delta J(\delta U^*)-\delta \hat{J}(\delta \hat{U}^*)|\leq \Bar{C}_1 \epsilon, \nonumber\\
        &\implies \Bar{C}_1 \epsilon \geq |\delta J(\delta U^*)-\delta J(\delta \hat{U}^*)+\delta J(\delta \hat{U}^*)-\delta \hat{J}(\delta \hat{U}^*)| \nonumber \\ &~~~~~~~~~~~~\geq |\delta J(\delta U^*)-\delta J(\delta \hat{U}^*)|-|\delta J(\delta \hat{U}^*)-\delta \hat{J}(\delta \hat{U}^*)|, \nonumber\\
        &\implies |\delta J(\delta U^*)-\delta J(\delta \hat{U}^*)| \leq 2\Bar{C}_1\epsilon. \label{eq:sameJ}
    \end{align}
    
    Let $\delta \hat{U}^* = \delta U^* + \delta U^R$. From Eq.~\ref{eq:sameJ}, 
    \begin{align*}
        |\delta J(\delta U^*)-\delta J(\delta U^* + \delta U^R)| \leq 2\Bar{C}_1\epsilon.
    \end{align*}

     Since the LTV dynamics are affine in control, $\delta J$ is quadratic in the input $\delta U$. Thus, we can write $|\delta J(\delta U^*)-\delta J(\delta U^* + \delta U^R)|=|\delta J(\delta U^*)-\delta J(\delta U^*)-\frac{\partial \delta J}{\partial \delta U}\Bigg|_{\delta U^*} \delta U^R - (\delta U^R)^T \frac{\partial^2 \delta J}{\partial (\delta U) ^2}\Bigg|_{\delta U^*} \delta U^R|$.
    % \begin{align*}
    %     |\delta J(\delta U^*)-\delta J(\delta U^* + \delta U^R)|=|\delta J(\delta U^*)-\delta J(\delta U^*)-\frac{\partial \delta J}{\partial \delta U}\Bigg|_{\delta U^*} \delta U^R - (\delta U^R)^T \frac{\partial^2 \delta J}{\partial (\delta U) ^2}\Bigg|_{\delta U^*} \delta U^R|.
    % \end{align*}
% \textcolor{red}{Why the $O(.)$ term? Isn't $\delta J$ exactly quadratic?}\\
    Since $\delta U^*$ is the minima of $\delta J$, $\frac{\partial \delta J}{\partial \delta u}\Bigg|_{\delta u^*}=0$. Let us define the weighted-norm $||.||_Q$ as $||x||_Q = (x^T Q x)^{1/2}$. Thus, 
    % \textcolor{red}{See above comment}. Thus,
    
    \begin{align*}
       &|\delta J(\delta U^*)-\delta J(\delta \hat{U}^*)|=||\delta U^R||_{\frac{\partial^2 \delta J}{\partial (\delta U) ^2}|_{\delta U^*}}^2 \leq 2 \Bar{C}_1 \epsilon.
    \end{align*}
    % \textcolor{red}{Define the $||.||_Q$ notation.}
     
    Now, $||\frac{\partial^2 \delta J}{\partial (\delta U) ^2}|_{\delta U^*}|| \geq \sigma_{min}(\frac{\partial^2 \delta J}{\partial (\delta U) ^2}|_{\delta U^*})$, where $\sigma_{min}(\cdot)$ is the smallest singular value of the function. Then,

    \begin{align*}
         &||\delta U^R||^2 \leq \frac{2\Bar{C}_1 \epsilon}{\sigma_{min}(\frac{\partial^2 \delta J}{\partial (\delta U) ^2}|_{\delta U^*})} \\ 
       &\implies ||\delta U^R||=||\delta U^*-\delta \hat{U}^*||\leq \sqrt{\frac{2\Bar{C}_1 \epsilon}{\sigma_{min}(\frac{\partial^2 \delta J}{\partial (\delta U) ^2}|_{\delta U^*})}}
    \end{align*}
    From A\ref{A4}, $\sigma_{min}(\frac{\partial^2 \delta J}{\partial (\delta U) ^2}|_{\delta U^*})\geq \Bar{\sigma}$. Thus,
    
    \begin{align*}
        ||\delta U^*-\delta \hat{U}^*||\leq \sqrt{\frac{2\Bar{C}_1 \epsilon}{\sigma_{min}(\frac{\partial^2 \delta J}{\partial (\delta U) ^2}|_{\delta U^*})}}\leq \sqrt{\frac{2\Bar{C}_1 \epsilon}{\Bar{\sigma}}}=\delta 
    \end{align*}
% \end{proof}

\subsection{Proof for Lemma 4}
\label{l4proof}
% \begin{proof}
    % \textcolor{red}{Isn't LQR Newton?.}
    Let the cost at some iteration $k$ be $J(U^k)$. From the line-search algorithm, we know that
\begin{align*}
    \frac{J(U^{k+1})}{\Delta J(\eta)}\geq \rho >0,
\end{align*}
where $\Delta J(\eta)=-\eta S_t^T \hat{S}_t$. 
Since the LQR optimization proceeds as a Newton method, the descent step for the \ref{eq:FO-LQR} can be written as 
\begin{align*}
    S_t = -(\nabla^2 J)^{-1} \nabla J.
\end{align*}

Since the model order reduction in the \ref{eq:RO-LQR} introduces errors in the gradient, the corresponding descent step is 
\begin{align*}
    \hat{S}_t = -(\nabla^2 \hat{J})^{-1}\nabla \hat{J} = -[(\nabla^2 J)^{-1} \nabla J + w],
\end{align*}
where $||w||\leq \delta$ due to Lemma \ref{L3}.

For the development below, let us write $J(U^k)$ as $J^k$ for ease of notation. Then,
\begin{align*}
    J^{k+1}-J^k &\leq -\eta \rho ([\nabla^2 J]^{-1})\nabla J)^T ([\nabla^2 J]^{-1} \nabla J + w) \\
    &\leq -\eta \rho ||[\nabla^2 J]^{-1}\nabla J||(||[\nabla^2 J]^{-1}\nabla J||-||w||).
\end{align*}

Now, we know that $\eta>0$, $\rho>0$ and $||[\nabla^2 J]^{-1}\nabla J||>0$ for all points besides the stationary point. Then,
\begin{equation}
    J^{k+1}-J^k \leq - \beta(||[\nabla^2 J]^{-1}\nabla J||-||w||). \label{eq:monotonic}
\end{equation}
For the algorithm to descend, $J^{k+1}-J^k <0$. Thus, 

\begin{equation*}
    ||[\nabla^2 J]^{-1}\nabla J||>||w||.
\end{equation*}
Thus, $\Hat{S}_t$ is guaranteed to descend if
\begin{equation*}
    ||[\nabla^2 J]^{-1}\nabla J||>\delta.
\end{equation*}

% Hence, from Lemma \ref{L4} and Eq.~\ref{eq:monotonic}, $J^{k+1}-J^k<0$ for any $U\notin S_\infty$, \textit{i.e.}, the cost decreases monotonically. 
% % \textcolor{red}{Isn't the above what you proved in Lemma 4? Also, check for typos and missing references.}

%     Given that the LQR optimization proceeds as a Newton method, the descent step for the \ref{eq:FO-LQR} can be written as 
% \begin{align*}
%     S_t = -(\nabla^2 J)^{-1} \nabla J.
% \end{align*}

% Since the model order reduction in the \ref{eq:RO-LQR} introduces errors in the gradient, the corresponding descent step is 
% \begin{align*}
%     \hat{S}_t = -(\nabla^2 \hat{J})^{-1}\nabla \hat{J} = -[(\nabla^2 J)^{-1} \nabla J + w],
% \end{align*}
% where $||w||\leq \delta$ due to Lemma 3. For $\hat{S}_t$ to be a descent direction of the original problem,
% \begin{align*}
%     &([\nabla^2 J]^{-1})\nabla J)^T \hat{S}_t <0,\\
%     \implies &-([\nabla^2 J]^{-1})\nabla J)^T ([\nabla^2 J]^{-1} \nabla J + w)= \\&-||[\nabla^2 J]^{-1} \nabla J||^2-([\nabla^2 J]^{-1} \nabla J)^T w <0.
% \end{align*}

% Now, $([\nabla^2 J]^{-1} \nabla J)^T w \geq -||[\nabla^2 J]^{-1}\nabla J||.||w||$. Thus,
% \begin{align*}
%     -||[\nabla^2 J]^{-1}\nabla J||(||[\nabla^2 J]^{-1}\nabla J||-||w||)<0.
% \end{align*}
% Since $||[\nabla^2 J]^{-1}\nabla J||>0$ everywhere except the stationary point,  
% \begin{equation*}
%     ||[\nabla^2 J]^{-1}\nabla J||>||w||.
% \end{equation*}
% Thus, $\Hat{S}_t$ is guaranteed to descend if
% \begin{equation*}
%     ||[\nabla^2 J]^{-1}\nabla J||>\delta.
% \end{equation*}
% \end{proof}

%% file: Bibliography.tex
\bibliographystyle{IEEEtran}
\bibliography{IEEEabrv,ICRA_refs,TAC_refs}

% \begin{thebibliography}{}
% \setlength{\itemindent}{-\leftmargin}
% %\makeatletter\renewcommand{\@biblabel}{}\makeatother
% % approximate DP
% @book{bensoussan,
%   title={Representation and control of infinite dimensional systems},
%   author={Bensoussan, Alain and Da Prato, Giuseppe and Delfour, Michel C and Mitter, Sanjoy K},
%   volume={1},
%   year={1992},
%   publisher={Birkh{\"a}user Boston}
% }

% \end{thebibliography}